\documentclass[a4paper,onecolumn,11pt,accepted=2025-05-21]{quantumarticle}
\pdfoutput=1
\usepackage[utf8]{inputenc}
\usepackage[english]{babel}
\usepackage[T1]{fontenc}
\usepackage{mathtools}
\usepackage{amssymb}
\usepackage{amsmath}
\usepackage{amsthm}
\usepackage{physics}
\usepackage{tikz}
\usepackage{lipsum}
\usepackage[numbers]{natbib}
\usepackage{enumitem}

\usepackage{thmtools}
\usepackage{thm-restate}
\declaretheorem[name=Theorem]{theorem}
\declaretheorem[name=Lemma]{lemma}
\declaretheorem[name=Corollary]{corollary}
\declaretheorem[name=Definition]{definition}
\declaretheorem[name=Proposition]{proposition}

\declaretheorem[name=Fact]{fact}
\declaretheorem[name=Remark]{remark}

\newcommand{\mO}{\mathcal{O}}

\newcommand{\Pclass}{\mathsf{P}}
\newcommand{\NP}{\mathsf{NP}}
\newcommand{\PSPACE}{\mathsf{PSPACE}}
\newcommand{\NEXP}{\mathsf{NEXP}}
\newcommand{\BQP}{\mathsf{BQP}}
\newcommand{\QMA}{\mathsf{QMA}}
\newcommand{\QCMA}{\mathsf{QCMA}}
\newcommand{\QCPCP}{\mathsf{QCPCP}}
\newcommand{\QPCP}{\mathsf{QPCP}}

\newcommand{\poly}{\mathrm{poly}}
\newcommand{\polylog}{\mathrm{polylog}}
\newcommand{\CLDM}{\textup{CLDM}}

\newtheorem{innercustomthm}{Theorem}
\newenvironment{customthm}[1]
  {\renewcommand\theinnercustomthm{#1}\innercustomthm}
  {\endinnercustomthm}

\newtheorem{innercustomlem}{Lemma}
\newenvironment{customlem}[1]
  {\renewcommand\theinnercustomlem{#1}\innercustomlem}
  {\endinnercustomlem}

\newtheorem{subdefinition}{Definition}[definition]

\usepackage[pagebackref]{hyperref}
\usepackage[nameinlink,capitalize]{cleveref}

\crefname{claim}{claim}{claims}
\crefname{claim}{Claim}{Claims}

\newcounter{protocolcounter}
\renewcommand{\theprotocolcounter}{\arabic{protocolcounter}}

\crefname{protocolcounter}{Protocol}{Protocols}
\Crefname{protocolcounter}{Protocol}{Protocols}

\crefname{custalgocounter}{Algorithm}{Algorithms}
\Crefname{custalgocounter}{Algorithm}{Algorithms}

\usepackage{mdframed}

\newenvironment{protocol}[1][]
{%
  \noindent \begin{minipage}{\textwidth}
    \begin{mdframed}[
      linewidth=0.8pt,
      roundcorner=5pt,
      backgroundcolor=white!,
      innertopmargin=10pt,
      innerbottommargin=10pt,
      innerleftmargin=10pt,
      innerrightmargin=10pt,
      skipabove=\topsep,
      skipbelow=\topsep,
    ]
      \refstepcounter{protocolcounter}
      \begin{center}
        \textbf{Protocol \theprotocolcounter:} #1
        \linebreak
      \end{center} 
}%
{%
    \end{mdframed}
  \end{minipage}
}

\newcounter{custalgocounter}
\renewcommand{\thecustalgocounter}{\arabic{custalgocounter}}

\usepackage{mdframed}

\newenvironment{custalgo}[1][]
{%
  \noindent \begin{minipage}{\textwidth}
    \begin{mdframed}[
      linewidth=0.8pt,
      roundcorner=5pt,
      backgroundcolor=white!,
      innertopmargin=10pt,
      innerbottommargin=10pt,
      innerleftmargin=10pt,
      innerrightmargin=10pt,
      skipabove=\topsep,
      skipbelow=\topsep,
    ]
      \refstepcounter{custalgocounter}
      \begin{center}
        \textbf{Algorithm \thecustalgocounter:} #1
        \linebreak
      \end{center} 
}%
{%
    \end{mdframed}
  \end{minipage}
}

\begin{document}

\title{Quantum PCPs: on Adaptivity, Multiple Provers and Reductions to Local Hamiltonians}

\author{Harry Buhrman}
\affiliation{QuSoft, University of Amsterdam \& Quantinuum, Partnership House, Carlisle Place, London}

\author{Jonas Helsen}
\affiliation{QuSoft \& CWI, Amsterdam, the Netherlands}
\orcid{0000-0001-7218-2585}

\author{Jordi Weggemans}
\affiliation{QuSoft \& CWI, Amsterdam, the Netherlands}
\orcid{0000-0002-8469-6900}
\email{jrw@cwi.nl}
\homepage{https://jordiweggemans.github.io}

\maketitle

\begin{abstract}
    \noindent We define a general formulation of quantum PCPs, which captures adaptivity and multiple unentangled provers, and give a detailed construction of the quantum reduction to a local Hamiltonian with a constant promise gap. This reduction turns out to be a versatile subroutine to prove properties of quantum PCPs, allowing us to show:
    \begin{enumerate}[label=(\roman*)]
        \item Non-adaptive quantum PCPs can simulate adaptive quantum PCPs when the number of proof queries is constant. In fact, this can even be shown to hold when the non-adaptive quantum PCP picks the proof indices \emph{uniformly} at random from a subset of all possible index combinations, answering an open question by Aharonov, Arad, Landau and Vazirani (STOC '09).
        \item If the $q$-local Hamiltonian problem with constant promise gap can be solved in $\QCMA$, then $\QPCP[q] \subseteq \QCMA$ for any $q \in \mO(1)$.
        \item If $\QMA[k]$ has a quantum PCP for any $2 \leq k \leq \poly(n)$, then $\QMA[2] = \QMA$, connecting two of the longest-standing open problems in quantum complexity theory.
    \end{enumerate}
    Moreover, we also show that there exist (quantum) oracles relative to which certain quantum PCP statements are false. Hence, any attempt to prove the quantum PCP conjecture requires, just as was the case for the classical PCP theorem, (quantumly) non-relativizing techniques. 
\end{abstract}

\newpage 
\setcounter{tocdepth}{2}
\tableofcontents
\newpage

\section{Introduction}
Arguably the most fundamental result in classical complexity theory is the Cook-Levin Theorem~\cite{cook1971the,levin1973universal}, which states that constraint satisfaction problems (CSPs) are $\NP$-complete. Kitaev showed that there exists a natural analogue in the quantum world: local Hamiltonian problems, which can be viewed as a quantum generalisation of classical constraint satisfaction problems, are complete for the class $\QMA$, which is a quantum generalisation of $\NP$~\cite{Kitaev2002ClassicalAQ}.\footnote{There are many possible quantum generalisations of $\NP$, see~\cite{gharibian2024guest} for a recent review.}

There are many variants of  $\NP$, which a priori seem either more powerful or more restrictive, but turn out to be equal to $\NP$. Examples of these are interactive proof systems of various flavours with a deterministic verifier, and perhaps most surprisingly and importantly, probabilistically checkable proof systems (PCPs)~\cite{Arora1998proof}. A PCP has a polynomial-time probabilistic verifier with query access to a proof (provided by a prover). Usually, a PCP is specified by two parameters: the number of random coins the verifier is allowed to use (which also upper bounds the length of the proof as some exponential in the number of coins) and the number of queries it can make to the proof. The PCP theorem~\cite{Arora1998proof,Arora1998probabilistic}, which originated from a long line of research on the complexity of interactive proof systems, states that all problems in $\NP$ can be decided, with a constant probability of error, by only using a logarithmic number of coin flips and by querying a constant number of bits of the proof. Equivalently, it implies that it is $\NP$-hard to decide whether an instance of a CSP is either completely satisfiable or no more than a constant fraction of its constraints can be satisfied. This is usually referred to as the \emph{hardness of approximation} formulation of the PCP theorem. It is also possible to prove the PCP theorem by reducing a CSP \textit{directly} to another CSP which has the above property. This transformation, due to Dinur~\cite{Dinur2007thepcp}, is usually referred to as \emph{gap amplification}, referring to the increase in the difference (the gap) in the fraction of constraints which can be satisfied in both the {\sc Yes}- and {\sc No}-instances.

Naturally, quantum complexity theorists have proposed proof-checking and hardness of approximation versions of a quantum PCP conjecture. The hardness of approximation formulation states that the local Hamiltonian problem with constant promise gap, relative to the operator norm of the Hamiltonian, is $\QMA$-hard. This formulation of the quantum PCP conjecture has been the predominant focus of the quantum PCP literature, and progress has been made in giving evidence both in favour and against the conjecture. Amongst the positive are the NLTS theorem and its cousins, excluding a large set of potential $\NP$-witnesses~\cite{anshu2023nlts,coble2023local,anschuetz2023combinatorial,coble2023hamiltonians,herasymenko2023fermionic}. Evidence against (assuming $\NP \neq \QMA$) are results showing that under many extra constraints one poses on the local Hamiltonian problem (constraining e.g.~the interaction graph or ground space structure) the problem is classically solvable when the relative promise gap is constant, whilst the problem remains (quantumly) hard when the gap is inverse polynomial~\cite{bansal2007classical,brandao2013product,Aharonov2019stoquastic,gharibian2021dequantizing,Cade2022Complexity,cade2023improved,weggemans2023guidable}.

The proof-checking formulation of the quantum PCP conjecture, which states that one can solve any promise problem in $\QMA$ by using a quantum verifier which only accesses a constant number of qubits from a quantum proof, has received considerably less attention. The reason for this is that it is not hard to show that both conjectures are equivalent under \textit{quantum} reductions. This was already pointed out in the first work which proposed a quantum PCP formulation~\cite{Aharonov2008The}.\footnote{Whilst this is widely known in the community, this reduction has (as far as the authors are aware) never been written down in full detail except in the works of~\cite{Grilo2018thesis} and~\cite{herasymenko2023fermionic}, which both consider quantum PCPs that are not (or at least shown to be) fully general. In both of these works, a quantum PCP is defined using a single POVM, where each POVM element acts on an ancilla register and a small part of the provided quantum proof. The verifier selects the POVM element by sampling from a classical probability distribution (in~\cite{Grilo2018thesis} this is simply the uniform distribution).} Perhaps that is why, even after more than two decades since the question of the existence of a quantum PCP was first posited~\cite{aharonov2002quantum}, many basic questions regarding the proof-checking formulation have not been addressed. For example, as already raised in~\cite{Aharonov2008The}, how robust are definitions of quantum PCPs under subtle changes, e.g.~the choice of picking the distribution over which the proof qubits are selected? Are adaptive queries to the proof more powerful than non-adaptive in the constant query setting or are they the same as is the case for the classical setting? Are there non-equivalent variations of quantum PCPs, similar to the many natural variations of $\QMA$? 

Motivated by these questions, this work is concerned with studying the properties of quantum PCPs through proof verification. Studying the connection between quantum PCPs and local Hamiltonians turns out to be a powerful tool for showing basic properties of quantum PCPs.

\subsection{Quantum PCPs are constant promise gap local Hamiltonians}
\label{ssec:QPCPs_are_intro}
We give a detailed definition of a quantum PCP verifier, which allows for adaptive queries as well as quantum proofs consisting of multiple unentangled quantum states given by different unentangled provers. Informally (see~\cref{def:adap_QPCP} for the full definition) the verifier makes $q$ queries to a quantum proof in the following way:
\begin{itemize}
    \item The verifier takes as an input a string $x$ of size $n$ (which will be hardcoded into the circuits), has a polynomial-sized ancilla register (which we will call the workspace) and a proof register containing $k$ unentangled quantum proofs of some specified size depending on $n$. Generally, $k=1$ unless specified otherwise.
    \item The verifier is going to do the following $q$ times. At stage $t \in [q]$, suppose the quantum PCP has decided to act on proof qubits from the set $I = \{ i_1,\dots,i_{t-1} \}$. The quantum PCP then applies a polynomial-time quantum circuit to the workspace and the proof qubits with indices $i_1,\dots,i_{t-1}$ followed by a measurement of some designated qubits to determine which next index $i_{t}$ is going to be added to the set $I$.
    \item Finally, another polynomial-time circuit is applied to the workspace and the proof qubits with indices $i_1,\dots,i_q$ and a single designated output qubit is measured in the computational basis. The quantum PCP accepts if and only if the outcome is $\ket{1}$. 
\end{itemize}
Note that this quantum PCP is adaptive; similarly, a non-adaptive quantum PCP can be defined using only a single measurement to decide the indices $i_1, \dots, i_q$. One can then construct complexity classes based on this verifier in a general setting, where we also include the possibility of multiple provers. 

Our first goal is to show that there exists a local Hamiltonian with PSD terms whose expectation value has a one-to-one correspondence with the acceptance probabilities of the quantum PCP given a quantum proof. This involves some subtleties compared to the classical case, as each intermediate measurement to determine the next index can alter the \emph{entire} quantum state of the quantum verifier due to entanglement across different registers. Nevertheless, we prove the following result.

\begin{customlem}{1.1 (Informal)}[from~\cref{lem:H_from_QPCP}] Let $x$ be an input and $V_x$ a $q$-query quantum PCP verifier with $x$ hardcoded into it. Then there exists a $q$-local Hamiltonian $H_x$ such that
\begin{align*}
    \Pr[V_x \text{ accepts } \xi] = 1 - \tr[H_x \xi],
\end{align*}
where $H_x = \sum_{i \in [m]} H_{x,i}$ with each $H_{x,i}$ PSD and $\norm{H_x} \leq 1$.
\label{lem:1.1_inf}
\end{customlem}
The core idea is that instead of focusing on conditional probabilities arising from specific query paths, we consider the overall probability of following a particular adaptive query path \emph{and} accepting. For a fixed input, we show that all input, ancilla, and circuit information can be directly encoded into a single POVM element whose expectation value captures this probability. Then, by linearity of expectation, we simply add all POVM elements together to form our final Hamiltonian term.

However, the quantum PCP conjecture is typically framed differently from the form of the Hamiltonian in~\hyperref[lem:1.1_inf]{Lemma~{1.1}}. First, it is often assumed that the interaction graph corresponding to the qubits on which the local terms act has constant degree.\footnote{The interaction degree here is defined as the maximum number of terms that act simultaneously on a single qubit.} However, this so-called \emph{interaction degree} depends on the support of the distribution over which the indices are selected, which could, in principle, include all possible choices of indices. As a result, the interaction graph can have degree $\Omega(n)$. Second, one assumes that $H$ is of the form $H = \sum_{i \in [m]} H_i$ with $0 \preceq H_i \preceq 1$ and that $\lambda_\textup{min}(H) \leq \alpha m$ or $\lambda_\textup{min}(H) \geq \beta m$. These Hamiltonians can be easily rewritten in the same form as $H_x$ by letting $H' = \sum_{i \in [m]} H'_i$ with $H'_i := H_i/m$, such that $\lambda_\textup{min}(H') \leq \alpha $ or $\lambda_\textup{min}(H') \geq \beta$. However, this formulation implicitly assumes that in the latter case, at least a constant fraction of the terms contribute on the order of $\sim \frac{1}{m}$ to the ground state energy, meaning that the energy is relatively spread out over all terms. If the Hamiltonian's interaction degree were in fact constant, then such energy spreading would indeed hold in the large energy case. However, for our definition of quantum PCPs, this does not necessarily have to be the case. This difference is subtle but important: one can no longer pick a constant number of terms uniformly at random and measure their energy to solve the corresponding local Hamiltonian problem, as only a small $o(1)$ fraction of the terms might have large energy.  Nevertheless, we show that there exists a transformation from any such Hamiltonian to another one which does satisfy this property, at the cost of increasing the locality by a factor of two and decreasing the promise gap exponentially in the locality, which is still constant if the locality is constant (\cref{lem:H_smoothing}).\\

Now we have shown that the probability that a $\QPCP$-verifier accepts a certain proof is equivalent to the expectation value of some Hamiltonian, we still need to show that the descriptions of the local terms can be obtained efficiently. For this, it is well-known that one can adopt a \textit{quantum reduction}~\cite{Aharonov2008The}. We show that a quantum reduction also exists for our general formulation of quantum PCPs. Moreover, by a simple trick, we find that we can round the Hamiltonian in such a way that every time the reduction succeeds (which can be done with a success probability exponentially close to one) it produces \textit{exactly} the same Hamiltonian.
\begin{customthm}{1.1 (Informal)}[from~\cref{thm:red} and~\cref{cor:fixed_H}] Let $x$ be an input. There exists a quantum reduction from a $q$-query quantum PCP with circuit $V_x$ to a $q$-local Hamiltonian $\hat{H}_x$ such that
\begin{align*}
    \abs{\Pr[V_x \text{ accepts } \xi] - \left(1 - \tr[\hat{H}_x \xi] \right)}\leq \epsilon,
\end{align*}
which succeeds with probability $1-\delta$ and runs in time $\poly(n,1/\epsilon,\log(1/\delta))$. Moreover, this reduction can be made such that it produces the same Hamiltonian every time it succeeds.
\label{thm:1.1_inf}
\end{customthm}
\noindent The specificity of the reduction (i.e.~it always yields the same Hamiltonian and the acceptance probability is directly encoded in the energy of the Hamiltonian) makes it a useful tool in deriving a variety of statements about quantum PCPs, since quantum PCPs themselves can perform the reduction. 

\subsection{Results from applying the reduction}
All applications are derived by adopting variations of the following general strategy:
\begin{enumerate}
    \item The reduction of a quantum PCP to a local Hamiltonian is used as a pre-processing step of some larger protocol, rounding to a fixed Hamiltonian with high success probability. 
    \item One uses tricks of Hamiltonian complexity to manipulate and solve the corresponding local Hamiltonian problems.
\end{enumerate}
It is important to emphasise that by this strategy, the reductions we propose are simply \textit{classical} polynomial-time reductions, even though they involve quantum reductions. The reason for this is that the quantum reductions are absorbed into the quantum protocols as pre-processing steps. A disadvantage of this strategy is that it in general fails to preserve perfect completeness.

\subsubsection{Reduction to average particle energy formulation}
As a first application, we consider a specific formulation of the quantum PCP conjecture in terms of an error constant relative to the number of sites (i.e., qubits or qudits) rather than the total number of terms. This formulation has appeared in several works (see, for example, the recent preprints~\cite{natarajan2024status} and~\cite{anshu2024uniqueqma}), but to our knowledge, it has never been shown to be complete for the class $\QPCP[\mO(1)]$. Using our reductions above, combined with standard techniques from tail bounds for sums of random matrices, we can show that even under this restriction, the local Hamiltonian with a constant promise gap remains complete for $\QPCP[\mO(1)]$.

\begin{customthm}{1.2 (Informal)}[from~\cref{cor:density_form}] There are constants $c < d$ such that deciding if the ground energy of an $n$-qubit local Hamiltonian $H$ with $m = \Theta(n)$ terms is (yes case) $\leq c $ or (no case) $\geq d $ is $\QPCP[\mO(1)]$-complete under quantum polynomial-time reductions.
\end{customthm}

\subsubsection{Upper bounds on quantum PCPs}
By a result of~\cite{weggemans2023guidable}, it seems unlikely that there exists a \emph{classical} reduction that satisfies the same properties as~\hyperref[thm:1.1_inf]{Theorem~{1.1}}, as that would show that $\BQP \subseteq \QCPCP \subseteq \NP$.\footnote{Here $\QCPCP$ is just as $\QPCP$ but with the promises for classical proofs instead of quantum proofs, similar to what $\QCMA$ is to $\QMA$.} The other direction (going from a local Hamiltonian with constant to a $\QPCP$ verification circuit) is easy to show through Kitaev's original $\QMA$ verification circuit for the local Hamiltonian problem (see~\cref{prot:Kitaev_protocol}), once the Hamiltonian has been transformed in the required form (equally bounded operator norm on all local terms). Hence, if one shows that $\QPCP[q] \subseteq \mathcal{C}$ for some complexity class $\mathcal{C}$ that can perform polynomial-time classical reductions (so any class $\mathcal{C}\supseteq \Pclass$ would suffice), then the local Hamiltonian problem with constant promise gap is in $\mathcal{C}$ as well. However, since the reduction from $\QPCP$ to a local Hamiltonian is a quantum reduction, such a statement does not hold in the other direction. Weaker conclusions can be drawn however. In particular we can show that an upper bound on the local Hamiltonian problem with a constant promise gap implies an upper bound on quantum PCPs with the same locality.

\begin{customthm}{1.3 (Informal)}[from~\cref{thm:QCMA}] If the $q$-local Hamiltonian problem with constant promise gap is in $\QCMA$, then $\QPCP[q] \subseteq \QCMA$.
\end{customthm}
This entails that, assuming $\QMA \neq \QCMA$, disproving $\QMA$-hardness for either the proof verification or Hamiltonian formulation of the quantum PCP conjecture (by placing it in $\QCMA$) does indeed disprove the other (similar to how proving $\QMA$ hardness of one implies $\QMA$-hardness of the other.)

\subsubsection{Adaptive and non-adaptive quantum PCPs}
\label{ssec:adaptive_vs_non_adaptive_intro}
\noindent It is a well-known fact that non-adaptive PCPs can simulate adaptive PCPs in the classical setting when the number of proof queries is only allowed to be constant. This can easily be shown by simulating the adaptive PCP verifier over all possible local proof settings for a fixed setting of the randomness, after which the actual proof can be read at the used locations to check which setting was correct. However, such tricks seem to have no quantum analogue, as it is generally impossible to fix the inherent randomness (or ``quantumness'') in a quantum algorithm~\cite{aaronson2021acrobatics} (see also the discussion in~\cref{ssec:QPCPs_are_intro}). However, we can show that a quantum analogue of this result still holds.

\begin{customthm}{1.4 (Informal)}[from~\cref{thm:A_vs_NA}] Non-adaptive quantum PCPs can simulate adaptive quantum PCPs when the number of proof queries is constant.
\end{customthm}

In fact, we are able to prove something stronger: we show that any constant query PCP can be simulated by a quantum PCP that picks the indices of the proof bits uniformly at random from a subset of all possible index choices, resolves an open question of~\cite{Aharonov2008The} (see~\cref{rem:uniform_QPCPs} in the main text).\footnote{This also shows that the quantum PCP formulation as in~\cite{Grilo2018thesis} is in fact fully general when one considers quantum PCPs which make a constant number of proof queries.} We prove these results by showing that there exists a non-adaptive quantum PCP verifier for the local Hamiltonian induced by the adaptive quantum PCP, which only has its promise gap exponentially smaller in the locality. Weak error reduction can then be used to boost the promise gap back to its original value.

\subsubsection{Quantum PCPs for \texorpdfstring{$\QMA[k]$}{QMA(k)}}
\noindent Finally, we consider a quantum PCP in the $\QMA[k]$ setting: the complexity class $\QMA[k]$ is a generalisation of $\QMA$ where there are multiple uninteracting provers, which are guaranteed to be unentangled with each other. It is known that $\QMA[k] = \QMA[2]$ by the result of Harrow and Montanaro~\cite{harrow2013testing}, but is generally believed that $\QMA \neq \QMA[2]$ since there are problems known to be in $\QMA[2]$ but not known to be in $\QMA$~\cite{liu2007quantum,aaronson2008power,beigi2008np,blier2009all}. Moreover, the best upper bound we have on $\QMA[2]$ is that it is contained in $\NEXP$, which follows by just guessing exponential-size classical descriptions of the two quantum proofs. The reason for this is that the maximum acceptance probability of a $\QMA[2]$ verifier is with respect to separable states, which means that the corresponding maximisation problem is no longer convex. Another surprising fact is that the separable local Hamiltonian problem is $\QMA$-complete, and we only have that a separable Hamiltonian problem becomes $\QMA[2]$-complete when one considers \textit{sparse} Hamiltonians, where the individual terms can only have a polynomial number of non-zero entries but do not have to be local~\cite{chailloux2012complexity}. Since our quantum reduction from a quantum PCP always produces a local Hamiltonian irrespective of the number of provers, a generalisation of the protocol of~\cite{chailloux2012complexity} to the $k$-separable local Hamiltonian setting allows us to show the following.

\begin{customthm}{1.5 (Informal)}[from~\cref{thm:qma2}] If $\QMA{[k]}$ has a $k$-prover quantum PCP for any $2 \leq k \leq \poly(n)$, then we have that $\QMA{[2]} = \QMA$.
\end{customthm}
\noindent This shows a connection between two of the biggest open problems in quantum complexity. In particular, if one believes that a quantum PCP conjecture for $\QMA$ should hold and that $\QMA[2] \neq \QMA$, then any proof of the quantum PCP conjecture should not translate to the multiple unentangled prover setting.

As a bonus, we improve the parameter range of $\QMA$-containment of the consistency of local density matrices ($\CLDM$) problem (\cite{liu2007consistency,broadbent2022qma}), by giving a new protocol.

\subsection{Proofs of the quantum PCP theorem do not relativize}
In the final section, we will show that given Aaronson and Kuperberg's quantum oracle~\cite{aaronson2007quantum} (under which $\QCMA \neq \QMA$) we have that the quantum PCP conjecture is false. Their oracle separation is fundamentally based on a lower bound combined with a counting argument, which exploits the fact that there are doubly exponential quantum states (with small pairwise fidelities) but only an exponential number of classical proofs. The reason this lower bound cannot be applied directly to $\QPCP$ comes from the fact that the total number of quantum proofs that can be given is the same as for $\QMA$. However, from the verifiers perspective, the total number of proofs he observes is limited by a set of local density matrices and indices indicating from what part of the global state this density matrix comes from. Using this observation, we can define an artificial variation of $\QPCP$, denoted $\QPCP_\epsilon$, which can be viewed as a ``proof-discretized''-version of $\QPCP$. Informally, this class can be viewed as a variant to $\QPCP$ where after the verifier decides which parts of the proof are going to be accessed, some ``magical operation'' takes this part of the proof and projects it onto the density matrix closest to it from some fixed finite set. We can show that with respect to all oracles, this new class contains $\QPCP$, so any oracle separation between $\QMA$ and $\QPCP_\epsilon$ also separates $\QPCP$ and $\QMA$. 

\begin{customthm}{1.6 (Informal)}[from~\cref{thm:q_oracle_sep} and~\cref{cor:classical_oracle}] There exists a quantum (classical) oracle relative to which the polylog (constant) quantum PCP conjecture is false.
\end{customthm}
Since for any oracle~$X$ encoding a $\PSPACE$-complete language we have $\QPCP[q]^X = \QMA^X$, and since the inclusions $\QPCP[q] \subseteq \QMA \subseteq \PSPACE$ all relativize with respect to classical oracles and $\PSPACE^\PSPACE = \PSPACE$, it follows that this result implies proving the quantum PCP conjecture by showing $\QPCP[q] = \QMA$ for some constant $q \in \mO(1)$ requires non-relativizing techniques---just as was the case for the classical PCP theorem~\cite{fortnow1994role}.

Perhaps the most important takeaway from this result is that it should apply to any quantum oracle separation involving $\QMA$ which exploits the fact that there are doubly exponentially many proofs in $\QMA$. Therefore, all oracle separations between $\QMA$ and $\QCMA$ based on such an argument immediately yield an oracle relative to which the quantum PCP conjecture is false.

\subsection{Discussion and open problems}
We have studied the quantum PCP conjecture in the proof-checking formulation (as opposed to the more popular local Hamiltonian formulation). Ironically, many of our results follow from leveraging the reduction to the local Hamiltonian problem. Our results take the form of both novel statements, and formal proofs of what we think of as ``folklore knowledge'', i.e., results we suspect are known to be true but for which we could not find a proof in the literature. Given the renewed interest in the quantum PCP conjecture we believe this was an effort worth undertaking. We conclude by listing several questions on quantum PCPs which we have not yet resolved but believe are interesting avenues for future work.

\paragraph{Locality reductions}
We do not know whether~\cref{thm:QCMA} already holds when only the $2$-local Hamiltonian problem is in $\QCMA$. In~\cite{brandao2013product}, it is claimed that the $2$-local Hamiltonian problem is complete for $\QPCP$ (which would imply the above statement), but it is not clear to us that this is actually the case. The reason for this is that the gadget constructions of~\cite{bravyi2008quantum} mentioned in~\cite{brandao2013product}, which transform a $q$-local Hamiltonian into a $2$-local Hamiltonian whilst preserving the relative promise gap, only work when every qubit is involved in only a \emph{constant} number of terms. In fact, the existence of a transformation which maps a $q$-local Hamiltonian with an interaction degree $\Omega(n)$ to a $2$-local Hamiltonian with an interaction degree $\Omega(n)$ would directly imply that the local Hamiltonian problem with constant promise gap is in $\NP$, disproving the quantum PCP conjecture for any constant $q$ directly. The reason for this is as follows\footnote{This argument is based on a similar---unpublished---argument due to Anshu and Nirkhe, which can be found at (at time of writing) \url{https://anuraganshu.seas.harvard.edu/files/anshu/files/bh_4local.pdf}}: take any $q$-local Hamiltonian $H$, $q \geq 2$ constant,  with some degree $d$ interaction graph where $d \geq 2$,  $\norm{H} = 1$ (which can be assumed w.l.o.g.), and completeness and soundness parameters $b$ and $a$, respectively, satisfying $(b-a)/\norm{H} = \Omega(1)$. Then we have the Hamiltonian $H' = H^2$ is $2q$-local, has operator norm still equal to $1$, has an interaction graph of degree $n$, and completeness and soundness parameters $b',a'$ with $b'-a'/\norm{H'} \geq (b-a)^2 = \Omega(1)$. If there would exist a transformation from $H'$ to $H''$ with $H''$ $2$-local, interaction degree $\mO(n)$ and with completeness $a''$ and soundness $b''$ satisfying $(b''-a'')/\norm{H''} = \Omega(1)$, then by Brandao-Harrow~\cite[Cor. $5$]{brandao2013product} one could decide the correct energy of $H$.

\paragraph{Quantum PCPs with perfect completeness}
A major downside of our technique---leveraging the $\QPCP$ to local Hamiltonian quantum reduction---is that it fails to preserve perfect completeness, which is due to the Hamiltonian being learned only up to some small error. It is an open question whether other techniques can be used to obtain similar results in the perfect completeness setting, for example whether non-adaptive quantum PCPs can simulate adaptive ones when the number of quantum proof queries is only allowed to be constant. 

\paragraph{Strong error reduction} It is easy to show that non-adaptive quantum PCPs with near-perfect completeness allow, just like $\QMA$~\cite{Mariott2005quantum}, for strong error reduction (\cref{app:strong_err_red}). However, we do not know whether this also holds in the general case. The idea of taking polynomials of the Hamiltonian to manipulate its spectrum seems not to be compatible with Kitaev's energy estimation protocol, as this in general introduces coefficients that can be negative and have large absolute values, which is incompatible with Kitaev's energy estimation protocol (\cref{subsec:KEEP}) which requires each term to have to be PSD and have operator norm at most one. 

\paragraph{Acknowledgements}
JW was supported by the Dutch Ministry of Economic Affairs and Climate Policy (EZK), as part of the Quantum Delta NL programme. JH acknowledges support from the Quantum Software Consortium (NWO Zwaartekracht 024.003.037) and a Veni grant (NWO Veni 222.331).

\section{Preliminaries}
\subsection{Notation}
For a Hilbert space $\mathcal{H}$ of dimension $d$, we denote $\mathcal{D}(\mathcal{H}) =\{\rho \in \text{PSD}(\mathcal{H}), \tr[\rho]=1\}$, where $\text{PSD}(\mathcal{H})$ is the set of all positive semidefinite matrices in $\mathcal{H}$, for the set of all density matrices in $\mathcal{H}$. For the set of all $d$-dimensional pure states, we write $\mathcal{P}(\mathcal{H}) =\{\rho \in \text{PSD}(\mathcal{H}), \tr[\rho^2]=1 \}$ (normally defined as $\mathbb{CP}^{d-1}$, i.e.~complex projective space with dimension $d-1$). For a tuple consisting of $q$ distinct elements, i.e., $(i_1,i_2,\dots,i_q)$, we write $\{(i_1,i_2,\dots,i_q)\}!$ for the set of all permutations of the tuple $(i_1,i_2,\dots,i_q)$. We denote $\binom{[n]}{k}$ for the set of all $k$-element subsets of $[n]$ (so unordered and without repetitions).  When the base is not explicitly stated, $\log$ always refers to the base-$2$ logarithm.

\subsection{Complexity theory}
All verifiers used to define quantum complexity classes in this work will be defined in terms of the quantum circuit model. We say that a quantum verifier $V_n$, taking an input $x \in {0,1}^n$ for some integer $n$, is polynomial-time if it uses a workspace of $\poly(n)$ ancilla qubits initialized in $\ket{0}$ (tensored with the input in $\ket{x}$) and at most $\poly(n)$ elementary quantum gates and elementary (intermediate) POVMs (or PVMs). To efficiently generate descriptions of its components depending only on the input size $n$, we need the notion of \emph{$\Pclass$-uniformity}.

\begin{definition} A set of verifiers $\{V_n : n \in \mathbb{N}\}$ is polynomial-time uniform (abbreviated as $\Pclass$-uniform) if there exists a polynomial-time deterministic Turing machine that, on input $1^n$, outputs a classical description of $V_n$. 
\label{def:Puniformity} 
\end{definition}

We start by recalling the basic definitions from complexity theory classes used in this work.
\begin{definition}[$\QMA{[k]}$] Let $n \in \mathbb{N}$ be the input size. A promise problem $A = (A_\text{yes},A_\text{no})$ is in $\QMA[k,c,s]$ if and only if there exists a $\Pclass$-uniform family of polynomial-time quantum verifiers $\{V_n\}$ and a polynomial $p$, where $V_n$ takes as input a string $x\in \{0,1\}^n$ and a $k p(n)$-qubit witness quantum state $\ket{\psi}$ and decides on acceptance or rejection of $x$ such that
\begin{itemize}
    \item if $x \in A_\textup{\sc yes} $ then there exist $p(n)$-qubit witness states $\ket{\psi_j} \in \left(\mathbb{C}^2\right)^{\otimes p(n)}$, $j\in [k]$, such that $V_n$ accepts $(x,\otimes_{j \in [k]} \ket{\psi_j})$ with probability $\geq c$,
    \item if $x \in A_\textup{\sc no}$ then for all $\ket{\psi_j} \in \left(\mathbb{C}^2\right)^{\otimes p(n)}$, $j\in [k]$, we have that $V_n$ accepts $(x,\otimes_{j \in [k]} \ket{\psi_j})$ with probability $\leq s$,
\end{itemize}
where $c-s = \Omega(1/\poly(n))$. If $c=2/3$ and $s=1/3$, we abbreviate to $\QMA[k]$, and if $k=1$, we have $\QMA = \QMA[1]$.\footnote{In the literature, $\QMA(k)$ is commonly used, but we chose to use ``$[.]$''-brackets, as this is customary in complexity theory for classes that depend on some parameters.}
\end{definition}
\begin{subdefinition}[$\QCMA$] This has the same definition as $\QMA$ but where the promises hold with respect to computational basis states $\ket{y}$, where $|y|=p(n)$. In this case we trivially have $\QCMA[k] =\QCMA$ for all $k \leq \poly(n)$.
\end{subdefinition}

By a result of Harrow and Montanaro, it is known that $\QMA$ with $2$ unentangled provers can simulate any polynomial number of unentangled provers.
\begin{lemma}[\cite{harrow2013testing}] For any $2 \leq k \leq \poly(n)$, we have 
\begin{align*}
    \QMA[k] = \QMA[2].
\end{align*}
\label{lem:QMAk}
\end{lemma}

\begin{fact} For any $1 \leq k \leq \poly(n)$, the completeness and soundness parameters in $\QMA[k]$ and $\QCMA$ can be made exponentially close to $1$ and $0$, respectively, i.e.~$c=1-2^{-\mO(n)}$ and $s=2^{-\mO(n)}$.
\label{fact:err_red_QMA}
\end{fact}

The $q$-local $k$-separable local Hamiltonian problem generalises the $q$-local Hamiltonian problem (which is the case for $k=1$) by giving a promise on the energies of all low-energy of $k$-fold tensor products, where each quantum state in the tensor product lives on a pre-determined number of qubits. It is shown to be $\QMA$-complete in the case where $k=2$ in~\cite{chailloux2012complexity}. 

We will define local Hamiltonian problems such that all local terms are PSD, as this will turn out to be the relevant setting in the context of quantum PCPs. Generally, one can apply the transformation $H_i \mapsto (H_i + \mathbb{I})/2$ for each local term, which makes every local term PSD.

\begin{definition}[$k$-separable $q$-local Hamiltonian problem] $(k,q)$-$\mathsf{LH}[\delta]$\\
    \textbf{Input:} 
    A classical description of a collection of $q$-local PSD operators $\{H_i\}_{i \in [m]}$, which define an $n$-qubit $k$-local Hamiltonian $H = \sum_{i=1}^{m} H_i$ with $\norm{H} \leq 1$, and two parameters $a,b \in \mathbb{R}$ such that $b - a = \delta$. We call $\delta$ the promise gap. \\
\textbf{Promise:} We have that either one of the following two holds:
\begin{enumerate}[label=(\roman*)]
    \item There exists a quantum state  $\xi = \otimes_{j \in [k]} \xi_j $, such that $\tr[H \xi]\leq a$
    \item For all quantum states $\xi = \otimes_{j \in [k]} \xi_j $ we have that $\tr[H \xi] \geq b$.
\end{enumerate}  
\textbf{Output:} \begin{itemize}
    \item If case (i) holds, output {\sc yes}.
    \item If case (ii) holds,, output {\sc no}.
\end{itemize}
\label{def:LH_sep}
If $k=1$, the above problem reduces to the standard local Hamiltonian problem which we abbreviate as $q$-$\mathsf{LH}[\delta]$.
\end{definition}

\section{Quantum probabilistically checkable proofs}
\label{sec:QPCP_def}
The bulk of this section is taken up a definition of a quantum PCP. Our definition is more detailed than that given in \cite{aharonov2013guest}, and more general than the definitions in \cite{herasymenko2023fermionic,Grilo2018thesis}, in that we also include the ability to query proof qubits adaptively.  Moreover, we allow for the possibility of having multiple provers (i.e.~unentangled proofs). Once we have this definition in place we can define an associated complexity class $\QPCP$. For convenience, we also explicitly define a non-adaptive version. We finish up by arguing that weak error reduction (and limited strong error reduction) is possible in $\QPCP$.

\begin{definition}[quantum PCP verifier] Let $p_1,p_2,p_3 : \mathbb{N} \rightarrow \mathbb{N}$, and $n \in \mathbb{N}$ be the input size. A $(k,q,p_1,p_2,p_3)\text{-}\QPCP$-verifier $V$ consists of the following:
\begin{itemize}
    \item a $n$-qubit input register $A$, initialised in input $\ket{x}$, $x \in \{0,1\}^n$;
    \item a $p_1$-qubit ancilla register $B$, initialised in $\ket{0}^{\otimes p_1(n)}$;
    \item a $k p_2$-qubit proof register $C$, initialised in $\xi= \otimes_{j=1}^k \xi_j$ for some quantum witnesses $\xi_j \in \mathcal{D}\left(\left(\mathbb{C}^{2}\right)^{\otimes p_2(n)}\right)$ for all $j \in [k]$;
    \item a collection of PVMs $ \Pi^t = \{\Pi_{i}^t \}$,  $i \in [kp_2(n)]$, with $\Pi_{i}^t =\ketbra{i} \otimes \mathbb{I}$ for all  $t \in [q]$, $\Pi^{\textup{output}} = \{\Pi_{0}^{\textup{output}},\Pi_1^{\textup{output}}\}$, with $\Pi_{0}^{\textup{output}} = \ketbra{0} \otimes \mathbb{I}$ and $\Pi_{1}^{\textup{output}} = \ketbra{1} \otimes \mathbb{I}$;
    \item collection of circuits $V^{t'}$, $t' \in [q+1]$, where circuit $V^{t'}$ only acts on at most $t'$ qubits of the proof register and consist of at most $p_3$ gates from some universal gate set.
\end{itemize}
Let $I = \emptyset$ be the set of all proof indices to be accessed. The quantum PCP verifier $V$ acts as follows:
\begin{enumerate}
    \item For $t \in [q]$:  it applies the circuit $V^{t}$ to registers $A$, $B$ and qubits $I$ from $C$, performs the measurement $\Pi^{t}$ and adds outcome $i_t$ to the set $I$;
    \item It applies $V^{q+1}$ followed by a measurement of $\Pi^\textup{output}$ of the first qubit and returns ``accept'' if the outcome was $\ket{1}$, and ``reject'' if the outcome was $\ket{0}$. 
\end{enumerate} 
If $ p_1, p_2 $, and $ p_3 $ are all polynomially bounded functions, then we abbreviate to a $(k,q)$\nobreakdash-$\QPCP$ verifier, and to a $(q)$\nobreakdash-$\QPCP$ verifier if additionally $ k = 1 $. For the remainder of this work, we will assume that in Step 1 for any $t$ the probability of measuring any outcome $i_t$ for which there exists a $i_{t'} \in I$, $t' \in [t-1]$ such that $i_t = i_{t'}$, is zero.
\label{def:adap_QPCP}
\end{definition}

We make the final assumption in~\cref{def:adap_QPCP} because it simplifies our notation, and there would be no benefit in querying the same proof index multiple times. This assumption can easily be enforced in a $\QPCP$-verifier by adding a random sampling step whenever a duplicate index is observed.

\begin{remark}
We will often write $V_x$ to denote the verifier with the input $x$ hardcoded. Moreover, in the case of multiple provers, an index $i_t$ will denote a tuple $(j,k)$, where $j$ indicates the corresponding proof and $k$ indicates the index of the qubit in this proof.
\end{remark}

\begin{subdefinition}[Non-adaptive quantum PCP verifier] 
A non-adaptive $(k,q)$-$\QPCP_{\textup{NA}}$-verifier is just like the $(k,q)$-$\QPCP$ verifier but instead has a single PVM $\Pi = \{ \Pi_{i_1,\dots,i_q} \}$ with $\Pi_{i_1,\dots,i_q} = \ketbra{i_1,\dots,i_q} \otimes \mathbb{I}$, which determines all $q$ qubits to be accessed. If $k=1$, we simply refer to a $(q)$\nobreakdash-$\QPCP_{\textup{NA}}$ verifier.
\label{def:QPCP_NA}
\end{subdefinition}

Some additional explanations are in order. First, one might wonder why~{\cref{def:adap_QPCP}} is not defined entirely in terms of PVMs that absorb the circuits ${ V^{t'} }$. As we will see later, this split is necessary to ensure an efficient unitary decomposition whenever the PVMs are at most $\mO(\log)$-local. Throughout this work, we assume an ordering on the tuples $(j,l)$, where $(j,l)$ represents the $l$th qubit of the $j$th proof, and we use basis state notation, treating strings and their integer representations interchangeably. Since there are only $k p_2(n)$ possible settings, we need to allocate at most  
$
\lceil \log \left( k p_2(n) \right) \rceil = \mO(\log n)
$  
qubits to determine the next proof index to be queried when $k, p_2(n) \in \poly(n)$. This ensures that the PVMs remain $\mO(\log n)$-local. In the constant-query setting (which is our focus), this property still holds in the non-adaptive case, even when a single PVM represents the parallel application of all single-index PVMs. As a result, each PVM has at most polynomial circuit complexity, making them efficiently implementable.

We can now define the associated complexity classes. 
\begin{definition}[$\QPCP$] Let $n \in \mathbb{N}$ be the input size, and $x \in \{0,1\}^n$ be some input. A promise problem $A = (A_\textup{yes}, A_\textup{\sc no})$ belongs to $\QPCP[k,q,c,s]$ if and only if there exist polynomially bounded functions $p_1, p_2, p_3 : \mathbb{N} \to \mathbb{N}$ and a $\Pclass$-uniform family of $(k,q,p_1,p_2,p_3)$-$\QPCP$ verifiers $\{V_n\}$ such that
\begin{itemize}
\item If $x \in A_\textup{\sc yes}$, then there exist quantum states $ \xi_j \in \mathcal{D}\left(\left(\mathbb{C}^{2}\right)^{\otimes p_2(n)}\right)$, $j \in [k]$, such that $V_n$ accepts $(x, \otimes_{j=1}^k \xi_j)$ with probability at least $c$,
\item If $y \in A_\textup{\sc no}$, then for all quantum states $\xi_j \in \mathcal{D}\left(\left(\mathbb{C}^{2}\right)^{\otimes p_2(n)}\right)$, $j \in [k]$, we have that $V_n$ accepts $(x, \otimes_{j=1}^k \xi_j)$ with probability at most $s$.
\end{itemize}
If $c = 2/3$ and $s = 1/3$, we simply write $\QPCP[k,q]$, and use $\QPCP[q]$ if also $k = 1$.
\label{def:QPCP}
\end{definition}

\begin{subdefinition}[$\QPCP_{\textup{NA}}$] This follows the same definition as for $\QPCP$ but now with a non-adaptive $\QPCP$-verifier as per~\hyperref[def:QPCP_NA]{Definition~{3.1a}}.
\end{subdefinition}

\begin{figure}
    \centering
    \tikzset{every picture/.style={line width=0.75pt}} 

\begin{tikzpicture}[x=0.75pt,y=0.75pt,yscale=-0.85,xscale=0.85]

\draw  [fill={rgb, 255:red, 255; green, 255; blue, 255 }  ,fill opacity=1 ] (258.49,86.36) -- (419,86.36) -- (419,50.19) -- (258.49,50.19) -- cycle ;
\draw   (264.51,58.01) -- (284.55,58.01) -- (284.55,79.52) -- (264.51,79.52) -- cycle ;
\draw   (297.9,58.01) -- (317.94,58.01) -- (317.94,79.52) -- (297.9,79.52) -- cycle ;
\draw   (284.66,68.47) .. controls (285.31,66.93) and (285.94,65.45) .. (286.67,65.46) .. controls (287.39,65.46) and (288.01,66.94) .. (288.66,68.48) .. controls (289.3,70.03) and (289.92,71.51) .. (290.64,71.51) .. controls (291.37,71.51) and (292,70.04) .. (292.66,68.5) .. controls (293.31,66.96) and (293.94,65.49) .. (294.67,65.49) .. controls (295.39,65.49) and (296.01,66.97) .. (296.66,68.52) .. controls (296.99,69.31) and (297.31,70.08) .. (297.65,70.65) ;
\draw   (318.66,68.8) .. controls (319.31,67.26) and (319.94,65.79) .. (320.67,65.79) .. controls (321.39,65.79) and (322.01,67.27) .. (322.66,68.82) .. controls (323.3,70.37) and (323.92,71.84) .. (324.64,71.84) .. controls (325.37,71.85) and (326,70.38) .. (326.66,68.83) .. controls (327.31,67.29) and (327.94,65.82) .. (328.67,65.82) .. controls (329.39,65.82) and (330.01,67.3) .. (330.66,68.85) .. controls (330.99,69.64) and (331.31,70.42) .. (331.65,70.99) ;
\draw    (471.4,244.2) -- (573.5,244.09) ;
\draw    (389.8,209.5) -- (564.02,209.5) ;
\draw    (389.8,226.1) -- (442.6,226.1) ;
\draw    (389.8,234.7) -- (442.92,234.7) ;
\draw    (389.8,244.2) -- (443.4,244.2) ;
\draw    (488.5,144.99) -- (573.5,144.99) ;
\draw    (389.8,167.56) -- (573.31,167.56) ;
\draw    (389.8,174.41) -- (573.31,174.41) ;
\draw    (389.8,180.27) -- (572.88,180.27) ;
\draw    (472.57,226.1) -- (573.5,226.1) ;
\draw    (472.57,234.7) -- (573.5,234.7) ;
\draw    (389.8,253.5) -- (563.72,253.5) ;
\draw    (389.8,268.3) -- (573.31,268.3) ;
\draw    (389.8,286) -- (574.17,286) ;
\draw    (389.8,296.5) -- (563.72,296.5) ;
\draw    (389.8,305.9) -- (573.74,305.9) ;
\draw    (389.8,349) -- (574.6,349) ;
\draw    (33.5,268.29) -- (321,268.29) ;
\draw    (32.5,286) -- (319.6,285.76) -- (321,285.75) ;
\draw    (33.5,296.51) -- (321,296.51) ;
\draw    (33.5,306.28) -- (320.95,306.28) ;
\draw    (33.5,349) -- (321,349.29) ;
\draw    (33.98,244.19) -- (162.87,244.19) ;
\draw    (33.98,234.51) -- (161.7,234.51) ;
\draw    (33.5,253.5) -- (321,253.5) ;
\draw    (33.5,225.8) -- (163,225.8) ;
\draw  [fill={rgb, 255:red, 255; green, 255; blue, 255 }  ,fill opacity=1 ] (119.83,219.99) -- (156.97,219.99) -- (156.97,249.81) -- (119.83,249.81) -- cycle ;
\draw  [fill={rgb, 255:red, 255; green, 255; blue, 255 }  ,fill opacity=1 ] (42.92,218.3) -- (112.81,218.3) -- (112.81,362) -- (42.92,362) -- cycle ;
\draw  [fill={rgb, 255:red, 255; green, 255; blue, 255 }  ,fill opacity=1 ] (495.51,132.28) -- (565.4,132.28) -- (565.4,361.02) -- (495.51,361.02) -- cycle ;
\draw    (195.94,244) -- (321,244) ;
\draw    (195.47,234.75) -- (320.61,234.75) ;
\draw    (195,226.2) -- (321,226.2) ;
\draw    (236.46,209.51) -- (321,209.51) ;
\draw  [fill={rgb, 255:red, 255; green, 255; blue, 255 }  ,fill opacity=1 ] (243.31,192.89) -- (313.2,192.89) -- (313.2,362) -- (243.31,362) -- cycle ;
\draw  [fill={rgb, 255:red, 255; green, 255; blue, 255 }  ,fill opacity=1 ] (17.49,85.36) -- (178,85.36) -- (178,49.19) -- (17.49,49.19) -- cycle ;
\draw   (23.51,57.01) -- (43.55,57.01) -- (43.55,78.52) -- (23.51,78.52) -- cycle ;
\draw   (56.9,57.01) -- (76.94,57.01) -- (76.94,78.52) -- (56.9,78.52) -- cycle ;
\draw    (150.72,89.43) .. controls (160.74,109.49) and (197.12,170.83) .. (209.13,192.42) ;
\draw [shift={(210,194)}, rotate = 241.19] [color={rgb, 255:red, 0; green, 0; blue, 0 }  ][line width=0.75]    (10.93,-3.29) .. controls (6.95,-1.4) and (3.31,-0.3) .. (0,0) .. controls (3.31,0.3) and (6.95,1.4) .. (10.93,3.29)   ;
\draw    (408.67,89.17) .. controls (425.16,103.72) and (440.71,117.33) .. (457.14,130.91) ;
\draw [shift={(458.67,132.17)}, rotate = 219.47] [color={rgb, 255:red, 0; green, 0; blue, 0 }  ][line width=0.75]    (10.93,-3.29) .. controls (6.95,-1.4) and (3.31,-0.3) .. (0,0) .. controls (3.31,0.3) and (6.95,1.4) .. (10.93,3.29)   ;
\draw  [fill={rgb, 255:red, 255; green, 255; blue, 255 }  ,fill opacity=1 ] (397.06,218.8) -- (434.21,218.8) -- (434.21,248.61) -- (397.06,248.61) -- cycle ;
\draw  [fill={rgb, 255:red, 255; green, 255; blue, 255 }  ,fill opacity=1 ] (574.51,132.6) -- (639,132.6) -- (639,162.41) -- (574.51,162.41) -- cycle ;
\draw   (213.95,197.99) -- (233.99,197.99) -- (233.99,219.5) -- (213.95,219.5) -- cycle ;
\draw  [fill={rgb, 255:red, 155; green, 155; blue, 155 }  ,fill opacity=0.23 ][dash pattern={on 0.84pt off 2.51pt}] (133.9,58.01) -- (153.94,58.01) -- (153.94,79.52) -- (133.9,79.52) -- cycle ;
\draw   (121.48,70.53) .. controls (120.54,71.85) and (119.65,73.1) .. (120.02,73.72) .. controls (120.39,74.34) and (121.91,74.17) .. (123.51,73.98) .. controls (125.12,73.79) and (126.64,73.61) .. (127.01,74.23) .. controls (127.38,74.85) and (126.49,76.11) .. (125.55,77.42) .. controls (124.62,78.73) and (123.72,79.98) .. (124.09,80.61) .. controls (124.46,81.23) and (125.99,81.05) .. (127.59,80.86) .. controls (129.19,80.67) and (130.72,80.49) .. (131.08,81.12) .. controls (131.45,81.74) and (130.56,82.99) .. (129.63,84.3) .. controls (128.69,85.62) and (127.8,86.87) .. (128.17,87.49) .. controls (128.54,88.11) and (130.06,87.93) .. (131.66,87.75) .. controls (133.27,87.56) and (134.79,87.38) .. (135.16,88) .. controls (135.53,88.62) and (134.64,89.88) .. (133.7,91.19) .. controls (132.76,92.5) and (131.87,93.75) .. (132.24,94.37) .. controls (132.61,95) and (134.14,94.82) .. (135.74,94.63) .. controls (137.34,94.44) and (138.87,94.26) .. (139.23,94.88) .. controls (139.6,95.51) and (138.71,96.76) .. (137.78,98.07) .. controls (136.84,99.38) and (135.95,100.64) .. (136.32,101.26) .. controls (136.69,101.88) and (138.21,101.7) .. (139.81,101.51) .. controls (141.41,101.32) and (142.94,101.15) .. (143.31,101.77) .. controls (143.68,102.39) and (142.79,103.64) .. (141.85,104.96) .. controls (140.91,106.27) and (140.02,107.52) .. (140.39,108.14) .. controls (140.76,108.77) and (142.29,108.59) .. (143.89,108.4) .. controls (145.49,108.21) and (147.02,108.03) .. (147.38,108.65) .. controls (147.75,109.28) and (146.86,110.53) .. (145.93,111.84) .. controls (144.99,113.15) and (144.1,114.41) .. (144.47,115.03) .. controls (144.84,115.65) and (146.36,115.47) .. (147.96,115.28) .. controls (149.56,115.09) and (151.09,114.92) .. (151.46,115.54) .. controls (151.83,116.16) and (150.94,117.41) .. (150,118.73) .. controls (149.06,120.04) and (148.17,121.29) .. (148.54,121.91) .. controls (148.91,122.54) and (150.44,122.36) .. (152.04,122.17) .. controls (153.64,121.98) and (155.16,121.8) .. (155.53,122.42) .. controls (155.9,123.05) and (155.01,124.3) .. (154.07,125.61) .. controls (153.14,126.92) and (152.25,128.17) .. (152.62,128.8) .. controls (152.99,129.42) and (154.51,129.24) .. (156.11,129.05) .. controls (157.71,128.86) and (159.24,128.68) .. (159.61,129.31) .. controls (159.98,129.93) and (159.09,131.18) .. (158.15,132.49) .. controls (157.21,133.81) and (156.32,135.06) .. (156.69,135.68) .. controls (157.06,136.3) and (158.59,136.13) .. (160.19,135.94) .. controls (161.79,135.75) and (163.31,135.57) .. (163.68,136.19) .. controls (164.05,136.81) and (163.16,138.07) .. (162.22,139.38) .. controls (161.29,140.69) and (160.4,141.94) .. (160.77,142.57) .. controls (161.13,143.19) and (162.66,143.01) .. (164.26,142.82) .. controls (165.86,142.63) and (167.39,142.45) .. (167.76,143.08) .. controls (168.13,143.7) and (167.24,144.95) .. (166.3,146.26) .. controls (165.36,147.58) and (164.47,148.83) .. (164.84,149.45) .. controls (165.21,150.07) and (166.74,149.9) .. (168.34,149.71) .. controls (169.94,149.52) and (171.46,149.34) .. (171.83,149.96) .. controls (172.2,150.58) and (171.31,151.84) .. (170.37,153.15) .. controls (169.44,154.46) and (168.55,155.71) .. (168.92,156.34) .. controls (169.28,156.96) and (170.81,156.78) .. (172.41,156.59) .. controls (174.01,156.4) and (175.54,156.22) .. (175.91,156.85) .. controls (176.28,157.47) and (175.39,158.72) .. (174.45,160.03) .. controls (173.51,161.34) and (172.62,162.6) .. (172.99,163.22) .. controls (173.36,163.84) and (174.88,163.66) .. (176.49,163.47) .. controls (178.09,163.29) and (179.61,163.11) .. (179.98,163.73) .. controls (180.35,164.35) and (179.46,165.6) .. (178.52,166.92) .. controls (177.59,168.23) and (176.7,169.48) .. (177.06,170.1) .. controls (177.43,170.73) and (178.96,170.55) .. (180.56,170.36) .. controls (182.16,170.17) and (183.69,169.99) .. (184.06,170.61) .. controls (184.43,171.24) and (183.53,172.49) .. (182.6,173.8) .. controls (181.66,175.11) and (180.77,176.37) .. (181.14,176.99) .. controls (181.51,177.61) and (183.03,177.43) .. (184.64,177.24) .. controls (186.24,177.05) and (187.76,176.88) .. (188.13,177.5) .. controls (188.5,178.12) and (187.61,179.37) .. (186.67,180.69) .. controls (185.74,182) and (184.85,183.25) .. (185.21,183.87) .. controls (185.58,184.5) and (187.11,184.32) .. (188.71,184.13) .. controls (190.31,183.94) and (191.84,183.76) .. (192.21,184.38) .. controls (192.57,185.01) and (191.68,186.26) .. (190.75,187.57) .. controls (189.81,188.88) and (188.92,190.13) .. (189.29,190.76) .. controls (189.66,191.38) and (191.18,191.2) .. (192.78,191.01) .. controls (194.39,190.82) and (195.91,190.65) .. (196.28,191.27) .. controls (196.65,191.89) and (195.76,193.14) .. (194.82,194.46) .. controls (193.89,195.77) and (193,197.02) .. (193.36,197.64) .. controls (193.73,198.27) and (195.26,198.09) .. (196.86,197.9) .. controls (198.46,197.71) and (199.99,197.53) .. (200.36,198.15) .. controls (200.72,198.78) and (199.83,200.03) .. (198.9,201.34) .. controls (197.96,202.65) and (197.07,203.9) .. (197.44,204.53) .. controls (197.81,205.15) and (199.33,204.97) .. (200.93,204.78) .. controls (202.54,204.59) and (204.06,204.41) .. (204.43,205.04) .. controls (204.63,205.38) and (204.46,205.91) .. (204.1,206.53) ;
\draw   (43.66,67.47) .. controls (44.31,65.93) and (44.94,64.45) .. (45.67,64.46) .. controls (46.39,64.46) and (47.01,65.94) .. (47.66,67.48) .. controls (48.3,69.03) and (48.92,70.51) .. (49.64,70.51) .. controls (50.37,70.51) and (51,69.04) .. (51.66,67.5) .. controls (52.31,65.96) and (52.94,64.49) .. (53.67,64.49) .. controls (54.39,64.49) and (55.01,65.97) .. (55.66,67.52) .. controls (55.99,68.31) and (56.31,69.08) .. (56.65,69.65) ;
\draw   (77.66,67.8) .. controls (78.31,66.26) and (78.94,64.79) .. (79.67,64.79) .. controls (80.39,64.79) and (81.01,66.27) .. (81.66,67.82) .. controls (82.3,69.37) and (82.92,70.84) .. (83.64,70.84) .. controls (84.37,70.85) and (85,69.38) .. (85.66,67.83) .. controls (86.31,66.29) and (86.94,64.82) .. (87.67,64.82) .. controls (88.39,64.82) and (89.01,66.3) .. (89.66,67.85) .. controls (89.99,68.64) and (90.31,69.42) .. (90.65,69.99) ;
\draw   (222.82,190.2) .. controls (220.98,190.32) and (219.23,190.43) .. (218.93,189.77) .. controls (218.63,189.11) and (219.86,187.87) .. (221.16,186.56) .. controls (222.45,185.25) and (223.68,184) .. (223.38,183.34) .. controls (223.08,182.69) and (221.33,182.8) .. (219.49,182.92) .. controls (217.66,183.04) and (215.91,183.15) .. (215.61,182.5) .. controls (215.31,181.84) and (216.54,180.59) .. (217.83,179.28) .. controls (219.13,177.97) and (220.36,176.73) .. (220.06,176.07) .. controls (219.76,175.41) and (218.01,175.52) .. (216.17,175.64) .. controls (214.34,175.76) and (212.59,175.88) .. (212.29,175.22) .. controls (211.98,174.56) and (213.22,173.31) .. (214.51,172) .. controls (215.8,170.7) and (217.03,169.45) .. (216.73,168.79) .. controls (216.43,168.13) and (214.68,168.24) .. (212.85,168.37) .. controls (211.01,168.49) and (209.26,168.6) .. (208.96,167.94) .. controls (208.66,167.28) and (209.89,166.04) .. (211.18,164.73) .. controls (212.48,163.42) and (213.71,162.17) .. (213.41,161.51) .. controls (213.11,160.85) and (211.36,160.97) .. (209.52,161.09) .. controls (207.69,161.21) and (205.94,161.32) .. (205.64,160.66) .. controls (205.34,160.01) and (206.57,158.76) .. (207.86,157.45) .. controls (209.15,156.14) and (210.39,154.89) .. (210.08,154.24) .. controls (209.78,153.58) and (208.03,153.69) .. (206.2,153.81) .. controls (204.36,153.93) and (202.61,154.05) .. (202.31,153.39) .. controls (202.01,152.73) and (203.24,151.48) .. (204.54,150.17) .. controls (205.83,148.87) and (207.06,147.62) .. (206.76,146.96) .. controls (206.46,146.3) and (204.71,146.41) .. (202.88,146.54) .. controls (201.04,146.66) and (199.29,146.77) .. (198.99,146.11) .. controls (198.69,145.45) and (199.92,144.21) .. (201.21,142.9) .. controls (202.51,141.59) and (203.74,140.34) .. (203.44,139.68) .. controls (203.14,139.02) and (201.39,139.14) .. (199.55,139.26) .. controls (197.72,139.38) and (195.97,139.49) .. (195.67,138.83) .. controls (195.37,138.18) and (196.6,136.93) .. (197.89,135.62) .. controls (199.18,134.31) and (200.41,133.06) .. (200.11,132.41) .. controls (199.81,131.75) and (198.06,131.86) .. (196.23,131.98) .. controls (194.39,132.1) and (192.64,132.22) .. (192.34,131.56) .. controls (192.04,130.9) and (193.27,129.65) .. (194.57,128.34) .. controls (195.86,127.04) and (197.09,125.79) .. (196.79,125.13) .. controls (196.49,124.47) and (194.74,124.58) .. (192.9,124.71) .. controls (191.07,124.83) and (189.32,124.94) .. (189.02,124.28) .. controls (188.72,123.62) and (189.95,122.37) .. (191.24,121.07) .. controls (192.54,119.76) and (193.77,118.51) .. (193.46,117.85) .. controls (193.16,117.19) and (191.42,117.31) .. (189.58,117.43) .. controls (187.74,117.55) and (186,117.66) .. (185.69,117) .. controls (185.39,116.35) and (186.62,115.1) .. (187.92,113.79) .. controls (189.21,112.48) and (190.44,111.23) .. (190.14,110.58) .. controls (189.84,109.92) and (188.09,110.03) .. (186.26,110.15) .. controls (184.42,110.27) and (182.67,110.39) .. (182.37,109.73) .. controls (182.07,109.07) and (183.3,107.82) .. (184.59,106.51) .. controls (185.89,105.21) and (187.12,103.96) .. (186.82,103.3) .. controls (186.52,102.64) and (184.77,102.75) .. (182.93,102.87) .. controls (181.1,103) and (179.35,103.11) .. (179.05,102.45) .. controls (178.75,101.79) and (179.98,100.54) .. (181.27,99.24) .. controls (182.56,97.93) and (183.79,96.68) .. (183.49,96.02) .. controls (183.19,95.36) and (181.44,95.48) .. (179.61,95.6) .. controls (177.77,95.72) and (176.02,95.83) .. (175.72,95.17) .. controls (175.42,94.52) and (176.65,93.27) .. (177.95,91.96) .. controls (179.24,90.65) and (180.47,89.4) .. (180.17,88.75) .. controls (179.87,88.09) and (178.12,88.2) .. (176.28,88.32) .. controls (174.45,88.44) and (172.7,88.56) .. (172.4,87.9) .. controls (172.1,87.24) and (173.33,85.99) .. (174.62,84.68) .. controls (175.92,83.37) and (177.15,82.13) .. (176.85,81.47) .. controls (176.54,80.81) and (174.8,80.92) .. (172.96,81.04) .. controls (171.12,81.17) and (169.38,81.28) .. (169.07,80.62) .. controls (168.77,79.96) and (170,78.71) .. (171.3,77.41) .. controls (172.59,76.1) and (173.82,74.85) .. (173.52,74.19) .. controls (173.22,73.53) and (171.47,73.65) .. (169.64,73.77) .. controls (167.8,73.89) and (166.05,74) .. (165.75,73.34) .. controls (165.45,72.69) and (166.68,71.44) .. (167.97,70.13) .. controls (168.42,69.68) and (168.86,69.23) .. (169.23,68.81) ;
\draw  [fill={rgb, 255:red, 155; green, 155; blue, 155 }  ,fill opacity=0.23 ][dash pattern={on 0.84pt off 2.51pt}] (384.24,58.02) -- (404.28,58.02) -- (404.28,79.53) -- (384.24,79.53) -- cycle ;
\draw   (464.61,134.83) -- (484.65,134.83) -- (484.65,156.33) -- (464.61,156.33) -- cycle ;
\draw   (367.93,72.02) .. controls (369.46,71.11) and (370.92,70.25) .. (371.49,70.71) .. controls (372.05,71.16) and (371.56,72.81) .. (371.04,74.53) .. controls (370.53,76.26) and (370.04,77.91) .. (370.6,78.36) .. controls (371.17,78.81) and (372.63,77.95) .. (374.16,77.04) .. controls (375.69,76.13) and (377.15,75.27) .. (377.72,75.72) .. controls (378.28,76.18) and (377.79,77.82) .. (377.27,79.55) .. controls (376.76,81.28) and (376.27,82.92) .. (376.83,83.38) .. controls (377.4,83.83) and (378.86,82.97) .. (380.39,82.06) .. controls (381.92,81.15) and (383.38,80.29) .. (383.95,80.74) .. controls (384.51,81.2) and (384.02,82.84) .. (383.5,84.57) .. controls (382.99,86.3) and (382.5,87.94) .. (383.06,88.4) .. controls (383.63,88.85) and (385.09,87.99) .. (386.62,87.08) .. controls (388.15,86.17) and (389.61,85.31) .. (390.18,85.76) .. controls (390.74,86.21) and (390.25,87.86) .. (389.74,89.59) .. controls (389.22,91.31) and (388.73,92.96) .. (389.29,93.41) .. controls (389.86,93.87) and (391.32,93.01) .. (392.85,92.1) .. controls (394.38,91.19) and (395.84,90.32) .. (396.41,90.78) .. controls (396.97,91.23) and (396.48,92.88) .. (395.97,94.61) .. controls (395.45,96.33) and (394.96,97.98) .. (395.52,98.43) .. controls (396.09,98.89) and (397.55,98.02) .. (399.08,97.11) .. controls (400.61,96.21) and (402.07,95.34) .. (402.64,95.8) .. controls (403.2,96.25) and (402.71,97.9) .. (402.2,99.62) .. controls (401.68,101.35) and (401.19,103) .. (401.75,103.45) .. controls (402.32,103.91) and (403.78,103.04) .. (405.31,102.13) .. controls (406.84,101.22) and (408.3,100.36) .. (408.87,100.82) .. controls (409.43,101.27) and (408.94,102.92) .. (408.43,104.64) .. controls (407.91,106.37) and (407.42,108.02) .. (407.98,108.47) .. controls (408.55,108.92) and (410.01,108.06) .. (411.54,107.15) .. controls (413.07,106.24) and (414.53,105.38) .. (415.1,105.83) .. controls (415.66,106.29) and (415.17,107.93) .. (414.66,109.66) .. controls (414.14,111.39) and (413.65,113.03) .. (414.22,113.49) .. controls (414.78,113.94) and (416.24,113.08) .. (417.77,112.17) .. controls (419.3,111.26) and (420.76,110.4) .. (421.33,110.85) .. controls (421.89,111.31) and (421.4,112.95) .. (420.89,114.68) .. controls (420.37,116.41) and (419.88,118.05) .. (420.45,118.51) .. controls (421.01,118.96) and (422.47,118.1) .. (424,117.19) .. controls (425.53,116.28) and (427,115.42) .. (427.56,115.87) .. controls (428.12,116.32) and (427.63,117.97) .. (427.12,119.7) .. controls (426.6,121.42) and (426.11,123.07) .. (426.68,123.52) .. controls (427.24,123.98) and (428.7,123.12) .. (430.23,122.21) .. controls (431.76,121.3) and (433.23,120.43) .. (433.79,120.89) .. controls (434.35,121.34) and (433.86,122.99) .. (433.35,124.72) .. controls (432.83,126.44) and (432.34,128.09) .. (432.91,128.54) .. controls (433.47,129) and (434.93,128.13) .. (436.46,127.22) .. controls (437.99,126.32) and (439.46,125.45) .. (440.02,125.91) .. controls (440.58,126.36) and (440.09,128.01) .. (439.58,129.73) .. controls (439.06,131.46) and (438.57,133.11) .. (439.14,133.56) .. controls (439.7,134.01) and (441.16,133.15) .. (442.69,132.24) .. controls (444.23,131.33) and (445.69,130.47) .. (446.25,130.92) .. controls (446.81,131.38) and (446.32,133.02) .. (445.81,134.75) .. controls (445.29,136.48) and (444.8,138.13) .. (445.37,138.58) .. controls (445.93,139.03) and (447.39,138.17) .. (448.92,137.26) .. controls (450.46,136.35) and (451.92,135.49) .. (452.48,135.94) .. controls (453.04,136.4) and (452.55,138.04) .. (452.04,139.77) .. controls (451.52,141.5) and (451.03,143.14) .. (451.6,143.6) .. controls (452.16,144.05) and (453.62,143.19) .. (455.15,142.28) .. controls (456.69,141.37) and (458.15,140.51) .. (458.71,140.96) .. controls (459.17,141.34) and (458.93,142.52) .. (458.53,143.89) ;
\draw   (412.74,68.85) .. controls (414.43,68.18) and (416.04,67.55) .. (416.52,68.09) .. controls (417,68.63) and (416.22,70.2) .. (415.4,71.84) .. controls (414.58,73.48) and (413.8,75.04) .. (414.29,75.59) .. controls (414.77,76.13) and (416.37,75.49) .. (418.06,74.83) .. controls (419.74,74.16) and (421.35,73.53) .. (421.83,74.07) .. controls (422.31,74.61) and (421.54,76.18) .. (420.72,77.82) .. controls (419.9,79.46) and (419.12,81.02) .. (419.6,81.56) .. controls (420.08,82.1) and (421.69,81.47) .. (423.37,80.81) .. controls (425.06,80.14) and (426.67,79.51) .. (427.15,80.05) .. controls (427.63,80.59) and (426.85,82.16) .. (426.03,83.8) .. controls (425.21,85.44) and (424.44,87) .. (424.92,87.54) .. controls (425.4,88.08) and (427,87.45) .. (428.69,86.79) .. controls (430.38,86.12) and (431.98,85.49) .. (432.46,86.03) .. controls (432.94,86.57) and (432.17,88.13) .. (431.35,89.78) .. controls (430.53,91.42) and (429.75,92.98) .. (430.23,93.52) .. controls (430.71,94.06) and (432.32,93.43) .. (434.01,92.76) .. controls (435.69,92.1) and (437.3,91.47) .. (437.78,92.01) .. controls (438.26,92.55) and (437.48,94.11) .. (436.66,95.75) .. controls (435.84,97.39) and (435.07,98.96) .. (435.55,99.5) .. controls (436.03,100.04) and (437.64,99.41) .. (439.32,98.74) .. controls (441.01,98.08) and (442.61,97.45) .. (443.09,97.99) .. controls (443.58,98.53) and (442.8,100.09) .. (441.98,101.73) .. controls (441.16,103.37) and (440.38,104.94) .. (440.86,105.48) .. controls (441.34,106.02) and (442.95,105.39) .. (444.64,104.72) .. controls (446.32,104.06) and (447.93,103.43) .. (448.41,103.97) .. controls (448.89,104.51) and (448.11,106.07) .. (447.29,107.71) .. controls (446.48,109.35) and (445.7,110.92) .. (446.18,111.46) .. controls (446.66,112) and (448.27,111.37) .. (449.95,110.7) .. controls (451.64,110.04) and (453.24,109.4) .. (453.73,109.94) .. controls (454.21,110.49) and (453.43,112.05) .. (452.61,113.69) .. controls (451.79,115.33) and (451.01,116.9) .. (451.49,117.44) .. controls (451.98,117.98) and (453.58,117.35) .. (455.27,116.68) .. controls (456.95,116.01) and (458.56,115.38) .. (459.04,115.92) .. controls (459.52,116.46) and (458.74,118.03) .. (457.93,119.67) .. controls (457.11,121.31) and (456.33,122.87) .. (456.81,123.41) .. controls (457.29,123.96) and (458.9,123.32) .. (460.58,122.66) .. controls (462.27,121.99) and (463.88,121.36) .. (464.36,121.9) .. controls (464.84,122.44) and (464.06,124.01) .. (463.24,125.65) .. controls (462.42,127.29) and (461.64,128.85) .. (462.13,129.39) .. controls (462.48,129.79) and (463.44,129.56) .. (464.59,129.14) ;
\draw  [fill={rgb, 255:red, 255; green, 255; blue, 255 }  ,fill opacity=1 ] (499.49,86.36) -- (660,86.36) -- (660,50.19) -- (499.49,50.19) -- cycle ;
\draw   (505.51,58.01) -- (525.55,58.01) -- (525.55,79.52) -- (505.51,79.52) -- cycle ;
\draw   (538.9,58.01) -- (558.94,58.01) -- (558.94,79.52) -- (538.9,79.52) -- cycle ;
\draw   (525.66,68.47) .. controls (526.31,66.93) and (526.94,65.45) .. (527.67,65.46) .. controls (528.39,65.46) and (529.01,66.94) .. (529.66,68.48) .. controls (530.3,70.03) and (530.92,71.51) .. (531.64,71.51) .. controls (532.37,71.51) and (533,70.04) .. (533.66,68.5) .. controls (534.31,66.96) and (534.94,65.49) .. (535.67,65.49) .. controls (536.39,65.49) and (537.01,66.97) .. (537.66,68.52) .. controls (537.99,69.31) and (538.31,70.08) .. (538.65,70.65) ;
\draw   (559.66,68.8) .. controls (560.31,67.26) and (560.94,65.79) .. (561.67,65.79) .. controls (562.39,65.79) and (563.01,67.27) .. (563.66,68.82) .. controls (564.3,70.37) and (564.92,71.84) .. (565.64,71.84) .. controls (566.37,71.85) and (567,70.38) .. (567.66,68.83) .. controls (568.31,67.29) and (568.94,65.82) .. (569.67,65.82) .. controls (570.39,65.82) and (571.01,67.3) .. (571.66,68.85) .. controls (571.99,69.64) and (572.31,70.42) .. (572.65,70.99) ;
\draw   (632.9,57.01) -- (652.94,57.01) -- (652.94,78.52) -- (632.9,78.52) -- cycle ;
\draw   (619.66,67.47) .. controls (620.31,65.93) and (620.94,64.45) .. (621.67,64.46) .. controls (622.39,64.46) and (623.01,65.94) .. (623.66,67.48) .. controls (624.3,69.03) and (624.92,70.51) .. (625.64,70.51) .. controls (626.37,70.51) and (627,69.04) .. (627.66,67.5) .. controls (628.31,65.96) and (628.94,64.49) .. (629.67,64.49) .. controls (630.39,64.49) and (631.01,65.97) .. (631.66,67.52) .. controls (631.99,68.31) and (632.31,69.08) .. (632.65,69.65) ;
\draw [line width=1.5]  [dash pattern={on 1.69pt off 2.76pt}]  (100,69) -- (115.4,69) ;
\draw [line width=1.5]  [dash pattern={on 1.69pt off 2.76pt}]  (343.8,69) -- (359.2,69) ;
\draw [line width=1.5]  [dash pattern={on 1.69pt off 2.76pt}]  (209.6,69) -- (225,69) ;
\draw [line width=1.5]  [dash pattern={on 1.69pt off 2.76pt}]  (452.8,69) -- (468.2,69) ;
\draw [line width=1.5]  [dash pattern={on 1.69pt off 2.76pt}]  (588,69) -- (603.4,69) ;
\draw [line width=1.5]  [dash pattern={on 1.69pt off 2.76pt}]  (346.2,271.6) -- (361.6,271.6) ;

\draw (68.62,273.91) node [anchor=north west][inner sep=0.75pt]    {$V^{1}$};
\draw (270.05,275.05) node [anchor=north west][inner sep=0.75pt]    {$V^{2}$};
\draw (514.68,274.78) node [anchor=north west][inner sep=0.75pt]    {$V^{q+1}$};
\draw (164.42,225.59) node [anchor=north west][inner sep=0.75pt]    {$|i_{1} \rangle \ $};
\draw (24,61) node [anchor=north west][inner sep=0.75pt]  [font=\tiny]  {${\displaystyle \xi _{11} \ }$};
\draw (268.58,20.67) node [anchor=north west][inner sep=0.75pt]    {$\textup{Proof } \xi =\otimes _{j=1}^{k} \xi _{j}$};
\draw (85.6,137.28) node [anchor=north west][inner sep=0.75pt]    {``query''};
\draw (443.13,224.99) node [anchor=north west][inner sep=0.75pt]    {$|i_{q} \rangle \ $};
\draw (591.71,168.11) node [anchor=north west][inner sep=0.75pt]    {$ \begin{array}{l}
\textup{Accept}/\\
\textup{Reject}
\end{array}$};
\draw (10,224) node [anchor=north west][inner sep=0.75pt]    {$| x\rangle $};
\draw (10,285) node [anchor=north west][inner sep=0.75pt]    {$| 0 \rangle $};
\draw (135,61) node [anchor=north west][inner sep=0.75pt]  [font=\tiny]  {${\displaystyle \xi _{i_{1}}}$};
\draw (216,199.99) node [anchor=north west][inner sep=0.75pt]  [font=\tiny]  {${\displaystyle \xi _{i_{1}}}$};
\draw (183,60) node [anchor=north west][inner sep=0.75pt]    {$\otimes $};
\draw (236,60) node [anchor=north west][inner sep=0.75pt]    {$\otimes $};
\draw (426,60) node [anchor=north west][inner sep=0.75pt]    {$\otimes $};
\draw (481,60) node [anchor=north west][inner sep=0.75pt]    {$\otimes $};
\draw (57.7,61) node [anchor=north west][inner sep=0.75pt]  [font=\tiny]  {${\displaystyle \xi _{12} \ }$};
\draw (265,61) node [anchor=north west][inner sep=0.75pt]  [font=\tiny]  {${\displaystyle \xi _{j1} \ }$};
\draw (298.51,61) node [anchor=north west][inner sep=0.75pt]  [font=\tiny]  {${\displaystyle \xi _{j2} \ }$};
\draw (386,61) node [anchor=north west][inner sep=0.75pt]  [font=\tiny]  {${\displaystyle \xi _{i_{q}}}$};
\draw (505.51,61) node [anchor=north west][inner sep=0.75pt]  [font=\tiny]  {${\displaystyle \xi _{k1}}$};
\draw (466,138) node [anchor=north west][inner sep=0.75pt]  [font=\tiny]  {${\displaystyle \xi _{i_{q}}}$};
\draw (539,61) node [anchor=north west][inner sep=0.75pt]  [font=\tiny]  {${\displaystyle \xi _{k2}}$};
\draw (633,61) node [anchor=north west][inner sep=0.75pt]  [font=\tiny]  {${\displaystyle \xi _{kp}}$};
\draw (127.8,226.4) node [anchor=north west][inner sep=0.75pt]    {$\Pi ^{1}$};
\draw (405.2,226.4) node [anchor=north west][inner sep=0.75pt]    {$\Pi ^{q}$};
\draw (580,138) node [anchor=north west][inner sep=0.75pt]    {$\Pi ^{\mathrm{output}}$};
\end{tikzpicture}

    \caption{Visualisation of a $k$-prover quantum PCP which allows for $q$ adaptive queries to the provided quantum proof through intermediate measurements as per~\cref{def:adap_QPCP}. Note that an index measurement result $i_t$ denotes a tuple ($j,l$) where $j$ indicates the corresponding proof and $l$ the index of the qubit in this proof.}
    \label{fig:gen_QPCP}
\end{figure}

\noindent We can easily show that in the single-prover setting we have weak error reduction by parallel repetition by using a similar argument as in~\cite{aharonov2002quantum}.

\begin{lemma}[Weak error reduction for the single prover case] 
Let $c - s \in \Omega(1)$. Then
\begin{align*}
    \QPCP[1,q,c,s] \subseteq \QPCP[1, \mO(qt), 1 - 2^t, 2^{-t}].
\end{align*}
\label{lemma:weak_err_red}
\end{lemma} 

\begin{proof} 
This follows from a standard parallel repetition argument, with special care given to the fact that the proof can be entangled. Let $V$ be the $(k,q,p_1,p_2,p_3)$-$\QPCP$ verifier with completeness $c$ and soundness $s$, where $p_1, p_2, p_3 \in \poly(n)$. Arthur expects to receive the proof $\xi^{\otimes R}$, runs $V$ $R$ times in parallel (acting only on $q$ qubits of each $\xi$), measures the output qubit, and accepts if at least a $(c+s)/2$-fraction of the outcomes are accepting. Completeness follows directly from a Chernoff bound. If Merlin provides the correct proof $\xi^{\otimes R}$, then each run of the verifier accepts with probability at least $c$. Let $X_i \in \{0,1\}$ be the random variable that indicates whether the $i$th run of the parallel repetition accepted ($X_i = 1$) or not ($X_i = 0$). Let $\{X_i\}_{i \in [R]}$ be the outcomes of the $R$ runs, with $\mu = \mathbb{E}[X_1] = c$. The total number of accepting runs is given by $S_R = \sum_{i=1}^R X_i$. By the Chernoff bound, the probability that fewer than a $(c+s)/2$-fraction of the runs accept is given by
\[
\Pr\left[S_R < \frac{c+s}{2} \cdot R\right] \leq \exp\left(-2R\left(\frac{c - s}{2}\right)^2\right).
\]
To ensure that $\Pr[S_R < (c+s)/2 \cdot R] \leq 2^{-t}$, it suffices to choose
\begin{align}
    R \coloneqq  \left\lceil \frac{2 t \ln 2}{\left(c - s\right)^2} \right\rceil.
    \label{eq:choice_of_R_repeat}
\end{align}

For soundness, let $\rho$ be the $p_2(n)t$-qubit proof that Merlin provides instead of $\xi^{\otimes t}$. From the soundness property of the verifier, we know that $\mathbb{E}[X_1] \leq s$. Now consider the expectation of $X_2$, which depends on the outcome of $X_1$. However, the soundness condition ensures that $\mathbb{E}[X_2 \mid X_1 = j] \leq s$ for all possible outcomes $j \in \{0,1\}$. By repeating this argument, we see that for any $i$, $\mathbb{E}[X_i \mid X_1, \ldots, X_{i-1}] \leq s$. Since this holds for any $i$, we can upper bound the acceptance probability by polynomially many independent Bernoulli trials with mean $\mu = s$, again with bias $(c-s)/2$. Applying a Chernoff bound for dependent variables (with bounded conditional expectations), we find that the acceptance probability is at most $2^{-t}$. Finally, the total number of queries to the proof is $qR = \mO(q t)$.
\end{proof}

We leave open the question of weak error reduction for the multiple-prover case, but we expect that it can be done using the ideas in~{\cite{harrow2013testing}}.

When we have near-perfect completeness, strong error reduction is also possible for non-adaptive quantum PCPs.
\begin{restatable}[Strong error reduction for non-adaptive QPCPs with near-perfect completeness]{claim}{strongerrredu} For $l \in \mO(1)$ it holds that 
\begin{align*}
    \QPCP_{\textup{NA}}[1,q,c,s] = \QPCP_{\textup{NA}}[1,l q,c',s']
\end{align*}
with $c= 1-2^{-\Omega(n)}$, $s=1/2$ and $c'= 1-2^{-\Omega(n)}$, $s' = 2^{-\mO(l)}$.
\label{claim:strong_err_red}
\end{restatable}
The proof can be found in~\cref{app:strong_err_red}.

\section{Local Hamiltonians from quantum PCPs}
\label{sec:LH_from_QPCPs}

The core of this section is~\cref{thm:red}, which argues that we can, from a $\QPCP$ verification circuit (as in~\cref{def:adap_QPCP}) efficiently produce a local Hamiltonian, such that the expectation value of a proof state is given by its acceptance probability by the verifier. Our construction broadly follows the ideas in~\cite{Grilo2018thesis}, but also includes the more general adaptive verifier case. This leads to a slightly more general class of local Hamiltonians than is usually assumed in the quantum PCP literature, and we have to perform some extra tricks to obtain the traditional Hamiltonian. \\

We begin by proving a basic lemma, which expresses the probability that a proof $\xi$ gets rejected by the quantum PCP verifier $V$ conditioned on taking the query path $(i_1,\dots,i_q)$, in terms of the PVMs and circuits of $V$. Throughout this work, we make a distinction between indices indicated surrounded by brackets (e.g.~``$(i_1,\dots,i_q)$'') and those that are not (e.g.~``$i_1,\dots,i_q$'') to make a distinction where the order does matter (the former) and where it does not (the latter).

\begin{lemma} 
Let $V_x$ be a $(k,q,p_1,p_2,p_3)$-$\QPCP$ verifier as in~\cref{def:adap_QPCP}, with hardcoded input $x \in \{0,1\}^n$. Define $M_{i_q,x}^{t'} = \Pi_{i_q}^{t'} V_x^{t'}$ for all $i_q \in [k p_2(n)]$, $t' \in [q]$. The probability that the quantum PCP rejects a proof $\xi$, conditioned on taking the query path $(i_1, \dots, i_q)$, is given by
\begin{align*}
    \Pr[V_x \text{ rejects } \xi | (i_1, \dots, i_q)] = \frac{\tr[ P_{x,(i_1, \dots, i_q)} \rho^0 ]}{\Pr[(i_1, \dots, i_q)]},
\end{align*}
where $\Pr[(i_1, \dots, i_q)]$ is the probability that $i_1, \dots, i_q$ are queried (and in this order), $\rho^{0} = \ket{0}\bra{0}^{\otimes n + p_1(n)} \otimes \xi$, and 
$P_{x,(i_1, \dots, i_q)}$ is a $(k + n + p_1(n))$-local operator given by
\begin{align}
     P_{x,(i_1, \dots, i_q)} = M_{i_1,x}^{1\dagger} \dots M_{i_q,x}^{q\dagger} V^{q+1\dagger}_x \Pi_{0}^{\textup{output}} V^{q+1}_x M_{i_q,x}^{q} \dots M_{i_1,x}^{1}.
     \label{eq:P_eq}
\end{align}
\label{lem:P_rej_cond}
\end{lemma}

\begin{proof}
    This follows from a straightforward application of the rules for post-measurement states in projective measurements. Let $\rho^{0} = \ket{0}\bra{0}^{\otimes n + p_1(n)} \otimes \xi$ be the initial state (note the extra $n$ term for the original input register). 

    Suppose the first PVM of the quantum PCP returns outcome $i_1$. The post-measurement state after this step is:
\begin{align*}
    \rho^1 = \frac{\Pi_{i_1}^1 V_x^1 \rho^0 V^{1\dagger}_x \Pi_{i_1}^1}{\tr[\Pi_{i_1}^1 V_x^1 \rho^0 V_x^{1\dagger}]} = \frac{\Pi_{i_1}^1 V_x^1 \rho^0 V_x^{1\dagger}\Pi^1_{i_1}}{\Pr[i_1]} = \frac{M_{i_1,x}^1 \rho^0 M_{i_1,x}^{1\dagger}}{\Pr[i_1]}.
\end{align*}
Similarly, assuming outcome $i_2$ for the second query, the state becomes:
\begin{align*}
    \rho^2 = \frac{\Pi_{i_2}^2 V_x^2 \rho^1 V_x^{2\dagger}\Pi_{i_2}^2}{\tr[\Pi_{i_2}^2 V_x^2 \rho^1 V_x^{2\dagger}]} = \frac{\Pi_{i_2}^2 V_x^2 \rho^1 V_x^{2\dagger}\Pi_{i_2}^2}{\Pr[i_2|i_1]} = \frac{M_{i_2,x}^{2} \rho^1 M_{i_2,x}^{2\dagger}}{\Pr[i_2|i_1]} = \frac{M_{i_2,x}^{2} M_{i_1,x}^{1} \rho^0 M_{i_1,x}^{1\dagger} M_{i_2,x}^{2\dagger}}{\Pr[i_2|i_1] \Pr[i_1]}.
\end{align*}
Repeating this procedure $q-2$ more times, assuming outcomes $i_3, \dots, i_q$, we find that the state after the $q$'th query becomes
\begin{align*}
    \rho^q = \frac{M_{i_q,x}^{q} \dots M_{i_1,x}^{1} \rho^0 M_{i_1,x}^{1\dagger} \dots M_{i_q,x}^{q\dagger}}{\prod_{l=1}^q \Pr[i_l | (i_1, \dots, i_{l-1})]} = \frac{M_{i_q,x}^{q} \dots M_{i_1,x}^{1} \rho^0 M_{i_1,x}^{1\dagger} \dots M_{i_q,x}^{q\dagger}}{\Pr[(i_1, \dots, i_q)]}.
\end{align*}
Now in the final step of the quantum PCP, a final circuit $V_l$ is applied, followed by the PVM $\Pi^\textup{output}$. The expected value of rejection is then given by
\begin{align*}
\Pr[V_x \text{ rejects } \xi | (i_1, \dots, i_q)] = \tr[\Pi_{0}^{\textup{output}} V^{q}_x \rho^q V^{q\dagger}_x] = \frac{\tr[\Pi_0 V^{q}_x M_{i_q,x}^{q} \dots M_{i_1,x}^{1} \rho^0 M_{i_1,x}^{1\dagger} \dots M_{i_q,x}^{q\dagger} V^{q\dagger}_x]}{\Pr[(i_1, \dots, i_q)]}.
\end{align*}
Using the cyclic property of the trace, we can write:
\begin{align*}
    \Pr[V_x \text{ rejects } \xi | (i_1, \dots, i_q)] &= \frac{\tr[M_{i_1,x}^{1\dagger} \dots M_{i_q,x}^{q\dagger} V^{q\dagger}_x \Pi_{0}^{\textup{output}} V^{q}_x M_{i_q,x}^{q} \dots M_{i_1,x}^{1} \rho^0]}{\Pr[(i_1, \dots, i_q)]} \\
    &= \frac{\tr[ P_{x,(i_1, \dots, i_q)} \rho^0]}{\Pr[(i_1, \dots, i_q)]},
\end{align*}
with
\begin{align*}
     P_{x,(i_1, \dots, i_q)} = M_{i_1,x}^{1\dagger} \dots M_{i_q,x}^{q\dagger} V^{q\dagger}_x \Pi_{0}^{\textup{output}} V^{q}_x M_{i_q,x}^{q} \dots M_{i_1,x}^{1}.
\end{align*}
\end{proof}

The next idea is that the expectation value of an operator $A$ acting on an $n$-qubit state consisting of a product state of a $q$-qubit state and a fixed $(n-q)$-qubit state can be represented as an expectation value of a $q$-qubit operator $B$ acting only on the $q$-qubit state. The following lemma proves this and gives an explicit expression of the local operator, assuming that the fixed state is pure.

\begin{lemma} 
Let $A$ be an operator acting on an $n$-qubit Hilbert space consisting of a variable $q$-qubit input state $\rho$ in a tensor product with some fixed $(n-q)$-qubit state $\rho_{\textup{fixed}}$. Then we have that
\begin{align*}
    \tr[A (\rho \otimes \rho_{\textup{fixed}})] = \tr[B \rho],
\end{align*}
where $B = B(\rho_{\textup{fixed}})$ is some $q$-local operator which depends on $\rho_{\textup{fixed}}$. Moreover, if $\rho_{\textup{fixed}} = \ketbra{\psi}$ for some pure state $\ket{\psi}$, we have that the ($\alpha,\alpha'$)-entry of $B$ in some basis $\{\alpha\}$ is given by
\begin{align*}
   b_{\alpha,\alpha'} =  \bra{\alpha}\bra{\psi} A \ket{\alpha'}\ket{\psi}.
\end{align*}
\label{lem:AtoB}
\end{lemma}
\begin{proof}
 We can decompose $A$ in two arbitrary bases $\{\alpha\}$ and $\{\beta\}$ for each part of the cut in the product state as
\begin{align*}
    A = \sum_{\alpha,\alpha',\beta,\beta'} a_{\alpha,\alpha',\beta,\beta'} \ketbra{\alpha}{\alpha'} \otimes \ketbra{\beta}{\beta'}.
\end{align*}
Using the linearity of the trace and the tensor product property,
\begin{align*}
    \tr[A (\rho \otimes \rho_{\textup{fixed}})] &=  \tr\left[\sum_{\alpha,\alpha',\beta,\beta'} a_{\alpha,\alpha',\beta,\beta'} 
    \ketbra{\alpha}{\alpha'} \otimes \ketbra{\beta}{\beta'} ( \rho \otimes \rho_{\textup{fixed}} )\right]\\
    &= \sum_{\alpha,\alpha',\beta,\beta'} a_{\alpha,\alpha',\beta,\beta'}  \tr[
    \ketbra{\alpha}{\alpha'} \rho ] \tr[ \ketbra{\beta}{\beta'} \rho_{\textup{fixed}}]\\
    &= \sum_{\alpha,\alpha'}  \left(\sum_{\beta,\beta'} a_{\alpha,\alpha',\beta,\beta'} \tr[ \ketbra{\beta}{\beta'} \rho_{\textup{fixed}}]\right) \tr[
    \ketbra{\alpha}{\alpha'} \rho ] \\
    &= \sum_{\alpha,\alpha'} b_{\alpha,\alpha'} \tr[
    \ketbra{\alpha}{\alpha'} \rho ] \\
    &= \tr[B \rho],
\end{align*}
where the operator $B$, given by
\begin{align*}
    B = \sum_{\alpha,\alpha'} b_{\alpha,\alpha'} \ketbra{\alpha}{\alpha'}, \quad b_{\alpha,\alpha'} = \sum_{\beta,\beta'} a_{\alpha,\alpha',\beta,\beta'} \tr[ \ketbra{\beta}{\beta'} \rho_{\textup{fixed}}],
\end{align*}
is indeed $q$-local.

For the second part of the lemma, we note that under the assumption that $\rho_{\textup{fixed}} = \ketbra{\psi}$ for some pure state $\ket{\psi}$, we have
\begin{align*}
\bra{\alpha}\bra{\psi} A \ket{\alpha'}\ket{\psi} &= \bra{\alpha}\bra{\psi} \left(\sum_{\alpha,\alpha',\beta,\beta'} a_{\alpha,\alpha',\beta,\beta'} \ketbra{\alpha}{\alpha'} \otimes \ketbra{\beta}{\beta'} \right)\ket{\alpha'}\ket{\psi}\\
&= \sum_{\beta,\beta'} a_{\alpha,\alpha',\beta,\beta'} \bra{\psi} \ketbra{\beta}{\beta'} \ket{\psi}\\
&= \sum_{\beta,\beta'} a_{\alpha,\alpha',\beta,\beta'} \tr[\ketbra{\beta}{\beta'} \ketbra{\psi}]\\
&= \sum_{\beta,\beta'} a_{\alpha,\alpha',\beta,\beta'} \tr[ \ketbra{\beta}{\beta'} \rho_{\textup{fixed}}]\\
&= b_{\alpha,\alpha'},
\end{align*} 
completing the proof.
\end{proof}

With these lemmas in hand we can argue that, given a verifier, there always exists a local Hamiltonian that captures the probability of acceptance of a proof.

\begin{lemma}[Hamiltonians from general quantum PCPs]
Let $q \in \mathbb{N}$ be some constant and let $p_1$, $p_2$, and $p_3$ be polynomials. Let $V_x$ be a $(k,q,p_1,p_2,p_3)\text{-}\QPCP$-verifier as in~\cref{def:adap_QPCP}, with hardcoded input $x$, $|x| = n$ for some $n \in \mathbb{N}$, and a $k p_2(n)$-qubit quantum proof $\xi = \otimes_{j=1}^k \xi_j$, where $\xi_j \in \mathcal{D}\left(\left(\mathbb{C}^2\right)^{\otimes p_2(n)}\right)$. Then there exists a Hamiltonian $H_x$ consisting of $q$-local PSD terms acting on $k p_2(n)$-qubits such that for all $\xi$, we have
\begin{align}
    \Pr[V_x \text{ accepts } \xi] = 1 - \tr[H_x \xi].
    \label{eq:QPCP_is_LH}
\end{align}
\label{lem:H_from_QPCP}
\end{lemma}
\begin{proof}
Let $\Omega$ be the set of all unordered subsets of size $q$ drawn from $[k p_2(n)]$, i.e.,
\[
\Omega = \binom{[k p_2(n)]}{q} = \{ \{i_1, \dots, i_q\} \mid i_j \in [k p_2(n)], \, i_j \neq i_k \text{ for } j \neq k \}.
\]
 By the definition of conditional probability,
\begin{align*}
    \Pr[V_x \text{ accesses qubits $(i_1,\dots,i_q)$ from $\ket{\xi}$ and rejects } ] &=\Pr[(i_1,\dots,i_q)] \cdot \Pr[V \text{ rejects } \xi | (i_1,\dots,i_q)]  \\
    &=\Pr[(i_1,\dots,i_q)] \cdot \frac{\tr[ P_{x,(i_1,\dots,i_q)} \rho]}{\Pr[(i_1,\dots,i_q)]}\\
    &=\tr [P_{x,(i_1,\dots,i_q)} \sigma_{i_1,\dots,i_q}],
\end{align*}
where
\begin{align*}
    \sigma_{i_1,\dots,i_q} = \tr_{\bar{C}_{i_1,\dots,i_q}} \big[\xi \otimes \ket{0}\bra{0}^{\otimes n + p_2(n)}\big]
\end{align*}
with $\bar{C}_{i_1,\dots,i_q}$ denoting all qubits of $\xi$ except for those with indices $i_1,\dots,i_q$.
Hence, we can write the probability that $V_x$ rejects $\xi$ as
\begin{align*}
   \Pr[V_x \text{ rejects } \xi] &= \sum_{\{i_1,\dots,i_q\} \in \Omega} \sum_{(i_1,\dots,i_q) \in \{(i_1,\dots,i_q)\}!}   \tr [P_{x,(i_1,\dots,i_q)} \sigma_{i_1,\dots,i_q}].
\end{align*}
For all $\{i_1,\dots,i_q\} \in \Omega$, we define the $2^q \cross 2^q$ matrix $H_{x,i_1,\dots,i_q} $ as
\begin{align}
   \bra{\alpha} H_{x,i_1,\dots,i_q} \ket{\beta} &:= \sum_{(i_1,\dots,i_q) \in \{(i_1,\dots,i_q)\}!}  \left( \bra{0}^{\otimes p_1(n)+n}  \otimes \bra{\alpha}\right) P_{x,(i_1,\dots,i_q)} \left( \ket{0}^{\otimes p_1 (n)+n}\otimes \ket{\beta} \right),
   \label{eq:local_terms_in_P}
\end{align}
where $\alpha,\beta \in \{0,1\}^q$. By~\cref{lem:AtoB}, we have that for any $q$-local density matrix $\rho$ we have that the expression
\begin{align}
    \tr[H_{x,i_1,\dots,i_q} \rho] &= \tr[\sum_{(i_1,\dots,i_q) \in \{(i_1,\dots,i_q)\}!} P_{x,(i_1,\dots,i_q)} \ket{0}\bra{0}^{\otimes n + p_2(n)} \otimes \rho]
    \label{eq:exp_local_terms}
\end{align}
holds. Moreover, since~\cref{eq:exp_local_terms} is the sum of all probabilities that a query path $(i_1,\dots,i_1)$ is taken and the proof is rejected, taken over all possible permutations of the indices, we must have that $H_{x,i_1,\dots,i_q}$ is PSD. Now consider the $q$-local Hamiltonian $H$ defined as 
\begin{align*}
    H_x = \sum_{\{i_1,\dots,i_q\} \in \Omega}  H_{x,i_1,\dots,i_q}.
\end{align*}
We have that
\begin{align*}
    \Pr[V_x \text{ accepts } \xi] = 1 - \tr[H_x \xi],
\end{align*}
since by the linearity of the trace
\begin{align*}
    \tr[H_{x} \xi] &=   
    \tr [\sum_{\{i_1,\dots,i_q\} \in \Omega}  \left(H_{x,i_1,\dots,i_q} \otimes \mathbb{I} \right) \xi ] \\
    &=   \sum_{\{i_1,\dots,i_q\} \in \Omega}   \tr [ \left(H_{x,i_1,\dots,i_q} \otimes \mathbb{I} \right) \xi ] \\
&=   \sum_{\{i_1,\dots,i_q\} \in \Omega}   \tr [H_{x,i_1,\dots,i_q} \tr_{\bar{C}_{i_1,\dots,i_q}}[\xi] ] \\
   &=\sum_{\{i_1,\dots,i_q\} \in \Omega} \sum_{(i_1,\dots,i_q) \in \{(i_1,\dots,i_q)\}!}  \tr [P_{x,(i_1,\dots,i_q)} \sigma_{(i_1,\dots,i_q)}] \\
   &=\sum_{\{i_1,\dots,i_q\} \in \Omega} \sum_{(i_1,\dots,i_q) \in \{(i_1,\dots,i_q)\}!} \Pr[(i_1,\dots,i_q)] \Pr[V \text{ rejects } \xi | (i_1,\dots,i_q)]\\ 
&= \Pr[V \text{ rejects } \xi]\\
&= 1- \Pr[V \text{ accepts }\xi],
\end{align*}
which also implies that $H_{x,i_1,\dots,i_q} \preceq 1$ for all $i_1,\dots,i_1 \in \Omega$.
\end{proof}

\subsection{Learning the Hamiltonian}
We have shown that the probability that a $\QPCP$ accepts a certain proof is equivalent to the expectation value of some Hamiltonian. We still have to show how to obtain the entries of each term efficiently. Before we state the final protocol, let us argue that we can indeed learn each local term up to arbitrary (inverse polynomial) precision. To do this we will use the Hadamard test, introduced in~\cite{aharonov2006polynomial}.

We need a simple generalisation of the one presented in~{\cite{aharonov2006polynomial}}, as we require two different input states on both sides of the inner product. A proof of this can be found in Chapter 2 of~{\cite{Grilo2018thesis}}.

\begin{lemma}[Hadamard test~\cite{aharonov2006polynomial}] Let $\ket{\psi},\ket{\phi} \in \mathbb{C}^{2^n}$ be quantum states with state preparation unitaries $U_{\phi}$, $U_{\psi}$, i.e. $\ket{\psi} = U_{\psi} \ket{0^n}$ and $\ket{\phi} = U_{\phi} \ket{0^n}$. Let $W \in \mathbb{U}(2^n)$ be some unitary. Then there exists a polynomial-time quantum algorithm that outputs an estimate $\hat{z}$ such that
\begin{align*}
    \abs{\hat{z}- \mathrm{Re}(\bra{\psi} W \ket{\phi})}\leq \epsilon
\end{align*}
with probability $\geq 1-\delta$, in
\begin{align*}
    \mO\left( \frac{ \log \left(\frac{1}{\delta}\right) }{\epsilon^2}\right).
\end{align*}
(controlled) queries to $U_\psi$, $U_\phi$ and $W$. Similarly, there exists a quantum algorithm to estimate $\mathrm{Im}(\bra{\psi} W \ket{\phi})$ at the cost of applying one additional single-qubit gate.
\label{lem:had}
\end{lemma}

Note that the Hadamard test only works for unitaries. Therefore, for our purposes we also need the following lemma, which shows that every (local) projector on a basis state can be written as a linear combination of unitaries with short circuit depth.

\begin{lemma}
Let $\Pi_q = \ketbra{q} $, where $ q \in \{0,1\}^k $ is a basis state. Then $\Pi_q$ can be written as  
\begin{align*}
    \Pi_q = \bigotimes_{i \in [k]} \frac{Z_i+(1-2q_i) \mathbb{I}}{2} = \frac{1}{2^k}\sum_{j \in \{0,1\}^k} a_j U_j,
\end{align*}
where $ U_j \in \{I,Z\}^{\otimes k} $ and $ a_j \in \{-1,+1\} $.
\label{lem:decomp_proj}
\end{lemma}

\begin{proof}
We have that $ \ket{0}\!\bra{0} = (Z+\mathbb{I})/2 $ and $ \ket{1}\!\bra{1} = (Z-\mathbb{I})/2 $. Hence, for $ q_i \in \{0,1\} $, we obtain  
\[
\ketbra{q_i}{q_i} = \frac{Z_i+(1-2q_i) \mathbb{I}}{2}.
\]
Since $ \ketbra{q} = \bigotimes_{i \in [k]} \ketbra{q_i}{q_i} $, we have that
\begin{align}
     \Pi_q = \bigotimes_{i \in [k]} \ketbra{q_i} = \bigotimes_{i \in [k]} \frac{Z_i+(1-2q_i) \mathbb{I}}{2} = \frac{1}{2^k}\sum_{j \in \{0,1\}^k} a_j U_j,
     \label{eq:proj_decomp}
\end{align}
where each $ U_j $ is of the form $ V^j_1 \otimes V^j_2 \otimes \dots \otimes V^j_k $ for $ V_i^j \in \{ \mathbb{I},Z\} $ and $ a_j \in \{-1,1\} $.
\end{proof}

Before we move to prove the existence of the reduction, we need to define some parameters. 
The operators $H_{x,i_1,\dots,i_q}$ (whose entries are defined in~\cref{eq:local_terms_in_P}) are composed of $q+1$ unitaries $\{V^t\}$, a total of $q$ of $\mO(\log n)$-local PVM elements (see~\cref{sec:QPCP_def}), and a single $1$-local PVM element (which is $\Pi_\textup{output}$). If we use $\log(k p_2(n)) +1$ qubits for each $\Pi^t$ and decompose each PVM element in $H_{x,i_1,\dots,i_q}$ into unitaries, as per~\cref{lem:decomp_proj}, we can write
\begin{align*}
   \bra{\alpha} H_{x,i_1,\dots,i_q} \ket{\beta} &= \frac{1}{\Gamma}  \sum_{(i_1,\dots,i_q) \in S(\{(i_1,\dots,i_q)\})}  \sum_{j \in [\Gamma]} a_j \left( \bra{0}^{\otimes p_1(n)+n}  \otimes \bra{\alpha}\right)  U_{j,x} \left( \ket{0}^{\otimes p_1(n)+n}  \otimes \ket{\beta}\right)\\,
\end{align*}
with
\begin{align*}
    U_{j,x} = \prod_{l \in [4q+3]}  U^l_{j,x},
\end{align*}
where the unitaries
$U^l_{j,x}$ are composed of a polynomial number of elementary gates and $a_j \in \{-1,+1\}$. The range of index $l$ can be seen from inspecting~\cref{eq:P_eq} of~\cref{lem:P_rej_cond}: each $ P_{x,(i_1,\dots,i_q)}$, which makes up a $ H_{x,(i_1,\dots,i_q)}$, consists of $q+1$ $V_t$'s and their conjugate transposes, two unitaries for each of the first $q$ PVM elements coming from its decomposition and a final single one for the final outcome measurement (which is ``sandwiched'' in the middle of~\cref{eq:P_eq}). Hence, we have a total of $2(q+1) + 2q+1 = 4q+3$ unitaries in the product.  The total number of unitaries in the linear combination for each $(i_1,\dots,i_q)$ is given by
\begin{align*}
    \Gamma := \left(2^{\log(k p_2(n)) +1}\right)^q \cdot 2 = 2^{q+1}( k p_2(n))^q.
\end{align*}

We define
\begin{align}
    z_{(i_1,\dots,i_q)}^{\alpha,\beta,j} := \left( \bra{0}^{\otimes p_1(n)+n}  \otimes \bra{\alpha}\right)  U_{j,x} \left( \ket{0}^{\otimes p_1(n)+n}  \otimes \ket{\beta}\right),
\end{align}
and
\begin{align}
    h_{i_1,\dots,i_q}^{\alpha,\beta} = \sum_{(i_1,\dots,i_q) \in S(\{(i_1,\dots,i_q)\})} \sum_{j \in [\Gamma]} a_j  z_{(i_1,\dots,i_q)}^{\alpha,\beta,j}.
\end{align}
such that
\begin{align*}
    \bra{\alpha} H_{x,i_1,\dots,i_q} \ket{\beta} = \frac{h_{i_1,\dots,i_q}^{\alpha,\beta}}{\Gamma},
\end{align*}
for all $\alpha,\beta \in \{ 0,1\}^q$ and $\{i_1,\dots,i_q\} \in \Omega$.

\begin{theorem}
Let $q \in \mathbb{N}$ be some constant and $p_1,p_2,p_3$ be polynomials. Let $V_x$ be a $[k,q,p_1,p_2,p_3]$-$\QPCP$-verifier as in~\cref{def:adap_QPCP}, taking input $x$, $|x| = n$ for some $n \in \mathbb{N}$, and a $k p_2$-qubit quantum proof $\xi$. For all $\epsilon > 0$, there exists a quantum reduction from $V_x$ to a $q$-local Hamiltonian $\hat{H}_x = \sum_{i \in [m]} \hat{H}_{x,i}$, with $m = \poly(n)$, $\hat{H}_{x,i}$ PSD for each $i \in [m]$, and $\norm{\hat{H}_x} \leq 1$, such that, given a proof $\xi$,
\begin{align}
    \abs{\Pr[V_x \text{ accepts } \xi] - \left(1 - \tr[\hat{H}_x \xi ]\right)} \leq \epsilon.
    \label{eq:error_H}
\end{align}
The quantum reduction succeeds with probability $1-\delta$ and runs in time $\poly(n,1/\epsilon,\log(1/\delta))$.
\label{thm:red}
\end{theorem}

\begin{proof}
The reduction is specified in~\cref{alg:reduction_QPCP} below. We proceed to show its correctness by arguing it has the required precision, success probability, and time complexity. 
\paragraph{Precision:} Step $1a$ of~\cref{alg:reduction_QPCP} produces estimates $\tilde{z}_{(i_1,\dots,i_q)}^{\alpha,\beta,j}$ for the parameters $z_{(i_1,\dots,i_q)}^{\alpha,\beta,j}$ using~\cref{lem:had}. We have that 
\begin{align*}
    \abs{z_{(i_1,\dots,i_q)}^{\alpha,\beta,j} - \tilde{z}_{(i_1,\dots,i_q)}^{\alpha,\beta,j}} \leq 2 \epsilon',
\end{align*}
since we estimated both the real and imaginary parts up to precision $\epsilon'$. By the triangle inequality (and $a_j \in \{1,-1\}$) we now have
\begin{align*}
    \abs{\frac{\tilde{h}^{\alpha,\beta}_{i_1,\dots,i_q}}{\Gamma} - \bra{\alpha}H_{x,i_1,\dots,i_q} \ket{\beta}} \leq 2 q! \epsilon'.
\end{align*}
Since $\tilde{H}_{x,i_1,\dots,i_q} = \sum_{\alpha,\beta \in \{0,1\}^q}  \tilde{h}^{\alpha,\beta}_{i_1,\dots,i_q} \ket{\alpha} \bra{\beta}$, we have that
\begin{align*}
    \norm{\tilde{H}_{x,i_1,\dots,i_q} - H_{x,i_1,\dots,i_q}} &\leq 2^q \max_{\alpha,\beta} \abs{\bra{\alpha} H_{x,i_1,\dots,i_q} \ket{\beta} -  \bra{\alpha} \tilde{H}_{x,i_1,\dots,i_q} \ket{\beta}} \\
    &\leq 2^{q+1} q! \epsilon',
\end{align*}
which follows from the bound on the operator norm by the max-norm. Now suppose that $\tilde{H}_{x,i_1,\dots,i_q}$ is not PSD. Since $H_{x,i_1,\dots,i_q}$ is PSD, we have that $\lambda_\textup{min}(\tilde{H}_{x,(i_1,\dots,i_q)}) \geq - 2^{q+1}  q! \epsilon'$, so we have that adding the identity term can only double the error, making it at most $2^{q+2} q! \epsilon'$. By another triangle inequality
\begin{align*}
    \norm{\tilde{H}_x - H_x} &\leq |\Omega| 2^{q+2} q! \epsilon' \leq \epsilon/4,
\end{align*}
for our choice of $\epsilon'$. Finally, the error introduced by step 3 of the protocol can be bounded by
\begin{align*}
    \norm{\hat{H}_x - H_x} &\leq  \norm{\hat{H}_x - \tilde{H}_x} +  \norm{\tilde{H}_x - H_x} \\
    &\leq  \abs{\frac{1}{1+\epsilon/4}-1} \norm{\tilde{H}_x} +  \frac{\epsilon}{4} \\
    &\leq \frac{\epsilon}{4} (1+\frac{\epsilon}{4}) + \frac{\epsilon}{4} \\
    &= \frac{\epsilon}{2} + \left(\frac{\epsilon}{4}\right)^2 \\
    &\leq \epsilon.
\end{align*}

Hence, for any state $\xi = \sum_{i} p_i \ketbra{\psi_i}$, with $\sum_{i} p_i = 1$, we have
\begin{align*}
\abs{\Pr[V_x \text{ accepts } \xi] - \left(1 - \tr[\hat{H}_x \xi]\right)} &=  \abs{\tr[\hat{H}_x \xi] - \tr[H_x \xi]} \\
&= \abs{\tr\left[(\hat{H}_x - H_x) \xi\right]} \\
    &= \abs{\tr\left[(\hat{H}_x - H_x) \sum_{i} p_i \ketbra{\psi_i}\right]} \\
    &= \sum_{i} p_i \abs{\bra{\psi_i} (\hat{H}_x - H_x) \ket{\psi_i}} \\
    &\leq \epsilon,
\end{align*}
as desired.

\paragraph{Success probability:} We have to count the number of times we run the Hadamard test of~\cref{lem:had}, each of which succeeds with success probability $\geq 1-\delta'$. Recall that $\Omega$ is the set of all possible indices (when order does not matter), which is given by $\Omega = \binom{[k p_2(n)]}{q}$ for proofs of length $p_2$. We run the Hadamard test for a total of $|\Omega| q! 4^q \Gamma$ times (see the number of iterations in~\cref{alg:reduction_QPCP}), and thus
\begin{align*}
    (1-\delta')^{|\Omega| q! 4^{q+1} \Gamma} \geq 1-\delta' |\Omega| q! 4^{q+1} \Gamma = 1- \delta,
\end{align*}
using the inequality $(1-x)^T \geq 1-T x$ for all $x \in [0,1]$. The extra factor of two again accounts for estimating both the real and imaginary parts.

\paragraph{Time complexity:} By definition of $\QPCP[q]$, we have that $V$ has gate complexity $\poly(n)$. Using~\cref{lem:decomp_proj}, we have that $V$, $U_{\phi}$, and $U_{\psi}$ always have polynomially bounded gate complexities. Filling in our choice of $\delta'$ and $\epsilon'$, we have that the total number of (controlled) applications of $V$, $U_{\phi}$, and $U_{\psi}$ can be upper bounded as
\begin{align*}
    \mO \left( q! 4^{q+1} \Gamma \frac{(|\Omega| 2^{q+4} q! )^2 \log \left(\frac{q! 4^{q+1} \Gamma }{\delta}\right) }{\epsilon^2} \right) = \poly(n,1/\epsilon, \log(1/\delta)),
\end{align*}
since for $k = \poly(n)$ and $q = \mO(1)$ we have $|\Omega| = \poly(n)$ and $\Gamma = \poly(n)$.
\end{proof}

\begin{custalgo}[Quantum reduction from $\QPCP{[k,q]}$ verification to a local Hamiltonian]
\noindent \textbf{Input:} A $(k,q,p_1,p_2,p_3)\text{-}\QPCP$ verifier $V_x$ with hardcoded input $x$, a precision parameter $\epsilon$ (variation: $\eta$), maximum error probability $\delta$.\\

\noindent \textbf{Set:} $\Omega \coloneqq  \binom{[k p_2(n)]}{q}$, $\Gamma :=  2^{q+1}( k p_2(n))^q$, $\epsilon' := \frac{\epsilon}{|\Omega| 2^{q+4} q! }$, $\delta' := \frac{\delta}{|\Omega| q! 4^{q+1} \Gamma }$.\\

\noindent \textbf{Algorithm:}
\begin{enumerate}
    \item For $\{i_1,\dots,i_q\} \in \Omega$
    \begin{enumerate}
        \item For $(i_1,\dots,i_q) \in \{(i_1,\dots,i_q)\}!$
   \begin{enumerate}
       \item  For $\alpha \in \{0,1\}^q$, $\beta \in \{0,1\}^q$:
    \begin{itemize}
            \item For $j \in [\Gamma]$ compute $a_j,U_{j,x}$ and estimate $z_{(i_1,\dots,i_q)}^{\alpha,\beta,j}$ as $\tilde{z}_{(i_1,\dots,i_q)}^{\alpha,\beta,j}$ using~\cref{lem:had} for both the real and imaginary 
            part with $V = U_{j,x}$, $U_\theta = \hat{V} (\otimes_{i=1}^q (X_i)^{\alpha_i} \otimes \mathbb{I}) $ and $U_\phi = \hat{V} (\otimes_{i=1}^q (X_i)^{\beta_i} \otimes \mathbb{I} )$, $\epsilon'$  and $\delta'$.
    \end{itemize}
    \item Set $   \tilde{h}^{\alpha,\beta}_{i_1,\dots,i_q} = \sum_{(i_1,\dots,i_q) \in S(\{(i_1,\dots,i_q)\}} \sum_{j \in [\Gamma]} a_j \tilde{z}_{(i_1,\dots,i_q)}^{\alpha,\beta,j}$.
        \end{enumerate}
    \item Set $\tilde{H}_{x,(i_1,\dots,i_q)} = \sum_{\alpha,\beta \in \{0,1\}^q}  \tilde{h}^{\alpha,\beta}_{i_1,\dots,i_q} \ketbra{\alpha}{\beta}$.
    \item Compute $\lambda_\textup{min}(\tilde{H}_{x,i_1,\dots,i_q})$. If $\lambda_\textup{min}(\tilde{H}_{x,i_1,\dots,i_q}) < 0$, let $\tilde{H}_{x,i_1,\dots,i_q}) \leftarrow \tilde{H}_{x,i_1,\dots,i_q}) -\lambda_\textup{min}(\tilde{H}_{x,i_1,\dots,i_q}) \mathbb{I} $, else continue.
   \end{enumerate}
   \item Let $\tilde{H}_x = \sum_{\{i_1,\dots,i_q\} \in \Omega} \tilde{H}_{x,i_1,\dots,i_q}$. Output
   \begin{align*}
        \hat{H}_x = \frac{\tilde{H}_x}{\textup{argmax}\{\norm{\tilde{H}_x},1\}}.
   \end{align*}

\end{enumerate}
\noindent \textbf{Variation:} To learn $\tilde{H}_x$ up to $\eta$ bits of precision, use $\epsilon':=\frac{2}{\left(2^\eta+1\right)}$ and keep only the first $\eta$ bits for the estimates of $z_{(i_1,\dots,i_q)}^{\alpha,\beta,j}$ as $\tilde{z}_{(i_1,\dots,i_q)}^{\alpha,\beta,j}$ in Step 1(a)i.
  \label{alg:reduction_QPCP}
\end{custalgo}

\

One can use a simple trick that exploits the freedom we have in the reduction to set the precision. Since every local term is bounded in the operator norm, we must also have that for each of these terms the matrix entries are bounded by $1$. Hence, we can convert a bound in terms of approximation in entry-wise error to one in having a least a certain number of bits from a certain bit-wise representation of the value being correct. The advantage of the latter is that it allows us to make the reduction deterministic in the sense that if it succeeds, it always produces \emph{exactly} the same Hamiltonian. This can alternatively be viewed as having a ``rounding scheme'' which, with high probability, always rounds to the same operator in operator space. Though not strictly needed, this allows one to always reason about the same Hamiltonian when the reduction is used as a subroutine.

\begin{restatable}{corollary}{fixedHcor} For any $\epsilon,\delta >0$, under the same setup as in~\cref{thm:red}, there exists a quantum reduction which produces a fixed Hamiltonian $\tilde{H_x}$ with probability $1-\delta$ which satisfies~\cref{eq:error_H} and runs in time $\poly(n,1/\epsilon,\log(1/\delta))$. 
\label{cor:fixed_H}    
\end{restatable}

We have put the proof of~\cref{cor:fixed_H} in~\cref{app:proof_cor_fixed_H}.

\begin{remark} Unlike these general quantum PCPs, where the probabilities of taking a certain query path can depend on the proof, we have that for non-adaptive quantum PCPs the probability distribution over which the proof is accessed only depends on the input. Generalising the proof of Lemma 6.3 in~\cite{weggemans2023guidable}, one can easily verify that for non-adaptive quantum PCPs the obtained Hamiltonian can be assumed to be of the form $H_x = \sum_{i \in [m]} p_{x,i} H_{x,i}$, where the $p_{x,i}$'s form a probability distribution and $0 \preceq H_{x,i} \preceq 1$ for all $q$-local terms. For these general (adaptive) quantum PCPs, the key difference is that each term of the Hamiltonian encodes \textit{both} the probability that certain parts of the proof are accessed as well as the eventual probability of acceptance. We have provided sufficient conditions on the error parameters of the reduction in~\cref{app:weighted_error_conditions}.
\label{remark:NA}
\end{remark}

\subsection{Local smoothings of local Hamiltonians}
In this subsection we show that one can always transform any Hamiltonian in the form of $H_x$  of Eq.~\eqref{eq:QPCP_is_LH} to match the form usually adopted in the literature; specifically, we want to transform a local Hamiltonian with PSD terms, operator norm at most $1$ and constant promise gap to another Hamiltonian such that the promise gap becomes $\Omega(m)$, where $m$ is the number of terms in the new Hamiltonian, whilst still having that each term $H_i$ is local and satisfies $0 \preceq H_i \preceq 1$. This way, we have that at least a constant fraction of all the $m$ terms contribute to the energy in the {\sc No}-case, which intuitively implies that only a constant number of terms have to be checked to find out whether the energy of the Hamiltonian is high or low (see~\cref{subsec:KEEP}). This turns out to be possible at the cost of increasing the locality by a factor of two and the promise gap by a constant factor.\footnote{We believe such a result might be known in existing literature, but as we could not find a reference we provide a proof for completeness.}

\begin{lemma}[Hamiltonian local smoothing lemma] Let $H$ be a $q$-local Hamiltonian on $n$-qubits, $q = \mO(1)$, such that $H = \sum_{i \in [m]} H_i$, where each $H_i$ is PSD and $0 \preceq H \preceq 1$, and $\lambda_{\textup{min}}(H) \leq a$ or $\geq b$ for $b-a = \gamma(m)$. Then there exists a polynomial time transformation to another $2q$-local Hamiltonian $H'$ on $n$ qubits such that $H' = \frac{1}{m} \sum_{i \in [m]} H'_i$, $0 \preceq H_i' \preceq 1$ for all $i \in [m]$ and  $\lambda_{\textup{min}}(H) \leq \alpha  $ or $\geq \beta $ for some $\beta,\alpha$ such that $\beta-\alpha \in \Omega( \gamma(m))$.   
\label{lem:H_smoothing}
\end{lemma}
\begin{proof}
Write $\rho_{d} = \frac{\mathbb{I}_{d}}{d}$ for the $d$-dimensional maximally mixed state.  We have 
\begin{align*}
    \tr\left[H \rho_{2^n} \right] = \tr\left[\sum_{i \in [m]}H_i \tr_{\bar{C}_i}\left[\rho_{2^n}\right] \right] =\tr\left[\sum_{i \in [m]}H_i \rho_{2^q}  \right] = \frac{1}{2^q} \sum_{i \in [m]} \tr[H_i].
\end{align*}
The variational principle tells us that
\begin{align*}
    \tr\left[H \rho_{2^n} \right] \leq \max_{\xi} \tr[H\xi] \leq \norm{H} \leq 1.
\end{align*}
Combining the two, this implies
\begin{align*}
       2^q \geq  \sum_{i \in [m]} \tr[H_i] \geq \sum_{i \in [m]} \norm{H_i} 
\end{align*}
using that $0 \preceq H_i \preceq 1$ implies that $\tr[H_i] \geq \norm{H_i}$. Now consider $\hat{H} = \sum_{i \in [m]} \hat{H}_i $ where $\hat{H}_i = \frac{1}{2^{q+3}} H_i$. We have that $\lambda_0(\hat{H}) \geq  \frac{b}{2^{q+3}} $ or $\lambda_0(\hat{H}) \leq \frac{a}{2^{q+3} } $. Let $\alpha_i = \norm{\hat{H}_i}$, for which we know that $\sum_{i \in [m]} \alpha_i \leq \frac{1}{8}$. Define index sets:
\begin{align*}
    L &= \{i \in [m] : \alpha_i \leq \frac{1}{2m}\}, \\
    U &= [m] \setminus L = \{i \in [m] : \alpha_i > \frac{1}{2m}\}.
\end{align*}
We have
\[
|U| \frac{1}{2m} \leq \sum_{i \in U} \alpha_i \leq \sum_{i \in [m]} \alpha_i \leq \frac{1}{8},
\]
which implies $|U| \leq \frac{m}{4}$ and hence $|L| \geq \frac{3m}{4}$.
We now construct a new Hamiltonian $H' = \sum_{i \in [m']} H_i'$, where each $H_i'$ is $2q$-local and satisfies $0 \preceq H_i' \preceq 1$, by redistributing each high-norm term $\hat{H}_j$ from $U$ into smaller pieces and assigning them to low-norm terms from $L$. This way, we are guaranteed that all new terms are positive semi-definite and have operator norm at most $1/m$, so we can simply scale with a factor $m$ to obtain an operator norm bound of $1$. For each $\alpha_j$, assuming $\alpha_j > 1/m$ with $j \in U$, we want to find the largest possible integer $t_j$ such that:
\begin{align*}
    \frac{1}{2m} \leq \alpha_j - t_j \frac{1}{2m} \leq \frac{1}{m}.
\end{align*}
Which implies that $t_j \geq 2m \alpha_j - 2 $ and $ t_j \leq 2m \alpha_j-1$, so we can take $t_j = \lfloor 2m \alpha_j -1 \rfloor$.

Consider the following procedure:
\begin{enumerate}
    \item Initialise $L' \coloneqq L$.
    
    \item For each $j \in U$, check if $\alpha_j > 1/m$. If this is not the case, continue the loop (or exit after the last $j$). If this is the case, set $t_j = \lfloor 2m\alpha_j -1 \rfloor$.  For the first $t_j$ indices $i \in L'$, define
    \[
    Q_i := \hat{H}_i  + \frac{1}{2m} \hat{H}_j.
    \]
    These $Q_i$ are at most $2q$-local and have operator norm at most $1/m$. Remove these $t_j$ indices $i$ from $L'$. Define the $q$-local leftover term
    \[
    Q_j := \left(1 - \frac{t_j}{2m} \right) \hat{H}_j.
    \]
    \item For each remaining $j \in U$ that was not used in any redistribution, simply set $Q_j := \hat{H}_j$.

    \item Let $ H_i' := m Q_i $ for each term. Return $ H' := \frac{1}{m}\sum H_i' $.
\end{enumerate}
We only have to check whether the above procedure does not run out of terms in $L$ to redistribute to. We obtain:
\begin{align*}
\sum_{j \in U} t_j 
\leq \sum_{j \in U} \lfloor 2m \alpha_j -1 \rfloor \leq \sum_{j \in U} 2m \alpha_j \leq \frac{1}{4} m \leq |L|, 
\end{align*}
using that $\sum_{j \in [m]} \alpha_j \leq 1/8$. By construction, each term in $H'$ is at most $2q$-local and satisfies $0 \preceq H_i' \preceq 1$, which means that $\norm{H} \leq 1$. Since $\hat{H} = H / 2^{q+3}$ and $H'$ has the same spectrum as $\hat{H}$ (we only combined different terms together and rescaled), we have:
\begin{itemize}
    \item If $ \lambda_0(H) \leq a $, then $ \lambda_0(H') \leq  \alpha  $ with $\alpha \coloneqq \frac{a}{2^{q+3}} $
    \item If $ \lambda_0(H) \geq b $, then $ \lambda_0(H') \geq \ \beta $ with $ \beta \coloneqq \frac{b}{2^{q+3}}  $.
\end{itemize}
Hence, $\beta - \alpha = \Omega(\gamma(m))$ when $q = \mO(1)$.
\end{proof}

\subsection{Kitaev's energy estimation protocol}
\label{subsec:KEEP}

In this section we saw that general quantum PCP can be transformed into a local Hamiltonian with a constant promise gap (and a specific form). 
In the following sections it will prove useful to have a verifier for the local Hamiltonian problem suitable for Hamiltonians from quantum PCP verifiers. 
For this we will use Kitaev's $\QMA$-protocol for the local Hamiltonian problem~\cite{Kitaev2002ClassicalAQ}, albeit with a small modification to allow for sampling the terms according to some probability weights associated with the Hamiltonian. We formulate the energy protocol in the pure setting as the mixed state case can simply be viewed as a convex combination of acceptance probabilities given by the pure states.\\

\begin{protocol}[Kitaev's energy estimation protocol]
\noindent \textbf{Input:} A classical description of a $n$-qubit, $q$-local Hamiltonian of the form $H = \sum_{i \in [m]} p_i H_i$, $0 \preceq H_i \preceq 1 $ with weights $\{p_i\}$ such that $\sum_{i} p_i = 1$.\\

\noindent \textbf{Protocol:}
      \begin{enumerate}
        \item The prover sends the state $\ket{\psi}$.
        \item For each $H_i$, $i \in [m]$, let $H_i = \sum_{j} \lambda_{i,j} \ketbra{\lambda_{i,j}}$ be its spectral decomposition. Define the $(q+1)$-local operator $W_i$ such that $W_i $ acts on a $(n+1)$-qubit space as
    \begin{align}
        W_i \ket{\lambda_{i,j}} \ket{b} = \ket{\lambda_{i,j}} \left(\sqrt{\lambda_{i,j}} \ket{b} + \sqrt{1-\lambda_{i,j}} \ket{b \oplus 1}\right).
    \end{align}
        The verifier picks a $i \in [m]$ with probability $p_i$, and applies $W_i$ on $\ket{\psi} \ket{0}$, and measures the final qubit.
        \item The verifier accepts if and only if the outcome is $\ket{1}$.
    \end{enumerate}
  \label{prot:Kitaev_protocol}
\end{protocol}
\\

\begin{lemma}[Kitaev's energy estimation protocol (weighted)] 
Consider a $q$-local, $n$-qubit Hamiltonian $H = \sum_{i \in [m]} p_i H_i$  with $\sum_{i \in [m]} p_i = 1$, $p_i \geq 0$, and $0 \preceq H_i \preceq 1$ for all $i \in [m]$. Then, given an $n$-qubit quantum state $\ket{\psi}$, there exists a measurement on $q$ qubits of $\ket{\psi}$ that outputs $1$ with probability 
\begin{align}
    1-\bra{\psi} H \ket{\psi}.
\end{align}
\label[lemma]{lem:Kitaevs_protocol}
\end{lemma}

\begin{proof}
This follows from a simple generalisation of Kitaev's original $\QMA$ verification protocol, which can be found in Chapter 14 of~\cite{Kitaev2002ClassicalAQ}.   Let us now show the correctness of~\cref{prot:Kitaev_protocol}. For any $i$, let $\ket{\psi} = \sum_{j} \alpha_{i,j} \ket{\lambda_{i,j}}$ be the decomposition of $\ket{\psi}$ in the eigenbasis of $H_i$. The probability that this protocol accepts, conditioned on picking term $i$, is given by
\begin{align*}
    \Pr[\mathcal{V} \text{ accepts } \ket{\psi} \mid  i ] &=  \left\| (\mathbb{I} \otimes \ketbra{1})  W_i \ket{\psi} \ket{0} \right\|_2^2 \\
    &= \left( \sum_{j} \bar{\alpha}_{i,j} \bra{\lambda_{i,j}} \bra{0} \right) W^\dagger_i (\mathbb{I} \otimes \ketbra{1}) W_i \left( \sum_{j} \alpha_{i,j} \ket{\lambda_{i,j}} \ket{0} \right) \\
    &= \left( \sum_{j} \bar{\alpha}_{i,j} \sqrt{1-\lambda_{i,j}} \bra{\lambda_{i,j}} \right) \left( \sum_{j} \alpha_{i,j} \sqrt{1-\lambda_{i,j}} \ket{\lambda_{i,j}} \right) \\
    &= \sum_j \left(1-\lambda_{i,j}\right) \bar{\alpha}_{i,j} \alpha_{i,j} \\
    &= 1-\bra{\psi} H_{i} \ket{\psi}.
\end{align*}
The overall acceptance probability is then given by the expectation value over all choices of $i$, which is
\begin{align*}
     \sum_{i \in [m]} p_i \Pr[\mathcal{V} \text{ accepts } \ket{\psi} \mid  i] =  1- \bra{\psi} \sum_{i \in [m]} p_i H_i \ket{\psi} = 1 - \bra{\psi} H \ket{\psi}.
\end{align*}
\end{proof}

Kitaev's energy estimation protocol (\cref{prot:Kitaev_protocol}) can be viewed as a $(1,q)$-$\QPCP_{\textup{NA}}$ verifier, where the completeness and soundness bounds correspond to one minus the promised upper and lower bounds on the ground state energy in the {\sc yes}- and {\sc no}-cases, respectively. When $q = \mathcal{O}(\log n)$, each $W_i$ acts non-trivially only on $q + 1$ qubits and can thus be implemented efficiently. By applying~\cref{lem:H_smoothing} in combination with weak error reduction (\cref{lemma:weak_err_red}), we can correctly decide any local Hamiltonian problem with positive semidefinite terms and a constant promise gap using Kitaev’s protocol, taking $p_i = 1/m$ for each $i \in [m]$.
\begin{corollary} 
For any constant $q \in \mathbb{N}$ and constant $\delta > 0$, we have that $q$-$\mathsf{LH}[\delta]$ is contained in $\QPCP_\textup{NA}[q']$ with $q' \in \mO(1)$.
\end{corollary}
\begin{proof}
    Let $a$, $b$ with $b-a = \delta$ be the completeness and soundness parameters of the $q$-$\mathsf{LH}[\delta]$ instance. By~\cref{lem:H_smoothing}, we can transform this into another instance of $2q$-$\mathsf{LH}[\delta']$ with a $2q$-qubit Hamiltonian $H' = \frac{1}{m} \sum_{i \in [m]} H_i'$, $0 \preceq H_i' \preceq 1$, completeness $a'$ and soundness $b'$, with $b'-a' = \delta' > 0$ being some other constant that depends on $q$. We have that $H'$ is in the desired form for~\cref{prot:Kitaev_protocol}, which has completeness $1-a'$ and soundness $1-b'$, and thus also promise gap $\delta'$. Using weak error reduction from~\cref{lemma:weak_err_red}, we can then boost the promise gap back so that it satisfies completeness $2/3$ and soundness $1/3$.
\end{proof}

\section{Applications}
\label{sec:applications}
This section discusses some applications of the ideas presented earlier in the paper, all of which are derived (with varying degrees of overhead) by using the reduction as a subroutine. The section is divided into four subsections, each of which can be read independently to some extent.

\subsection{Reduction to the average particle energy formulation}
As a first application, we consider a specific formulation of the quantum PCP conjecture, stated in terms of an error constant relative to the number of sites (i.e., qubits or qudits), rather than the total number of terms. For this, we rely on the following lemma due to Tropp.

\begin{lemma}[\cite{tropp2012user}] Consider a finite sequence $\{X_k\}$ of independent, random, Hermitian matrices with dimension $d$, and let $\{A_k\}$ be a sequence of fixed Hermitian matrices. Assume that each random matrix satisfies
\[
\mathbb{E}[X_k] = 0 \quad \text{and} \quad X_k^2 \preceq A_k^2 \quad \text{almost surely}.
\]
Then for any $ t \geq 0 $,
\[
\Pr\left[\lambda_{\textup{max}}\left(\sum_k X_k\right) \geq t\right] \leq d \cdot e^{-t^2 / 8\sigma^2},
\]
where
\[
\sigma^2 \coloneqq  \norm{\sum_k A_k^2 }.
\] 
Here, $\lambda_{\textup{max}}(\cdot)$ denotes the largest eigenvalue.
\label[lemma]{lem:matrix_Hoeffding}
\end{lemma}

\begin{proposition} 
Consider a $q$-local Hamiltonian problem with Hamiltonian $H =\frac{1}{m} \sum_{i \in [m]} H_i$, $0 \preceq H_i \preceq 1$, with completeness and soundness parameters $a, b \geq 0$ such that $b - a = \gamma$ for some $\gamma = \Omega(1)$. Then for any $\delta = \Omega(2^{-n})$, there exists a randomized polynomial-time reduction to a $2q$-local Hamiltonian problem with Hamiltonian $G' = \frac{1}{l} \sum_{j \in [l]} G_j'$, $0 \preceq G_j' \preceq 1$, where $l = \mO(n)$, with completeness and soundness parameters $c, d$ such that $c - d \geq \gamma' = \Omega(1)$, which succeeds with probability $1 - \delta$.
\label{prop:red_energy_density}
\end{proposition}

\begin{proof}
    For any positive integer $k$, let $X_k$ be a random variable defined as $H_i - H$ with probability $1/m$. Let $l \in \mathbb{N}$ be the length of the sequence $\{X_k\}_{k \in [l]}$, which is to be determined later. Clearly, for all $k \in [l]$,
\begin{align*}
    \mathbb{E}[X_k] = \frac{1}{m} \sum_{i \in [m]} \left( H_i - H \right) = H - H = 0.
\end{align*}
Since $0 \preceq H_i \preceq 1$ for all $i \in [m]$, it follows that $0 \preceq H \preceq 1$, and therefore $-1 \preceq X_k \preceq 1$, which implies $X_k^2 \preceq 1$. Now, for all $k \in [l]$, set $A_k = \mathbb{I}$, where $A_k = A_k^2$, so that $X_k^2 \preceq A_k^2$ holds. We then have
\begin{align*}
    \sigma^2 := \norm{\sum_{k \in [l]} A_k^2} = \norm{l \mathbb{I}} = l.
\end{align*}

Given the sequence $\{X_k\}_{k \in [l]}$, define $G_k = X_k + H$ and let $G = \frac{1}{l} \sum_{k \in [l]} G_k$. We can then express
\begin{align*}
    \Pr\left[\lambda_{\textup{max}}\left( \frac{1}{l} \sum_{k \in [l]} X_k \right) \geq \frac{t}{l}\right] &= \Pr\left[\lambda_{\textup{max}}\left( H - \frac{1}{l} \sum_{k \in [l]} G_k \right) \geq \frac{t}{l}\right] \\
    &= \Pr\left[\norm{H - G} \geq \frac{t}{l}\right].
\end{align*}
By applying~\cref{lem:matrix_Hoeffding}, we get
\begin{align*}
   \Pr\left[\norm{H - G} \geq \epsilon \right]  &\leq 2^n \cdot e^{-\epsilon^2 l / 8}.
\end{align*}
Now, set $\epsilon = \gamma / 4$. In this case, if $\norm{H - G} \leq \epsilon$, then it must hold that:
\begin{itemize}
    \item if $\lambda_\textup{min}(H) \leq a$, then $\lambda_\textup{min}(G) \leq a' := a + \epsilon$;
    \item if $\lambda_\textup{min}(H) \geq b$, then $\lambda_\textup{min}(G) \geq b' := b - \epsilon$,
\end{itemize}
where $b' - a' \geq \gamma / 2 = \Omega(1)$. To achieve a success probability of at least $1 - \delta$, we require
\begin{align*}
    2^n \cdot e^{-\gamma^2 l / 128} \leq \delta,
\end{align*}
which implies a condition on $l$ of
\begin{align*}
    l \geq \frac{128}{\gamma^2} \left( n \ln(2) + \ln\left(\frac{1}{\delta}\right) \right).
\end{align*}
For $\gamma = \Omega(1)$ and $\delta = \Omega(2^{-n})$, this yields $l = \Theta(n)$. Finally, we must ensure that each $G_k$ satisfies $0 \preceq G_k \preceq 1$. The condition $0 \preceq G_k$ is trivially satisfied, but $G_k \preceq 1$ might not hold if we sample the same term $H_i$ multiple times. To resolve this, we apply the deterministic transformation from~\cref{lem:H_smoothing} to obtain a Hamiltonian  $G' = \frac{1}{l} \sum_{j \in [l]} G_j'$, $0 \preceq G_j' \preceq 1$ at most $2q$-local, which maintains the constant relative promise gap while increasing the locality by a factor of two.
\end{proof}

It is now straightforward to combine the reductions of~\cref{thm:red},~\cref{lem:H_smoothing}, and~\cref{prop:red_energy_density} to obtain the following corollary:

\begin{corollary}
For any $q = \mO(1)$, there exist a $q' = \mO(q)$ and a constant $\delta > 0$ such that the problem $q'$-$\mathsf{LH}[\delta]$, restricted to Hamiltonians with $m = \Theta(n)$ positive semidefinite terms, is $\QPCP[q]$-hard under quantum reductions.
\label{cor:density_form}
\end{corollary}

\subsection{Proof checking versus local Hamiltonian formulations of quantum PCP}
The goal of this subsection is proving~\cref{thm:QCMA}. We will utilize the fact that any $\QCMA$-verifier would able to perform the reduction in~\cref{alg:reduction_QPCP} (up to a polynomial number of bits of precision). This can be used to show that a $\QCMA$ upper bound on the local Hamiltonian problem with constant promise gap implies a $\QCMA$ upper bound on a quantum PCP system with the same locality. \\

\begin{protocol}[$\QCMA$ protocol for $\QPCP{[\mO(1)]}$ assuming that $q$-$\mathsf{LH{[\Theta(1)]}} \in \QCMA$.]
\noindent \textbf{Input:} A classical description of a $(1,q,p_1,p_2,p_3)\text{-}\QPCP$ verifier $V_x$ with input $x$ hardcoded into it, with completeness $c$ and soundness $s$.\\

\noindent \textbf{Set:} $\epsilon \coloneqq  (c-s)/4$, $\Gamma' \coloneqq  (q+1)2 p_2 (n)$, $\eta \coloneqq  \left\lceil q \log \left( \frac{4 \Gamma'(\Gamma' + 1)}{\epsilon} - 1 \right) \right\rceil$,  $\delta \coloneqq 1-\sqrt{\frac{2}{3}}$, $a\coloneqq  c+\epsilon/4$ and $b\coloneqq c-\epsilon/4$.\\

\noindent \textbf{Protocol:}
\begin{enumerate}
    \item The prover sends the witness  $y$.
    \item The verifier performs the \textbf{variation} of~\cref{alg:reduction_QPCP} with precision $\eta$ and maximum error probability $\delta$ to obtain a $q$-local Hamiltonian $\tilde{H}_x = \sum_{i \in [m]} \tilde{H}_{x,i}$ up to $\eta$ bits of precision. 
    \item It accepts if and only if the $\QCMA$ protocol, having completeness $\sqrt{\frac{2}{3}}$ and soundness $1-\sqrt{\frac{2}{3}}$, with witness $y$ for $(\hat{H},a,b)$ accepts.
\end{enumerate}
  \label{prot:QCMA}
\end{protocol}

\

\begin{theorem}[Local Hamiltonian versus proof verification]
     If the $q$-local Hamiltonian problem with a constant promise gap can be decided in $\QCMA$, we have that
    \begin{align*}
        \QPCP[q] \subseteq \QCMA,
    \end{align*}
    for all constant $q \in \mathbb{N}$.
    \label{thm:QCMA}
\end{theorem}
\begin{proof}
Let $A = (A_\textup{yes},A_\textup{no})$ be any promise problem in $\QPCP[q]$, $x$ be an input, $V_x$ the $\QPCP$ verifier with $x$ hardcoded into it. Let $\tilde{H}_x$ be the corresponding local Hamiltonian obtained via the reductions in Step 2, conditioned on Step 2 succeeding. By~\cref{cor:fixed_H}, we have that in this case the same Hamiltonian is always produced with probability $\geq 1-\delta = \sqrt{2/3}$, so the prover can provide the proof without knowing the outcome of the reduction. Moreover, we note that for our choice of parameters, Step 2 can be performed in quantum polynomial-time.

By assumption, the $q$-local Hamiltonian problem with constant promise gap is in $\QCMA$. Hence, there exists a $\QCMA$ verifier $Q$ such that:
\begin{itemize}
    \item if $\lambda_0(H) \leq a$, then there exists $y \in \{0,1\}^{p(n)}$ such that 
    \[
    \Pr[Q \text{ accepts } ((H, a, b), y)] \geq \sqrt{\frac{2}{3}}.
    \]
    \item if $\lambda_0(H) \geq b$, then for all $y \in \{0,1\}^{p(n)}$;
    \[
     \Pr[Q \text{ rejects } ((H, a, b), y)] \geq \sqrt{\frac{2}{3}},
    \]
\end{itemize}
since $\QCMA$ allows for strong error reduction. Now consider~\cref{prot:QCMA}. Let $Q'$ be the $\QCMA$ verifier that first perform the reduction from Step 2 to obtain $\tilde{H}_x$ to $\eta$ bits of precision, and then runs $Q(\tilde{H}_x, a, b, \ket{y})$. Since Step 2 succeeds with probability $1-\delta = \sqrt{\frac{2}{3}}$,  we have:
\begin{itemize}
    \item If $x \in A_\textup{\sc yes}$, there exists $y \in \{0,1\}^{p(n)}$ such that $\Pr[Q' \text{ accepts } (x,y)] \geq 2/3$.
    \item If $x \in A_\textup{\sc no}$, then for all $y \in \{0,1\}^{p(n)}$, $\Pr[Q' \text{ accepts } (x,y)] \leq 1/3$,
\end{itemize}
for our choice of $\delta$. This shows that $\QPCP[q] \subseteq \QCMA$.
\end{proof}

\subsection{Adaptive versus non-adaptive quantum PCPs}
Since adaptive quantum PCPs generalise non-adaptive quantum PCPs, we have that the inclusion $\QPCP_{\textup{A}}[q] \supseteq \QPCP_{\textup{NA}}[q]$ is immediate.  In this subsection, we will prove that non-adaptive QPCPs can also simulate adaptive QPCPs with only constant overhead, by using our~\cref{lem:H_smoothing} and weak error reduction lemma (\cref{lemma:weak_err_red}) for non-adaptive QPCPs.\\

\begin{protocol}[Non-adaptive simulation of an adaptive quantum PCP]
\noindent \textbf{Input:} A classical description of a $(1,q,p_1,p_2,p_3)\text{-}\QPCP$ verifier $V_x$ with input $x$ hardcoded into it, with completeness $c$ and soundness $s$.\\

\noindent \textbf{Set:} $\epsilon \coloneqq  (c-s)/4$, $\delta \coloneqq  1-\sqrt{\frac{2}{3}}$, $\Gamma' \coloneqq  (q+1)2 p_2 (n)$, $\eta \coloneqq  \left\lceil q \log \left( \frac{4 \Gamma'(\Gamma' + 1)}{\epsilon} - 1 \right) \right\rceil$, and $R \coloneqq  \left\lceil 2\left(\frac{2^{q+4}}{c-s}\right)^2 \right\rceil$.\\

\noindent \textbf{Protocol:}
\begin{enumerate}
    \item The prover sends a quantum state $\ket{\psi}$.
    \item The verifier runs the \textbf{variation} of~\cref{alg:reduction_QPCP}, with precision $\eta$ and maximum error probability $\delta$, to obtain a $q$-local Hamiltonian $\tilde{H}_x$.
    \item The verifier applies the transformation of~\cref{lem:H_smoothing} to obtain a $2q$-local Hamiltonian $\hat{H}_x = \frac{1}{m} \sum_{i} \hat{H}_{x,i}$.
    \item The verifier runs~\cref{prot:Kitaev_protocol} $R$ times for $\hat{H}_x$, and accepts if and only if at least a $\frac{(c-s)}{2^{q+4}}$-fraction of the outcomes equal $\ket{1}$.
\end{enumerate}
\label{prot:non-adap}
\end{protocol}

\

\begin{theorem}[Adaptive versus non-adaptive]
For any $c-s = \Omega(1)$ and $q = \mO(1)$, we have that
\begin{align*}
    \QPCP_{c,s}[q] \subseteq \QPCP_{\textup{NA}}[q']
\end{align*}
with
\[
q' = \mO\left(q\left(\frac{4^{q}}{c-s}\right)^2 \right).
\]
\label{thm:A_vs_NA}
\end{theorem}

\begin{proof}
We verify the correctness of~\cref{prot:non-adap}, which defines a $(q')$\nobreakdash-$\QPCP_{\textup{NA}}$ verifier for some $q'$, which we will show to be $\mO(q)$. Let $A = (A_\textup{\sc yes}, A_\textup{\sc no})$ be a promise problem in $\QPCP_{\textup{A}}[q]$ for an arbitrary constant $q$, and let $x \in \{0,1\}^n$ be an input.

\paragraph{Correctness.} Step 2 of the protocol produces a $q$-local Hamiltonian $\tilde{H}_x = \sum_{i \in [m]} \tilde{H}_i$ satisfying:
\begin{itemize}
    \item If $x \in A_\textup{\sc yes}$, then $\lambda_0(\tilde{H}_x) \leq s + \epsilon$,
    \item If $x \in A_\textup{\sc no}$, then $\lambda_0(\tilde{H}_x) \geq c - \epsilon$,
\end{itemize}
with probability at least $1 - \delta$. After applying the transformation from~\cref{lem:H_smoothing} to obtain $\hat{H}_x$ from $\tilde{H}_x$, we get:
\begin{itemize}
    \item If $x \in A_\textup{\sc yes}$, then $\lambda_0(\hat{H}_x) \leq \frac{s + \epsilon}{2^{q+3}} $,
    \item If $x \in A_\textup{\sc no}$, then $\lambda_0(\hat{H}_x) \geq \frac{c - \epsilon}{2^{q+3}} $,
\end{itemize}
Applying Kitaev's energy estimation protocol (\cref{prot:Kitaev_protocol}) to $\hat{Q}_x = \sum_{i} \frac{1}{m} \hat{H}_{x,i}$ yields a  $(2q)$\nobreakdash-$\QPCP_{\textup{NA}}$ verifier with a promise gap of at least
\[
\gamma \coloneqq \frac{(c-s) - 2\epsilon}{2^{q+3}} = \frac{(c-s)}{2^{q+4}} = \Theta(1),
\]
since $q$ was constant and $c-s = \Omega(1)$. Since Step~4 performs $R$ parallel runs of a $(2q)\text{-}\QPCP_{\textup{NA}}$ verifier, accepting if and only if a $\frac{c - s}{2^{q+4}}$-fraction of the verifiers accept, it forms an $(Rq)\text{-}\QPCP_{\textup{NA}}$ verifier with amplified completeness and soundness. By~\cref{lemma:weak_err_red}, and our choice of $R$ (see~\cref{eq:choice_of_R_repeat} in the proof of~\cref{lemma:weak_err_red}), it follows that, conditioned on Step 2 succeeding, the acceptance (resp.~rejection) probability in Step 4 is at least $\sqrt{2/3}$ in the {\sc yes}-case (resp.~{\sc no}-case). Moreover, we have that 
\[
q' = q R =  q \left\lceil 2\left(\frac{2^{q+4}}{c-s}\right)^2 \right\rceil = \mO\left(q\left(\frac{4^{q}}{c-s}\right)^2 \right).
\]
Since Step 2 succeeds with probability at least $\sqrt{2/3}$, the overall completeness (soundness) is at least $2/3$ (at most $1/3$).
\end{proof}

We conclude this subsection with the following observation.

\begin{remark}
The proof of~\cref{thm:A_vs_NA} actually demonstrates that using a uniform distribution (over a subset of all proof qubits) to decide which proof qubits to access suffices. Indeed, Kitaev's energy estimation protocol can be applied with $p_i = 1/m$ for all $i \in [m]$, where $m$ is the number of terms in the Hamiltonian. Since we still apply weak error reduction (\cref{lemma:weak_err_red}) to boost the completeness and soundness parameters back to their original values, the resulting distribution over queries becomes uniform over a subset of all possible index combinations. This is because parallel repetition restricts us to combinations where each supposed copy is used only once. This resolves an open question posed in~\cite{Aharonov2008The}, and shows that the definition of a quantum PCP given in~\cite{Grilo2018thesis} is in fact fully general.
\label{rem:uniform_QPCPs}
\end{remark}

\subsection{A multi-prover quantum PCP for \texorpdfstring{$\QMA[2]$}{QMA(2)} implies \texorpdfstring{$\QMA[2] = \QMA$}{QMA(2)=QMA}}
In this section we prove~\cref{thm:qma2}. Of course, if $\QMA[2]$ would allow a single-prover quantum PCP, i.e., $\QMA[2] \subseteq \QPCP[1,\mO(1)]$, then we would trivially have that $\QMA[2] = \QMA$ as $\QPCP[1,\mO(1)] \subseteq \QMA$. The point will be that for any polynomial number of unentangled provers, each providing a polynomially-sized proof, a quantum PCP that checks a constant number of qubits across the tensor product of these proofs can be simulated in $\QMA$.

The proof of this result is based on the ideas in~\cite{chailloux2012complexity} where it is proven that the $2$-separable local Hamiltonian problem is in $\QMA$. Since we need a slight generalisation of their result (namely from $2$-separable to $k$-separable), we include a full proof for completeness. The proof relies on the consistency of local density matrices problem ($\CLDM$), which will be introduced first. Moreover, for completeness and future reference, we give a new proof of $\QMA$-containment which holds for a larger range of problem parameter settings than previous proofs~\cite{liu2007consistency,broadbent2022qma}.

\subsubsection{The consistency of local density matrices problem}
We start by stating the definition of the consistency of local density matrices problem, first defined in~\cite{liu2007consistency}.

\begin{definition}[Consistency of local density matrices, $\CLDM(q,\alpha,\beta)$] 
\phantom{.}
\label[definition]{def:CLDM}
\begin{description}
    \item[Input:] A classical description of a collection of local density matrices $\{\rho_i\}_{i \in [m]}$ on $n$ qubits, $m = \poly(n)$, where each $\rho_i$ is a density matrix over qubits $C_i \subseteq [n]$ with $|C_i| \leq q$. For each $i \in [m]$, write $\overline{C}_i = [n] \setminus C_i$ for the complementary subset. Additionally, we are given two efficiently computable real numbers $\alpha, \beta$ such that $\beta-\alpha >0$.
    \item[Promise:] One of the following two cases holds:
    \begin{enumerate}[label=(\roman*)]
        \item There exists an $n$-qubit mixed state $\sigma$ such that for all $i \in [m]$,
        \[
        \norm{\tr_{\overline{C}_i}[\sigma] - \rho_i}_1 \leq \alpha.
        \]
        \item For all $n$-qubit mixed states $\sigma$, there exists an $i \in [m]$ such that
        \[
        \norm{\tr_{\overline{C}_i}[\sigma] - \rho_i}_1 \geq \beta.
        \]
    \end{enumerate}
    \item[Output:] {\sc yes} if (i) holds, and {\sc no} if (ii) holds.
\end{description}
If $\alpha = 0$, we write $\CLDM(q,\beta)$.
\end{definition}

Liu showed that $\CLDM$ is in $\QMA$ for  $\beta/4^q =: \epsilon = \Omega(1/\poly(n))$ and $\alpha = 0$. Broadbent and Grilo~\cite{broadbent2022qma} provided a proof for arbitrary $\alpha$ (and $\epsilon = \beta-\alpha$), but their proof contains a small error.\footnote{The proof relies on Hoeffding’s inequality to show soundness. However, the random variables it is used on are generally not independent in this setting, since the proof can be highly entangled. } Liu's proof can easily be modified to also incorporate the case for general $\alpha$, however, this still leads to a condition on $\beta$ which depends on $q$, leaving open as to whether $\CLDM(q,\beta,\alpha)$ is in $\QMA$ for any $\beta-\alpha = \Omega(1/\poly(n))$.\footnote{The initial arXiv version of the present paper contains this modified proof.}

We will now demonstrate that a different proof technique can be used to lift this restriction on $\beta$. Though not strictly needed to obtain our results, we include it for completeness and future reference. Our protocol is based on two ideas: (i) with a sufficient number of copies, one can learn the local marginals of any given state, and (ii) a specific formulation of a quantum de Finetti theorem under local measurements.

\subsubsection{Full state tomography of marginals}
The task of full state tomography is, given access to copies of an unknown quantum state~$\rho$, to learn an $\epsilon$-approximation~$\tilde{\rho}$ (with respect to some distance measure on quantum states), while minimising the number of copies required and, potentially, also the total processing time. Since we only care about overall efficiency, we will use the most basic state tomography protocol, as it allows for the simplest analysis. For an $n$-qubit system, a \emph{Pauli word} (also called a \emph{Pauli string}) is any operator of the form $
P = \sigma_1 \otimes \sigma_2 \otimes \dots \otimes \sigma_n$, where each $\sigma_i \in \{\mathbb{I}_2,X,Y,Z\}$ is a Pauli operator. For a Pauli word $P_j$, write $M_j$ for the measurement $M_j = \{  (P_j + \mathbb{I})/2, (\mathbb{I} -P_j)/2 \}$ with corresponding outcomes $\{+1,-1\}$. Hence, given a state $\rho$, we have that the random variable $X_j \in \{+1,-1\}$ corresponding to the measurement of $\rho$ using $M_j$, satisfies $\mathbb{E}[X_j] = \tr[\frac{1}{2}(P_j + \mathbb{I}) \rho](+1) + \tr[ \frac{1}{2}(\mathbb{I} -P_j) \rho](-1) = \tr[P_j \rho]$. Since the Pauli words $\{P_j\}_{j \in [d^2]}$ form a basis for the space of $d$-dimensional Hermitian matrices, we can write any density operator $\rho = \sum_{j \in [d^2]} c_j P_j $ where $c_j = \tr[P_j \rho]$. If instead we have estimates $\tilde{c}_j$ such that $\abs{\tilde{c}_j -c_j } \leq \epsilon$, then we have that $\tilde{\rho} = \sum_{j \in [d^2]} \tilde{c}_j P_j $ satisfies $\norm{\tilde{\rho} - \rho }_1 \leq d^2 \epsilon$.

\begin{lemma}
Let $ \rho \in \mathcal{D}(\mathbb{C}^{2^n}) $ and let $ \{C_i\}_{i \in [m]} $ be a collection of subsets of qubit indices, each satisfying $ |C_i| \leq q $. Write $\rho_i = \tr_{\overline{C}_i}[\rho]$. Then, there exists a measurement $M = \left\{ M^{(1)}_{a_1} \otimes \dots \otimes M^{(l)}_{a_l} \;\middle|\; a \in \{\pm 1\}^l \right\}$ on the state $ \rho^{\otimes l} $ and a classical algorithm running in $\poly(l) $ time that, given the measurement outcomes, outputs a classical description of $ \{\tilde{\rho}_i\}_{i \in [m]} $ satisfying
\begin{align*}
    \|\tilde{\rho}_i - \rho_i\|_1 \leq \epsilon \quad \text{for all } i \in [m],
\end{align*}
with probability at least $ 1 - \delta $. This uses
\begin{align*}
    l = \mathcal{O}\left(m q 16^q \log(m/\delta)/\epsilon^2 \right)
\end{align*}
copies of $ \rho $.
\label{lem:marginals_tomography}
\end{lemma}

\begin{proof}
The measurement consists of applying a tensor product of different Pauli measurements $M_j$ across the many copies of $\rho$ to obtain estimates $\tilde{c}_{j,i} \approx \tr[P_j \sigma_i]$, from which all the marginals $\rho_i$ can be approximately reconstructed. Since each measurement only acts on a single density matrix of a single copy of $\rho$, this ensures the tensor product structure of the overall measurement. Let $d_i = 2^{|C_i|}$, so that $\rho_i = \tr_{\overline{C}_i}[\rho]$ is the $d_i$-dimensional density matrix corresponding to the reduced density matrix of $\rho$ that has indices from $C_i$. We have $\rho_i = \sum_{j \in [d_i^2]} c_{j,i} P_j$, where $c_{j,i} = \tr[P_j \rho_i]$.  Using the measurement $M_j$ corresponding to $P_j$, where each measurement outcome is bounded, standard mean estimation (see for example~\cite{lee2020lecture}) gives an estimate $\tilde{c}_{j,i}$ such that $\abs{\tilde{c}_{j,i} - c_{j,i}} \leq \epsilon/d_i^2$ with probability $1-\delta/(md_i^2)$, using $\mO(d_i^4 \log(md_i^2/\delta)/\epsilon^2)$ copies of $\rho$. We set $\tilde{\rho}_i = \sum_{j \in [d_i^2]} \tilde{c}_{j,i} P_j$. By a union bound and converting a bound on the max norm to the trace norm, we must have that for each $i \in [m]$, $\norm{\tilde{\rho_i} - \rho_i}_1 \leq \epsilon$ holds with probability at least $1 - \delta/m$ for each $i$.  Another union bound shows that the probability of $\norm{\tilde{\rho_i} - \rho_i}_1 \leq \epsilon$ holding for all $i \in [m]$ simultaneously is at least $1 - \delta$. Using the upper bound $d_i \leq 2^q$ for all $i \in [m]$, the total number of copies (and thus measurements in the tensor product) can be upper bounded as
\begin{align*}
    \mO\left(m q 16^{q} \log(m/\delta)/\epsilon^2 \right).
\end{align*}
As the classical post-processing consists primarily of the addition of all obtained measurement outcomes, it clearly can be done in time $\poly(l)$.
\end{proof}

\subsubsection{Quantum de Finetti under local measurements} 
We first need to introduce some additional notation. For bipartite states $\rho^{XY}$, we use the convention that omitting subscripts corresponds to taking the partial trace over those systems; for example, $\rho^X = \tr_Y[\rho^{XY}]$. We say that $\rho^{A_1 \dots A_k}$ is permutation symmetric if $\rho^{A_{\pi(1)} \dots A_{\pi(k)}} = \rho^{A_1 \dots A_k}$ for any permutation $\pi \in \textup{S}_k$ (with $\textup{S}_k$ is the symmetric group of order $k$). We associate to any POVM $\{M_k\}$ a map $\Lambda(X) = \sum_k \text{Tr}(M_k X) \ketbra{k}$, where $\{|k\rangle\}$ is an orthonormal basis. Thus, the $\Lambda(X)$ are so-called quantum-classical channels, as they map density operators to other density matrices that are diagonal in the basis $\{\ket{k}\}$. This implies that for two states $\rho$ and $\sigma$, we have
\begin{align*}
    \frac{1}{2}\norm{\Lambda(\rho - \sigma)}_1 = D_{\textup{TV}} \left(\{p_k\}, \{q_k\}\right)  
\end{align*}
where $p_k = \tr{M_k \rho}$, $q_k = \tr{M_k \sigma}$ and $D_{\textup{TV}}(\cdot,\cdot)$ is the total variation distance. 

Informally, the classical de Finetti Theorem states that if the joint probability distribution of a sequence of random variables is invariant under any permutation of the variables, then the marginal probability distribution of a subset of $l \ll k$ variables from such a $k$-variable sequence will be close to a convex combination of i.i.d.~variables \cite{diaconis1980finite}. Quantum versions of the de Finetti Theorem posit that an $l$-partite quantum state $\rho^{A_1 \dots A_l}$, which is the reduced state of a permutation-symmetric state $\rho^{A_1 \dots A_k}$ on $k \gg l$ subsystems, is close to a convex combination of i.i.d.~quantum states, i.e., $\rho^{A_1 \dots A_l} \approx \int d\mu(\sigma) \, \sigma^{\otimes l}$,
with $\mu$ a probability measure on quantum states. Different quantum de Finetti theorems consider different notions of ``closeness'', e.g.,~\cite{renner2007symmetry}, \cite{christandl2007one}, and \cite{brandao2011faithful}. We will focus on closeness with respect to local measurements performed on the subsystems, for which a quantum de Finetti theorem was proven by Brand\~{a}o and Harrow \cite{brandao2013quantum}.

\begin{lemma}[\cite{brandao2013quantum}] Let $\rho^{A_1 \dots A_k} \in \mathcal{D}(A^{\otimes k})$ be a permutation-invariant state. Then for every $0 \leq l \leq k$, there is a measure $\nu$ (that depends on $\rho$) on $\mathcal{D}(A)$ such that
\begin{align*}
\max_{\Lambda_2, \dots, \Lambda_l} \left\| (\mathbb{I} \otimes \Lambda_2 \otimes \dots \otimes \Lambda_l) \left(\rho^{A_1 \dots A_l} - \int \nu(d\sigma) \sigma^{\otimes l} \right)\right\|_1 \leq \sqrt{\frac{2l^2 \ln |A|}{k-l}}.    
\end{align*}
\label[lemma]{lem:BH_QdeFinetii}
\end{lemma}
Given any input state on $k$ registers, the permutation-invariant assumption can always be enforced by randomly permuting all the input registers (this does of course alter the input state if the state was not already permutation-invariant).

 \subsubsection{The quantum marginal problem is in $\QMA$}
\label{sec:CLDM_QMA}
We now have all ingredients to give our $\QMA$ protocol for $\CLDM(q,\alpha,\beta)$, which is given in~\cref{prot:CLDM}.

The core idea is that, with enough copies of a state, the verifier can estimate its local marginals via tomography. Then, using~\cref{lem:BH_QdeFinetii}, we argue that any permutation-invariant state (which includes the verifier's post-processed state) is close to a separable state. Hence, from the verifier’s perspective, the tomography is effectively performed on a state that is nearly separable, even if the prover were to send a highly entangled state. From this, we can show that the acceptance probability of the protocol is very close to that of an idealised version in which the prover sends a mixture of actual tensor-product copies, from which soundness follows.
\\

\begin{protocol}[$\QMA$ protocol for $\CLDM(q,\alpha,\beta)$]
\noindent \textbf{Input:} Classical descriptions of the density matrices $\{\rho_i\}_{i\in [m]}$ and the indices $\{C_i\}_{i \in [m]}$ of the qubits on which they are supported, problem parameters $q$, $\beta$ and $\alpha$.\\

\noindent \textbf{Set:} $\epsilon \coloneqq(\beta-\alpha)/4$, $\delta\coloneqq 1/6$, $l\coloneqq \mO\left(m q 16^{q}\log(m/\delta)/\epsilon^2\right) $, \\
$k\coloneqq \frac{2 l \delta^2 + l^2 n \ln 2}{2 \delta^2}$.\\

\textbf{Protocol:} 
\begin{enumerate}
    \item The prover sends a state $\hat{\rho}^{A_1 \dots A_k} \in \mathcal{D}(A^{\otimes k})$.
    \item The verifier randomly permutes the index labels and traces out the last $k-l$ registers, creating the state $\rho^{A_1 \dots A_{l}}$.
    \item The verifier performs the measurement of \cref{lem:marginals_tomography}, with desired precision $\epsilon$ and success probability $1-\delta$, on the registers $A_1,\dots,A_{l}$ and uses the measurement outcome to construct $\{\tilde{\rho}_i\}_{i \in [m]}$.
    \item Accept if $\norm{\tilde{\rho}_i-\rho_i}_1 \leq \alpha + \epsilon $ for all $i \in [m]$, and reject otherwise.
\end{enumerate}
  \label{prot:CLDM}
\end{protocol}
\\

\begin{theorem} 
$\CLDM(q,\alpha,\beta)$ is in $\QMA$ for any $q = \mO(\log n)$ and $\beta - \alpha = \Omega(1/\poly(n))$.
\label{thm:CLDM_in_QMA}
\end{theorem}

\begin{proof}
We will show the correctness of \cref{prot:CLDM}.\\

\paragraph{Completeness.} 
In this case, the prover sends the state $\sigma^{\otimes k}$, which is already permutation-invariant. After Step $2$, the resulting state is $\sigma^{\otimes l}$. Hence, we have access to $l$ perfect copies of $\sigma$, meaning that by \cref{lem:marginals_tomography}, the estimates $\tilde{\sigma}_i$ of all $m$ density matrices $\sigma_i$ will be retrieved up to precision $\epsilon$ with high probability. Since we are in a {\sc yes}-instance, this means that, conditioned on the estimation procedure succeeding, we have $\norm{\tilde{\sigma}_i - \rho_i}_1 \leq \alpha + \epsilon$ for all $i \in [m]$, so we will accept with probability 1. Thus, the overall success probability is lower bounded by the success probability of the estimation procedure, which is $\geq 1 - \delta \geq 2/3$ for our choice of $\delta$. 

\paragraph{Soundness.}
Let the prover send any arbitrary state $ \hat{\rho}^{A_1 \dots A_k} $. Step 2 includes a random permutation of the subsystems, so the \emph{expected} state on the remaining $ l $ systems (after tracing out $ k - l $ registers) is the symmetrised marginal
\[
\frac{1}{k!} \sum_{\pi \in S_k} \tr_{A_{k - l + 1} \dots A_k} \left[ \pi \hat{\rho}^{A_1 \dots A_k} \pi^\dagger \right].
\]
Since the acceptance probability is linear in the input state, we may analyse it as though the prover had sent a symmetric state from the start. Thus, from here onward, we assume without loss of generality that $ \rho^{A_1 \dots A_l} $ is permutation-invariant.

Let $P(\rho^{A_1 \dots A_{l}})$ be the probability that the overall protocol in \cref{prot:CLDM} accepts. We will first argue that if the prover provides a tensor product of multiple copies of some state we will reject with high probability. Next, we will argue that by~\cref{lem:BH_QdeFinetii}, Step 2 ensures that any arbitrary (even entangled) state must be close to such a state, and thus will also be rejected with high probability. We now formalise this.

Just as in the completeness case, suppose the state after Step 2 is of the form
\[
    \rho^{A_1 \dots A_{l}} = \sigma^{\otimes l}
\]
for some arbitrary $\sigma \in \mathcal{D}(A)$. Conditioned on the success of Step 3, which occurs with probability $1 - \delta$, we will find a set $\{\tilde{\sigma}_i\}_{i \in [m]}$ such that $\norm{\tilde{\sigma}_i - \tr_{\overline{C}_i}\sigma}_1 \leq \epsilon$ for all $i \in [m]$. However, the promise implies that for any such $\sigma$, there must be some $i \in [m]$ such that $\norm{\tilde{\sigma}_i - \rho_i}_1 \geq \beta - \epsilon > \alpha + \epsilon$. Hence, we have $P(\sigma^{\otimes l}) \leq \delta$. By convexity, the same argument implies that $P\left( \int \nu(d\sigma) \sigma^{\otimes l} \right) \leq \delta$ for any measure $\nu$ on $\mathcal{D}(A)$.

Recall that the $\Lambda_j(\cdot)$ are quantum-classical channels, mapping quantum states to discrete probability distributions, where each channel implements a single measurement $M^{(j)} = \{M^{(j)}_{+1},M^{(j)}_{-1} \} $ corresponding to the one in~\cref{lem:marginals_tomography} (note that the index $j$ here does not represent the Pauli word, but labels the measurement instead). Step 4 can potentially distinguish between the two probability distributions described by the following two density matrices, which are diagonal in some orthonormal basis:
\[
    \rho_1 \coloneqq (  \Lambda_1 \otimes \dots \otimes \Lambda_{l})\left(\rho^{A_1 \dots A_{l}}\right),
\]
and
\[
    \rho_2 \coloneqq (\Lambda_1 \otimes \dots \otimes \Lambda_{l})\left(\int \nu(d\sigma) \sigma^{\otimes l}\right),
\]
for some measure $\nu$ as per~\cref{lem:BH_QdeFinetii}. Write $p_1$, $p_2$ for the distributions associated with measuring $\rho_1$, $\rho_2$ in their eigenbasis, respectively.  By \cref{lem:BH_QdeFinetii}, and using the fact that the total variation distance upper bounds the maximum bias in the output by a single-sample distinguisher, we have
\begin{align*}
    \abs{P(\rho^{A_1 \dots A_{l}}) - P\left( \int \nu(d\sigma) \sigma^{\otimes l} \right)} 
    &\leq  D_{\textup{TV}}(p_1, p_2),
\end{align*}
where
\begin{align*}
    D_{\textup{TV}}(p_1, p_2)&= D(\rho_1,\rho_2)  \\
    &= \frac{1}{2} \left\| (\Lambda_1 \otimes \dots \otimes \Lambda_{l})\left(\rho^{A_1 \dots A_{l}} - \int \nu(d\sigma) \sigma^{\otimes l} \right) \right\|_1\\
    &\leq \max_{\Lambda_2', \dots, \Lambda_l'} \left\| (\mathbb{I} \otimes \Lambda_2' \otimes \dots \otimes \Lambda_l') \left(\rho^{A_1 \dots A_l} - \int \nu(d\sigma) \sigma^{\otimes l} \right)\right\|_1 \\
    &\leq \frac{1}{2} \sqrt{\frac{2 l^2 \ln |A|}{k - l}}\\
    &\leq \delta,
\end{align*}
for our choice of parameters in~\cref{prot:CLDM}. Here we also used that the trace distance is non-increasing under the application of a quantum channel. Thus, for our choice of $\delta$, we have $P(\rho^{A_1 \dots A_{l}}) \leq 2 \delta  = 1/3$, as desired.

\paragraph{Complexity.}
Step 2 clearly takes time polynomial in $n$ when $k = \poly(n)$. Step 3 runs in polynomial time when $l$ is polynomial in $n$, and Step 4 runs in polynomial time when the sizes of the density matrices are at most polynomial in $n$ and $m = \poly(n)$. All conditions hold when $q = \mO(\log n)$, $\beta - \alpha \geq 1/\poly(n)$, and $m = \poly(n)$.
\end{proof}

\subsubsection{Simulating a multi-prover constant-query quantum PCP in \texorpdfstring{$\QMA$}{QMA}}
Before we prove the main result of this subsection, we need to show the following simple property of $\QMA$, which shows that one can compute the AND-function of multiple promise problems in $\QMA$ with a single $\QMA$-verifier.

\begin{lemma} Let $1 \leq l \leq \poly(n)$. Suppose that $A_1,A_2,\dots,A_l$ are promise problems in $\QMA$. Then the promise problem $B = (B_\textup{\sc yes}, B_\textup{\sc no})$ defined as
\begin{itemize}
    \item $x= (x_1,\dots,x_l) \in B_\textup{\sc yes}$, if $x_i \in A_{i, \textup{\sc yes}}$ for all $i\in [l]$;
    \item $x= (x_1,\dots,x_l) \in B_\textup{\sc no}$, if there exists an $j \in [l]$ such that $x_j \in A_{j, \textup{\sc no}}$, with $x_i \in \{A_{i, \textup{\sc yes }}, A_{i, \textup{\sc no}} \}$ for all $i\in [l]$;
\end{itemize}
is in $\QMA$.
\label[lemma]{lem:QMA_parallel}
\end{lemma}

\begin{proof}
Since $A_i \in \QMA$ for all $i \in [l]$, we have that for each $A_i$, and for every polynomial $p_2(n) \geq 1$, there exists a uniform family of quantum circuits $\{U_n^i \mid n \in \mathbb{N}\}$ such that:
\begin{enumerate}
    \item If $x_i \in A_{i,\textup{yes}}$, then there exists a proof $\ket{\psi_i}$ such that $\Pr[U_n^i \text{ accepts } \ket{\psi_i}] \geq 1 - 2^{-p_2(n)}$.
    \item If $x_i \in A_{i,\textup{no}}$, then for all quantum proofs $\ket{\psi_i}$, we have $\Pr[U_n^i \text{ accepts } \ket{\psi_i}] \leq 2^{-p_2(n)}$,
\end{enumerate}
by standard error reduction for $\QMA$. We will set $p_2$ later. We now define the verifier $U_n$ as follows: it expects a quantum proof $\bigotimes_{i \in [l]} \ket{\psi_i}$, runs all $U_n^i$ in parallel, measures all $l$ designated output qubits in the computational basis, and accepts if and only if all measurement outcomes are $\ket{1}$.

\paragraph{Case (i):} If $x_i \in A_{i,\textup{yes}}$ for all $i \in [l]$, then there exists a state $\bigotimes_{i \in [l]} \ket{\psi_i}$ such that
\begin{align*}
    \Pr[U_n \text{ accepts } \otimes_{i \in [l]} \ket{\psi_i}] \geq (1 - 2^{-p_2(n)})^l \geq 1 - 2^{-p_2(n)} l  \geq 1 - 2^{-p(n)}
\end{align*}
whenever $p_2(n) \geq p(n) + \log(l)$.

\paragraph{Case (ii):} Suppose there exists $j \in [l]$ such that $x_{j} \in A_{j,\textup{no}}$. We must argue that it does not help the prover to provide a highly entangled state. This follows from the same reasoning as in the proof that $\QMA$ admits weak error reduction. Let $\gamma$ be the total quantum proof and $C_i$, $i\in [l]$ be the index sets corresponding to the proof qubits to be used for input $x_i$. For all $i < j$, suppose the verifier has already executed the subprotocols and all of them accepted (otherwise we are done). Let the resulting density matrix on qubits from $C_{j}$ of the remaining state be $\gamma'_{\overline{C_{j}}}$. By convexity of acceptance probability, it follows that all mixed states $\gamma'_{\overline{C_{j}}}$ have acceptance probability at most $2^{-p_2(n)}$ for $U_n^{j}$, since this holds for all pure states. Hence, the acceptance probability at step $j$ is at most $2^{-p_2(n)}$, regardless of previous outcomes.  Since we set $p_2(n) \geq p(n)$ for case (i), this proves case (ii).

Since $1 \leq l \leq \poly(n)$, we can just set $p_2(n) \coloneqq p(n) + l$ so that both cases are satisfied Thus, $B$ is in $\QMA$.
\end{proof}

Next we show that the $k$-separable $q$-local Hamiltonian problem (Definition~\cref{def:LH_sep}) is in $\QMA$. This relies directly on the $\QMA$ containment of the $\CLDM$ problem. We first state a $\QMA$ protocol. \\

\begin{protocol}[$\QMA$ protocol for the $k$-separable $q$-local Hamiltonian problem]
\noindent \textbf{Input:} Classical descriptions of $m$ $q$-local terms $\{H_i\}$, completeness and soundness parameters $b, a$, and a number $k$.\\

\noindent  \textbf{Set} $a' \coloneqq a + \frac{b - a}{4}$, $b' \coloneqq b - \frac{b - a}{4}$, $\beta \coloneqq \frac{b - a}{2q m}$, $\alpha \coloneqq \beta / 8^q$, and $\delta \coloneqq \frac{k}{3}$.
\begin{enumerate}
    \item The prover sends a quantum witness $\gamma$ and a classical witness $\{\rho_i^j\}$.
    
    \item The verifier performs the following three checks:
    \begin{itemize}
        \item \textbf{Check 1:} It checks that each reduced density matrix $\rho_i^j$ is positive semidefinite and has trace one.
        
        \item \textbf{Check 2:} It checks that $\tr\left[H \otimes_{j \in [k]} \rho_i^j\right] \leq a'$.
        
        \item \textbf{Check 3:} It splits up the indices of the quantum witness $\gamma$ in $k$ disjoint sets $C_j$ of equal size, $j\in [k]$. It runs $k$ parallel executions of the $\QMA$ protocol given by~\cref{prot:CLDM}, each with completeness $1-\delta$ and soundness $\delta$, on the $k$ instances of the $\CLDM(q,\alpha, \beta)$ with respective input $\{\rho_i^j\}$ and proof $\tr_{\bar{C_j}}[\gamma]$. The check is passed if and only if all $\CLDM(q,\alpha, \beta)$ verifications accept.
    \end{itemize}
    
    \item The verifier accepts if and only if all three checks are passed.
\end{enumerate}
\label{prot:k_sep_LH}
\end{protocol}

\

\begin{lemma} The $k$-separable $q$-local Hamiltonian problem ($\textup{LH}(k,q,a,b)$) is in $\QMA$ for any $1 \leq k \leq \poly(n)$, $q = \mO(\log n)$ and $b-a \geq 1/\poly(n)$.
\label[lemma]{lem:sep_LH_QMA}
\end{lemma}
\begin{proof} This follows the same ideas as in~\cite{chailloux2012complexity}, but now for a $k$- instead of $2$-separable state. 

First, we observe that Check~3 defines the logical AND of multiple promise problems in~$\QMA$, so by the proof of~\cref{lem:QMA_parallel}, it follows that for our choice of~$\delta$, the overall procedure defined by Check~3 has completeness at least~$ 2/3 $ and soundness at most~$ 1/3 $.

\paragraph{Completeness.} By the promise, there exists a state $\xi = \otimes_{j \in [k]} \xi^j$ such that $\tr[H \xi] \leq a$. The prover sends as a quantum witness $ \gamma = \bigotimes_{j \in [k]} \gamma^j $, where each $ \gamma^j $ is a witness for the $\CLDM$ instance corresponding to $ \xi^j $. For the classical witness, the prover sends descriptions of $ \{ \tilde{\xi}_{i}^j \} $ of the $ q $-local reduced density matrices $ \xi_{i}^j $ of $ \xi $, each specified up to trace distance $ \epsilon = 2^{-p(n)} $ for some large polynomial $ p(n) $, such that the second and third checks are satisfied (see below).
This is possible, as the entries of each density matrix can be described using a polynomial number of bits. In this case, the protocol proceeds as follows:
\begin{itemize}
    \item The first check of~\cref{prot:k_sep_LH} is passed directly. 
    \item For the second check, we observe that
\begin{align*}
    \sum_{i \in [m]} \tr[H_i  \otimes_{j \in [k]} \tilde{\xi}_i^j]  \leq a + m k \epsilon \leq a',
\end{align*}
provided $\epsilon$ is chosen sufficiently small, which means that it also passes.
\item   For Check 3, by~\cref{lem:QMA_parallel}, this check also succeeds with probability at least $2/3$. Hence, the overall acceptance probability is at least $2/3$.
\end{itemize}

\paragraph{Soundness.} If Check 1 or Check 2 fails, we are done. If Check 1 succeeds, we can be certain that each matrix $\rho_i^j$ is a density matrix up to an exponentially small correction. If Check 2 succeeds, we must have that 
\begin{align*}
      \sum_{i \in [m]} \tr[H_i  \otimes_{j \in [k]} \rho_i^j] \leq a'.
\end{align*}
According to the promise, we have that for any state $\xi = \otimes_{j \in [k]} \xi^j$ 
\begin{align*}
     \sum_{i \in [m]} \tr[H_i  \otimes_{j \in [k]} \xi_i^j] \geq b'.
\end{align*}
This means that for our choice of $a'$, $b'$ we have
\begin{align*}
    \frac{b-a}{2} &\leq \sum_{i \in [m]} \left( \tr[H_i  \otimes_{j \in [k]} \xi_i^j] -  \tr[H_i \otimes_{j \in [k]} \rho_i^j] \right)\\
    &= \sum_{i \in [m]} \tr[H_i \left( \otimes_{j \in [k]} \xi_i^j -  \otimes_{j \in [k]} \rho_i^j \right)] \\
    &\leq \sum_{i \in [m]} \norm{\otimes_{j \in [k]} \xi_i^j] -  \otimes_{j \in [k]} \rho_i^j}_1
\end{align*}
which implies that there must exist an $i$ such that $ \frac{b-a}{2m} \leq \norm{ \otimes_{j \in [k]} \xi_i^j] -  \otimes_{j \in [k]} \rho_i^j}_1$. In the worst case, we have that $k \geq q$ and the qubits comprising this density matrix may be distributed across $q$ distinct proof registers. However, by the subadditivity property of the trace distance with respect to tensor products, we have that for any subset $S \subseteq [k]$ with $|S| \leq q$:
\begin{align*}
   \norm{ \otimes_{j \in S}\xi_i^j- \otimes_{j \in S} \rho_i^j}_1 \leq \sum_{j \in S} \norm{ \xi_i^j- \rho_i^j}_1,
\end{align*}
which implies that there must exist a $i,j$ pair such that $ \beta \coloneqq \frac{b-a}{2qm} \leq \norm{ \xi_i^j-\rho_i^j}_1 $, which satisfies the promise of a {\sc no}-instance for $\CLDM(q,\alpha, \beta)$. As we have already argued, Step 3 has a success probability of at least $ 2/3 $ for detecting a single {\sc no}-instance, by the proof of~\cref{lem:QMA_parallel}. Hence, we have that the overall acceptance probability is at most $1/3$.
\end{proof}

Finally, we arrive at the main result for this section. We will prove it by arguing that $\QPCP{[k,q]}$ is contained in $\QMA$ for any $k = \poly(n)$.\\

\begin{protocol}[$\QMA$ protocol for $\QPCP{[k,q]}$.]
\noindent \textbf{Input:} A classical description of a $(k,q,p_1,p_2,p_3)\text{-}\QPCP$ verifier $V_x$ with input $x$ hardcoded into it, with completeness $c$ and soundness $s$.  \\

\noindent \textbf{Set:} $\epsilon\coloneqq (c-s)/4$, $\delta = \delta' \coloneqq  1-\sqrt{\frac{2}{3}}$, $\Gamma' \coloneqq  (q+1)2k p_2(n)$, $\eta \coloneqq  \left\lceil q \log \left( \frac{4 \Gamma'(\Gamma' + 1)}{\epsilon} - 1 \right) \right\rceil$.\\

\noindent \textbf{Protocol:}
\begin{enumerate}
    \item The prover sends the witness  $\xi$.
   \item The verifier runs the \textbf{variation} of~\cref{alg:reduction_QPCP}, with precision $\eta$ and maximum error probability $\delta$, to obtain a $q$-local Hamiltonian $\tilde{H}_x$.
    \item The verifier runs~\cref{prot:k_sep_LH} for Hamiltonian $\tilde{H}_x$ with completeness $1-\delta'$ and soundness $\delta'$, with $a =  s + \epsilon$ and $b= c - \epsilon $ , and accepts if and only if~\cref{prot:k_sep_LH} accepts.
\end{enumerate}
  \label{prot:QMA2}
\end{protocol}

\

\begin{theorem}[PCPs for ${\QMA[2]}$]
If there exists a $2 \leq k' \leq \poly(n)$ and $q = \mathcal{O}(1)$ such that $\QMA[2] \subseteq \QPCP[k', q]$, then $\QMA[2] = \QMA$.
\label{thm:qma2}
\end{theorem}

\begin{proof}
Suppose that such a $k'$ and $q$ indeed exist.
Let $A = (A_{\textup{\sc yes}}, A_{\textup{\sc no}})$ be any problem in $\QPCP[k', q]$, with verifier $V$ and input $x$. We have already verified in the proof of~\cref{thm:A_vs_NA} that Step 2 of~\cref{prot:QMA2} can be made to with probability at least $1 - \delta' \geq \sqrt{2/3}$, producing a fixed Hamiltonian when succeeding, and Step 3 can be made to succeed with probability at least $1 - \delta \geq \sqrt{2/3}$ by standard error reduction for $\QMA$ (the arguments holds for any $1 \leq k \leq \poly(n)$). Thus, we have:
\begin{itemize}
    \item If $x \in A_{\textup{\sc yes}}$, then $\Pr[\text{\cref{prot:QMA2} accepts}] \geq 2/3$,
    \item If $x \in A_{\textup{\sc no}}$, then $\Pr[\text{\cref{prot:QMA2} accepts}] \leq 1/3$,
\end{itemize}
which implies $A \in \QMA$, and hence $\QPCP[k', q] \subseteq \QMA$. Hence, since the assumption implies  $\QMA[2] \subseteq \QPCP[k', q] \subseteq \QMA$ and $\QMA \subseteq \QMA[2]$ holds trivially, the result follows.
\end{proof}

\section{Proofs of quantum PCP theorems do not (quantumly) relativize}
\label{sec:oracle_seps}
\subsection{Quantum proof discretisation}

We start by recalling the definition of $h$-nets on pure states.
\begin{definition}[$h$-net] Let $\mathcal{H}$ be a $d$-dimensional Hilbert space. We say a set $P^d_h = \{ \ket{\psi_i}\} \subseteq \mathcal{P}(\mathcal{H})$ forms an $h$-net for $\mathcal{P}(\mathcal{H})$ if for all $\ket{\phi} \in\mathcal{P}(\mathcal{H})$ there exists a $\ket{\psi_i} \in P^d_h$ such that $\abs{\bra{\phi}\ket{\psi_i}}\geq h$.
\label{def:hnet}
\end{definition}
It is known that for all $h$ there exists such a $h$-net of size that scales roughly exponentially in the dimension, as illustrated by the following lemma.
\begin{lemma}[\cite{boroczky2003covering}] Let $\mathcal{H}$ be a $d$-dimensional Hilbert space. For all $0< h < 1$, there exists a $h$-net $P^d_h$ for $\mathcal{P}(\mathcal{H})$ of size smaller than
\begin{align*}
   C' \left( \frac{d^{3/2} \log(2+d h^2)}{(1-h^2)^d} \right),
\end{align*}
for some universal constant $C' >0$.
\label{lem:hnet_size}
\end{lemma}
To the best of our knowledge, no explicit construction has yet been found that achieves this scaling in $d$.

We give an analogous definition of $h$-nets for mixed states with respect to trace distance, which we call \textit{$\epsilon$-covering sets of density matrices}.
\begin{definition}[$\epsilon$-covering set of density matrices] Let $\mathcal{H}$ be some $d$-dimensional Hilbert space. We say a set of density matrices $D^d_\epsilon = \{\rho_i \} \subseteq \mathcal{D}(\mathcal{H})$ is $\epsilon$-covering for $\mathcal{D}(\mathcal{H})$ if for all $\sigma \in \mathcal{D}(\mathcal{H})$ there exists a $\rho \in D^d_\epsilon$ such that $\frac{1}{2}\norm{\rho - \sigma}_1 \leq \epsilon$.
\label{def:PEDMS}
\end{definition}
It is easy to show that the existence of an $h$-net implies the existence of a $\epsilon$-covering set of density matrices with $\epsilon = \sqrt{1-h^2}$.
\begin{proposition} For all $0< \epsilon < 1$, there exists a $\epsilon$-covering set of density matrices $D_\epsilon^d$ of size smaller than
\begin{align*}
    C \left(\frac{1}{\epsilon}\right)^{5d },
\end{align*}
for some universal constant $C >0$.
\end{proposition}
\begin{proof}
One can view the $h$-net $P_\epsilon^{2d}$ of~\cref{lem:hnet_size} on a $2d$-dimensional Hilbert space $\mathcal{H} = \mathcal{H}_A \otimes \mathcal{H}_B$, both of equal dimension $d$, as containing purifications of density matrices from the set $D_\epsilon^d \subseteq \mathcal{D}(\mathcal{H}_A)$.  Let $\rho$ be some arbitrary mixed state in $\mathcal{D}(\mathcal{H})$ and let  $\ket{\psi} \in \mathcal{P}(\mathcal{H} \otimes \mathcal{H})$ be a purification of $\rho$. Then there must exists a state $\ket{\phi} \in P_\epsilon^{2d}$, with reduced density matrix $\sigma$, such that
\begin{align*}
    \abs{\bra{\phi}\ket{\psi}}\geq h \iff 1-\abs{\bra{\phi}\ket{\psi}}^2 \leq 1-h^2,
\end{align*}
which means that setting $h := \sqrt{1-\epsilon^2}$ should have that
\begin{align*}
    \epsilon \geq \sqrt{1-|\bra{\phi}\ket{\psi}|^2} = \frac{1}{2}\norm{\ket{\psi}\bra{\psi}-\ket{\phi}\bra{\phi}}_1 \geq \frac{1}{2}\norm{\rho-\sigma}_1,
\end{align*}
since the trace distance can only decrease under partial trace. Hence, setting 
\begin{align*}
{D_\epsilon^d := \{ \tr_\textup{B}[\ket{\phi}\bra{\phi}] \text{ $:$ } \ket{\phi} \in P_\epsilon^{2d} \}}    
\end{align*}
satisfies conditions of~\cref{def:PEDMS}. We can now upper bound the cardinality of $D_\epsilon^d$ using~\cref{lem:hnet_size} as

\begin{align*}
   |D^d_\epsilon| &\leq C' \left( \frac{(2d)^{3/2} \log(2+2d (1-\epsilon^2))}{(\epsilon^2)^{2d}} \right)\\
   &\leq C \left(\frac{1}{\epsilon}\right)^{5d}
\end{align*}
for some universal constant $C >0$.
\label{prop:dens_mat_set}
\end{proof}
Let us now formally define the proof-discretised version of $\QPCP$, denoted as $\QPCP_\epsilon$. We will use a non-adaptive formulation of the quantum PCP, since our proof of $\QPCP[\mO(1)] = \QPCP_\textup{NA}[\mO(1)]$ in~\cref{thm:A_vs_NA} relativizes.
\begin{definition}[Proof-discretised $\QPCP$] Let $q \in \mathbb{N}$, $\epsilon \in [0,1]$ be some constant, and let $\mathcal{H} = \mathbb{C}^{2^q}$. Let $D^{2^q}_\epsilon$ be a $\epsilon$-covering density matrix set as in~\cref{prop:dens_mat_set}. Then we define $\QPCP_{\epsilon}[k,q,c,s]$ as the complexity class which is just as $\QPCP[k,q,c,s]$, but where the reduced density matrix $\rho \in \mathcal{D}(\mathcal{H})$ obtained from the proof is projected (and subsequently re-normalised) on the density matrix $\rho' \in D^{2^q}_\epsilon$ closest to $\rho$ in trace distance. If $c=2/3$ and $s=1/3$, we simply write $\QPCP_\epsilon[q]$.
\end{definition}

Whilst $\QPCP_\epsilon$ is a highly artificial class, it will be helpful for us to prove some results since we can now use counting arguments on the exact size of the proof space that Merlin sees. However, note that the definition of this does \emph{not} put any limitations on the proofs that Merlin can send, as the projection on a state from the set $D_\epsilon^{2^q}$ (and following re-normalisation) happens only after the proof has been received by the verifier. Therefore, the verifier has the guarantee that Merlin sends a valid quantum state, circumventing the problem of checking the consistency of the local density matrices with a global quantum state (which is $\QMA$-hard by the results of~\cite{liu2007consistency,broadbent2022qma}).

Let us now justify that, with respect to all (quantum) oracles,  we have that putting this extra restriction does not limit the power of $\QPCP_\epsilon$ as long as the $\epsilon$-covering set is large enough.

\begin{lemma} Let $q \in \mathbb{N}$, $\epsilon \in (0,1/6)$ constant and let $D^{2^q}_\epsilon$ be a $\epsilon$-covering density matrix set as in~\cref{prop:dens_mat_set}. We have that for all quantum oracles $U$ that
\[ 
\QPCP^U[k,q] \subseteq \QPCP^U_{\epsilon}[k,q,2/3-\epsilon,1/3+\epsilon],
\]
\label{lem:just_proof_disc}
\end{lemma}
\begin{proof}
Since the proof of $\QPCP[\mO(1)] = \QPCP_\textup{NA} [\mO(1)]$ relativizes (\cref{thm:A_vs_NA}), our assumptions that the PCP is non-adaptive is w.l.o.g. even in a relativized setting. Let $V_x^U$ be a $\QPCP_\textup{NA}^U$ verification circuit as in~\cref{def:adap_QPCP} (but with the input $x$ hardcoded into it) that uses $p_1(n)$ ancilla qubits and has access to the quantum oracle $U$, and $x \in \{0,1\}^n$. We define
    \begin{align*}
        p_x(i_1,\dots,i_q) := \norm{\Pi_{i_1,\dots,i_q} V_{x,0}^U \ket{x} \ket{0}^{\otimes p_1(n)}}^2.
    \end{align*}
    Let $\xi_{i_1,\dots,i_q} = \tr_{C_{i_1,\dots,i_q}}[\xi]$ and $\hat{\xi}_{i_1,\dots,i_q}$ be the reduced density matrix that is accessed by $\QPCP^U_\epsilon$ instead. Next, define $\sigma_{i_1,\dots,i_q} = \xi_{i_1,\dots,i_q} \otimes \rho_{i_1,\dots,i_q} $ where $\rho_{i_1,\dots,i_q}$  is the $p_1(n)+n$-qubit post-measurement state after the measurement outcome to decide the indices of the proof to be queried returned $i_1,\dots,i_q$. In similar fashion, define $\tilde{\sigma}_{i_1,\dots,i_q}$ using $\hat{\xi}_{i_1,\dots,i_q}$. We have
    \begin{align*}
        &\abs{\Pr[V_x^U \text{accepts } \xi] -\Pr[\hat{V}_x^U \text{accepts } \xi] }\\
        &\hspace{6em}=  \sum_{\{i_1,\dots,i_q\} \in \Omega} p_x(i_1,\dots,i_q) \left( \tr[\Pi^1_{\textup{output}} V_{x,1}^{U\dagger} (\sigma{i_1,\dots,i_q}-\hat{\sigma}_{i_1,\dots,i_q} ) V_{x,1}^U \Pi^1_{\textup{output}}] \right)\\
         &\hspace{6em}= \sum_{\{i_1,\dots,i_q\} \in \Omega} p_x(i_1,\dots,i_q) \left( \tr[ V_{x,1}^U  \Pi^1_{\textup{output}} V_{x,1}^{U\dagger} (\sigma_{i_1,\dots,i_q}-\hat{\sigma}_{i_1,\dots,i_q} )] \right)\\
         &\hspace{6em}\leq \sum_{\{i_1,\dots,i_q\} \in \Omega} p_x(i_1,\dots,i_q) \norm{\sigma_{i_1,\dots,i_q}-\hat{\sigma}_{i_1,\dots,i_q}}_1\\
         &\hspace{6em}= \sum_{\{i_1,\dots,i_q\} \in \Omega} p_x(i_1,\dots,i_q) \norm{\xi_{i_1,\dots,i_q}-\hat{\xi}_{i_1,\dots,i_q}}_1\\
         &\hspace{6em}\leq \epsilon,
    \end{align*}
    where we used the cyclic property of the trace in going from line 1 to line 2 and that the trace distance is preserved under tensor products (with the same state) in going from line 3 to line 4.  
\end{proof}

\subsection{Oracle separations}
We will now follow the proof of the oracle separation in~\cite{aaronson2007quantum} to show that the same oracle $U$ that shows $\QCMA^U \neq \QMA^U$ also shows that $\QPCP[\mO(1)]^U \neq \QMA^U $. We will omit most of the details, focusing on the parts where we depart from~\cite{aaronson2007quantum}. 

Let $\mu$ be the uniform probability measure over $\mathcal{P}(\mathcal{H})$. We will need the notion of $p$-uniform probability measures, which have the following definition.
\begin{definition}[$p$-uniform probability measures] For all $p \in [0,1]$, a probability measure $\sigma$ over $\mathcal{P}(\mathcal{H})$ is called $p$-uniform if $p \sigma \leq \mu$. Equivalently, $\sigma$ is  $p$-uniform if it can be obtained by starting from $\mu$, and then conditioning on an event that occurs with probability at least $p$.
\end{definition}

We also state the following lemma from \cite{aaronson2007quantum}.
\begin{lemma}[Expected fidelity $p$-uniform random states~\cite{aaronson2007quantum}] Let $\mathcal{H}$ be a $d$-dimensional Hilbert space and let $\sigma$ be a $p$-uniform probability measure over $\mathcal{P}(\mathcal{H})$. Then for all mixed states $\rho$, we have that
\begin{align*}
    \mathbb{E}_{\ket{\Psi}\in \sigma} \bra{\Psi}\rho \ket{\Psi} = \mO \left(\frac{1+\log 1/p}{d}\right).
\end{align*}
\label{lem:geometric_lemma}
\end{lemma}
The key idea behind the oracle separation in~\cite{aaronson2007quantum} is to provide a lower bound on a decision version of ``quantum search'', where one is given a black-box unitary which either marks a single $n$-qubit quantum state or applies the identity operation. This is reminiscent of the lower bound on search by using the OR-function, where one is given a black-box Boolean function which ``marks'' a single bit string (i.e.~$f(x)=1$ for a single $x$) or for all bit strings outputs $0$ ($f(x)$ for all $x$). By exploiting the fact that there are a doubly exponential number of quantum states that have small pairwise overlap,~\cite{aaronson2007quantum} can show that a polynomial number of classical bits to point you towards the right quantum state does not help much, and still requires one to make an exponential number of queries to the oracle. 
\begin{lemma} Let $q \in \mathbb{N}$ and let $D_\epsilon^{2^q}$ be a $\epsilon$-covering set of mixed states on $q$ qubits as per~\cref{prop:dens_mat_set}. Suppose we are given oracle access to an $n$-qubit unitary $U$, and want to decide which of the following holds:
\begin{enumerate}[label=(\roman*)]
    \item There exists an $n$-qubit ``quantum marked state'' $\ket{\Psi}$ such that $U \ket{\Psi} = - \ket{\Psi}$, but $U \ket{\phi} = \ket{\phi}$ whenever $\bra{\phi}\ket{\Psi}=0$; or
    \item $U=I$ is the identity operator. 
\end{enumerate}
Then even if we have a witness tuple $(i_1,\dots,i_q,\rho)$, where $i_{i_1},\dots,i_{i_q}$ is a set of the $q$ indices and a quantum witness $\rho \in D_\epsilon^{2^q}$ in support of case (i), we still need 
\begin{align*}
 \Omega \left( \sqrt{\frac{2^n}{5 \cdot 2^q \log(\frac{1}{\epsilon}) +q \cdot \polylog(n,q)  + C} } \right).
\end{align*}
queries to verify the witness, with a bounded probability of error. Here $C >0 $ is some universal constant.
\label{lem:lb_qproof}
\end{lemma}
\begin{proof}
This follows the same structure as Theorem 3.3 in~\cite{aaronson2007quantum}, but considers $q$-qubit proofs from $D_\epsilon^{2^q}$ instead of classical proofs. Let $\mathcal{H}$ be a $2^n$-dimensional Hilbert space. Let $A$ be a quantum algorithm that queries the oracle $U$ a total of $T$ times, with the goal to output $1$ in case (i) and output 0 in case (ii). For each $n$-qubit state $\ket{\Psi}$, we fix both the indices $(i_1,\dots,i_q), i_j \in [p(n)]$ for all $j \in [q]$ \emph{and} corresponding quantum witness $\rho \in D_\epsilon^{2^q}$ that maximises the difference in the probability that $A$ accepts as compared to the case where $U$ is the identity. This is allowed, since by~\cref{rem:uniform_QPCPs} the choice of distribution over indices does not have to depend on the oracle nor the input. The expected difference in accepting given $U_{\ket{\Psi}}$ or $I$ satisfies
\begin{align*}
    g &\leq \max_{\xi}
   \mathbb{E}_{i_1,\dots,i_q} \Pr[V^{U_{\ket{\Psi}}} \text{ accepts querying } \xi_{i_1, \dots, i_q}] - \mathbb{E}_{i_1,\dots,i_1} \Pr[V^I \text{ accepts querying } \xi_{i_1, \dots, i_q}] \\
   &= \max_{\xi}
   \mathbb{E}_{i_1,\dots,i_q} [\Pr[V^{U_{\ket{\Psi}}} \text{ accepts querying } \xi_{i_1, \dots, i_q}] - \Pr[V^I \text{ accepts querying } \xi_{i_1, \dots, i_q}]]\\
   &\leq \max_{\xi} \max_{i_1,\dots,i_q}
  \Pr[V^{U_{\ket{\Psi}}} \text{ accepts querying } \xi_{i_1, \dots, i_q}] - \Pr[V^I \text{ accepts querying } \xi_{i_1, \dots, i_q}],
\end{align*}
so we can set $\rho = \xi_{i_1, \dots, i_q}$ to be the density matrix as in the last line, and our lower bound would also apply in the setting where the indices are chosen probabilistically. Here $g = \Omega(1)$ is the promise gap between both cases. If $g > \frac{1}{2}$, then we have that the completeness satisfies $c \geq g$ and the soundness $s \leq 1-g$ with $c -s = \Omega(1)$. Let $S(i_1,\dots,i_q, \rho)$ be the set of $\ket{\Psi}$'s that belong to the index/state tuple $(i_1,\dots,i_q,\rho)$. Then we have that the $S(i_1,\dots,i_q,\rho)$'s form a partition of $\mathcal{P}(\mathcal{H})$, and hence we must have that there exists a $(i_1^*,\dots,i_q^*,\rho^{*})$ such that
\begin{align*}
    \text{Pr}_{\ket{\Psi}\in \mu} \left[ \ket{\Psi} \in S(i_1^*,\dots,i_q^*,\rho^*)\right] &\geq  \left(\binom{p(n)}{q} C \left( \frac{1}{\epsilon}\right)^{5 \cdot 2^q}\right)^{-1}\\
    &\geq \left( \left(\frac{e p(n)}{q}\right)^q C \left( \frac{1}{\epsilon}\right)^{5 \cdot 2^q}\right)^{-1} =:p
\end{align*}
Since $\rho$ is a $q$-qubit quantum state, there exists a unitary $U$ of circuit depth $\mO(4^q \poly(n))$  such that $U\ket{0}^{\otimes 2q} = \ket{\phi}$ and $\norm{\tr_B[\ket{\phi}\bra{\phi}] - \rho}_1 \leq 1/\exp(n)$ (by the Solovay-Kitaev Theorem). We hardcode $U$ and the information $i_q,\dots,i_q$ into $A$. Now let $\sigma$ be the probability measure over $S(i_1^*,\dots,i_q^*, \rho^*)$. Note that it is $p$-uniform. Following the rest of the steps of the proof in~\cite[Theorem 3.3]{aaronson2007quantum}, this yields the desired bound of 
\begin{align*}
   T \geq \Omega \left( \sqrt{\frac{2^n}{5 \cdot 2^q \log(\frac{1}{\epsilon}) +q \cdot \polylog(n,q)  + C} } \right).
\end{align*}
\end{proof}

\begin{theorem} There exists a quantum oracle $U$ relative to which $\QCMA^U \neq \QMA^U$ and $\QPCP[k,q]^U \neq \QMA^U$, for all $q \in  \mO(\log n)$, $1 \leq k \leq \poly(n)$.
\label{thm:q_oracle_sep}
\end{theorem}
\begin{proof}
This follows from using the same proof as~\cite{aaronson2007quantum}, Theorem  1.1, using the lower bound of~\cref{lem:lb_qproof} to show that $\QPCP_{P^{k'}_\epsilon}[q]^U \neq \QMA^U$ for all constant $q \in \mO(\log n)$, from which the separation $\QPCP[q]^U \neq \QMA^U$ follows by~\cref{lem:just_proof_disc}.
\end{proof}

We remark that our proof technique should work for any oracle separation between $\QCMA$ and $\QMA$ that uses counting arguments exploiting the doubly exponentially large number of quantum states and the fact that $\QCMA$ has only access to an exponential number of proofs. This shows that any proof of the quantum PCP conjecture (via the proof verification formulation) requires (as expected) quantumly non-relativising techniques.

The previous quantum oracle separations crucially exploit the doubly exponentially large number of quantum states that have low mutual fidelities. However, one may wonder whether a similar idea might be used to also show that for low-complexity states a similar separation holds? It turns out that this is surprisingly easy, by simply using a similar oracle to the one that separates $\BQP$ from $\NP$~\cite{bennett1997strengths}.  Observe this oracle is classical, and can alternatively be viewed as a state $1$-design over the set of all quantum states used in the previous separation. We state the following result, for which the proof is given in~\cref{app:classical_or_sep}.

\begin{restatable}{corollary}{classorsep} There exists a classical oracle $\mathcal{O}$ relative to which $\QPCP[k,q]^{\mathcal{O}} \neq \QMA^{\mathcal{O}}$ (in fact $\neq \NP^{\mathcal{O}}$) for all $q \in \mO(1)$.
\label{cor:classical_oracle}
\end{restatable}

Note that a much stronger statement also holds, we can even combine $\QCMA$ and $\QPCP$ (so a classical proof and quantum proof which can be accessed locally) and the same separation would still hold.

\bibliography{main.bib}

\appendix
\section{Strong error reduction for non-adaptive quantum PCPs with near-perfect completeness}
\label{app:strong_err_red}
In this appendix we argue that non-adaptive quantum PCPs allow for strong error reduction, so long as they are nearly perfectly complete. We will use the well-known gentle measurement lemma. 
\begin{lemma}[Gentle measurement lemma] Consider a state $\rho$ and a measurement operator $\Lambda$ where $0 \preceq \Lambda \preceq 1$. Now suppose that $\tr[\Lambda \rho] \geq 1-\epsilon$, where $0 < \epsilon \leq 1$. Then the post-measurement state
\begin{align*}
    \rho' = \frac{\sqrt{\Lambda} \rho \sqrt{\Lambda}}{\tr[\Lambda \rho]}
\end{align*}
satisfies $\norm{\rho - \rho'}_1 \leq 2 \sqrt{\epsilon}$.
\label{lem:GML}
\end{lemma}

The error reduction protocol will be a variant of the parallel repetition protocol but now we will use the \emph{same} quantum proof in all repetitions.

\strongerrredu* 
\begin{proof}
Let $\mathcal{A}$ be a $(q)$-$\QPCP_{\textup{NA}}$ verifier consisting of circuits $V_0$ and $V_1$, and let $\xi$ be the provided quantum proof. Let $\rho^0_{i_1,\dots,i_q} \ket{i_1} \bra{i_1} \dots \ket{i_q} \bra{i_q} \otimes \rho_{\text{rest}}$ be the post-measurement state after PVM $\Pi_1$ is performed to obtain indices $i_1,\dots,i_q$.  Now let $\rho^1_{i_1,\dots,i_q}$ be the state after $V_1$ acts on $\ket{i_1} \bra{i_1} \dots \ket{i_q} \bra{i_q} \otimes \rho_1$ and qubits $i_1,\dots,i_q$ of $\xi$. In the {\sc no}-case, we have that the soundness property holds with respect to all proof states. So for every reduced density matrix $\xi'$ after the measurement is performed, we have 
\begin{align*}
    \mathbb{E}_{i_1,\dots,i_q} \left[\tr[ \Pi^1_{\text{output}} \rho^0_{i_1,\dots,i_q}] \right] \leq 1/2.
\end{align*}
Hence, the bound follows from the same argument as in the weak error reduction case (\cref{lemma:weak_err_red}).

In the {\sc yes}-case, we have that $\Pi_0^{\text{output}}$ satisfies
\begin{align*}
    \mathbb{E}_{i_1,\dots,i_q} \left[\tr[ \Pi^1_{\text{output}} \rho^0_{i_1,\dots,i_q}] \right] \geq 1-2^{-c_0 n}
\end{align*}
which implies that
\begin{align*}
    \mathbb{P}_{i_1,\dots,i_q}[\tr[ \Pi^1_{\text{output}} \rho^0_{i_1,\dots,i_q}] \geq  1-2^{-c_0 n/2} ] \geq 1 - 2^{c_0n/2}.
\end{align*}
By~\cref{lem:GML}, we have that the post-measurement state $\rho^{'}_{i_1,\dots,i_q}$ satisfies 
\begin{align*}
    \frac{1}{2}\norm{\rho^{1}_{i_1,\dots,i_q} - \rho^{'}_{i_1,\dots,i_q}}_1 \leq 2^{-n c_0/4}.
\end{align*}
with probability $\geq 1 - 2^{c_0 n/2} $.
Hence, we can simply apply $V_1^\dagger$ such that the density matrix $\xi'$ in the proof register satisfies
\begin{align*}
     \frac{1}{2}\norm{\xi' - \xi}_1 \leq 2^{-c_0 n/4},
\end{align*}
using the fact that the trace distance can only decrease under the partial trace.
Therefore, our new acceptance probability satisfies
\begin{align*}
    \mathbb{E}_{i_1,\dots,i_q} \left[\tr[ \Pi^1_{\text{output}} \rho^0_{i_1,\dots,i_q}] \right] \geq 1-2^{-c_0 n} - 2^{-c_0n/4} \geq 1-2^{-c_1 n}
\end{align*}
for some constant $c_1 >c_0$. Hence, after $l = \mO(1)$ of such steps, we have that with probability at least
\begin{align*}
    \prod_{i \in [l]} 1 - 2^{-c_i n/2} \geq \left(1 - 2^{-c_l n/2}\right)^l
\end{align*}
each step had an acceptance probability $\geq 1-2^{-c_1 n}$. Hence, the majority vote definitely accepts with probability at least
\begin{align*}
   \left(1 - 2^{-c_l n/2}\right)^l \left(1 - 2^{-c_l n}\right)^l \geq 1-2^{-\Omega(n)},
\end{align*}
when $l \in \mO(1)$, as desired.
\end{proof}

\section{Sufficient error bounds for learning weighted Hamiltonians from non-adaptive quantum PCPs}
\label{app:weighted_error_conditions}
In the lemma proven in this appendix we give sufficient parameters to adopt the quantum reduction of~\cite{weggemans2023guidable} to the $\QPCP_{\textup{NA}}$ setting.
\begin{lemma}
Let $H = \sum_{i \in [m]} p_i H_i$ be a $k$-local Hamiltonian consisting of weights $p_i \in [0,1]$ such that $\sum_{i \in [m]} p_i = 1$, and $k$-local terms $H_i$ for which $\norm{H_i} \leq 1$ for all $i \in [m]$. Suppose $\tilde{H} = \sum_{i \in [m]} \tilde{p}_i \tilde{H}_i$ is another Hamiltonian such that, for all $i \in [m]$, we have $\tilde{H}_i$ PSD, $\abs{\tilde{p}_i -p_i} \leq \epsilon_0$ and  $\norm{H_i - \tilde{H}_i } \leq \epsilon_1$. Let $\tilde{W} = \sum_{i \in [m]} \tilde{p}_i$ and, 
\begin{align*}
    \hat{p}_i = \frac{\tilde{p}_i}{\tilde{W}},   \quad \hat{H}_i =\frac{\tilde{H}_i}{\mathrm{argmax}_{i \in [m]}\{ \tilde{H}_i,1\}}
\end{align*}
 for all $i \in [m]$. Then we have that setting
   \begin{align*}
    \epsilon_0 \leq \frac{\epsilon}{8m^2}  \quad \epsilon_1 \leq \frac{\epsilon}{6},
\end{align*}
suffices to have 
\begin{align*}
    \norm{H_x - \tilde{H}_x } \leq \epsilon.
\end{align*}
\label{lem:weighted_error_conditions}
\end{lemma}
\begin{proof}
Suppose $\epsilon_0 < 1/2m$. Then
\begin{align*}
    \abs{\hat{p}_i - p_i} &= \abs{\frac{\tilde{p_i}}{ \tilde{W}} - p_i}\\
    &\leq \frac{p_i + \epsilon_0}{1-m\epsilon_0} - p_i\\
    &=  \frac{\epsilon_0(m p_i +1)}{1-m \epsilon_0} \\
    &\leq \frac{\epsilon_0(m+1)}{1-m \epsilon_0} \\
    &\leq 4 \epsilon_0 m.
\end{align*}
For the local terms we have
\begin{align*}
    \norm{\hat{H}_i - H_i} &\leq  \norm{\hat{H}_i - \tilde{H}_i} +  \norm{\tilde{H}_i - H_i}\\
    &\leq  \abs{\frac{1}{1+\epsilon_1}-1} \norm{\tilde{H}_i} +  \epsilon_1 \\
    &\leq \epsilon_1(1+\epsilon_1) + \epsilon\\
    &=2 \epsilon_1 + \epsilon_1^2
\end{align*}
So finally,
\begin{align*}
    \norm{H - \tilde{H}} &= \norm{\sum_{i \in [m]]} \hat{p}_i \tilde{H}_i - \sum_{i \in [m]} p_i H_{i}}\\
    &\leq \sum_{i \in [m]}\norm{\hat{p}_i \tilde{H}_i - p_i H_{i}}\\
    &\leq \sum_{i \in [m]} \norm{\hat{p}_i \left(\tilde{H}_i - {H}_{i}\right) + H_{i} \left(  \hat{p}_i -  p_i  \right)} \\
    &\leq \sum_{i \in [m]} \norm{\hat{p}_i \left(\tilde{H}_i - {H}_{i}\right)} + \norm{H_{i} \left(  \hat{p}_i -  p_i  \right)} \\
    &\leq \sum_{i \in [m]} \hat{p}_i  \norm{\tilde{H}_i - {H}_{i}} + \norm{H_{i}}  \abs{\hat{p}_i -  p_i  } \\
    &\leq  2 \epsilon_1 + \epsilon_1^2 + 4 m^2 \epsilon_0
\end{align*} 
which can be made $\leq \epsilon$ if for example
\begin{align*}
\epsilon_0 := \frac{\epsilon}{8m^2}, \quad \epsilon_1 := \frac{\epsilon}{6}.
\end{align*}
\end{proof}

\section{Proof of~\texorpdfstring{\cref{cor:fixed_H}}{Corollary 1}}
\label{app:proof_cor_fixed_H}
\fixedHcor*
\begin{proof}
    We have that $H_x$ (resp. $\tilde{H}_x$ are specified by the $4^q \binom{kp_2(n)}{q}$ complex numbers $\bra{\alpha} H_{x,i_1,\dots,i_q} \ket{\beta}$ (resp. $\bra{\alpha} \tilde{H}_{x,(i_1,\dots,i_q)} \ket{\beta}$). Note that $\abs{\bra{\alpha} H_{x,i_1,\dots,i_q} \ket{\beta}} \leq 1$ since 
\begin{align*}
\abs{\bra{\alpha} H_{x,i_1,\dots,i_q} \ket{\beta}} \leq \max_{\alpha,\beta} \abs{\bra{\alpha} H_{x,i_1,\dots,i_q} \ket{\beta}}  \leq \norm{H_{x,i_1,\dots,i_q}} \leq 1. 
\end{align*}
Therefore, we can adopt the following binary notation to specify the values of $\bra{\alpha} H_{x,i_1,\dots,i_q} \ket{\beta} \in \mathbb{C}$: we use the most-significant bit to indicate whether it is the real or complex part, all remaining bits to specify a value in $[-1,1]$ in evenly spaced intervals. We have that
\begin{gather*}
    \text{ $  \text{Re} ( \bra{\alpha} \tilde{H}_{x,(i_1,\dots,i_q)} \ket{\beta} )$ is correct up to $\eta$ bits} \\ \Updownarrow  \\  \abs{ \text{Re}(\bra{\alpha} \tilde{H}_{x,(i_1,\dots,i_q)} \ket{\beta})-\text{Re}(\bra{\alpha} \tilde{H}_{x,(i_1,\dots,i_q)} \ket{\beta}) } \leq \frac{2}{\left(2^\eta+1\right)}.
\end{gather*}
The same argument holds for the imaginary part of $\bra{\alpha} \tilde{H}_{x,(i_1,\dots,i_q)} \ket{\beta}$. By the triangle inequality, we have that in this case
\begin{align*}
    \abs{ \bra{\alpha} \tilde{H}_{x,(i_1,\dots,i_q)} \ket{\beta}-\bra{\alpha} \tilde{H}_{x,(i_1,\dots,i_q)} \ket{\beta} } \leq \frac{4}{2^\eta +1 } \coloneqq  \epsilon'.
\end{align*}
To achieve
\begin{align*}
    \frac{4}{2^\eta +1 } \leq \frac{\epsilon}{|\Omega| 2^{q+4} q! },
\end{align*}
it suffices to set 
\begin{align*}
    \eta \coloneqq  \lceil \log \left( \frac{4 |\Omega| 2^{q+4} q!}{\epsilon} - 1 \right) \rceil.
\end{align*}
Since the Hadamard test can learn $\eta$ bits of precision in $\mO(2^\eta)$ time, we have that the time complexity of the reduction in~\cref{alg:reduction_QPCP} still runs in time $\poly(n,1/\epsilon,\log(1/\delta))$ as $\eta = \mO(\log(  \poly(n,1/\epsilon,\log(1/\delta)) )$.
\end{proof}

\section{A classical oracle separation}
\label{app:classical_or_sep}
Here we show that the same idea behind the quantum oracle separation of~\cref{sec:oracle_seps} can be used to give a classical oracle separation. We first prove a lower bound analogous to the one in~\cref{lem:lb_qproof} but now for ``classical search''.
\begin{lemma} Let $q \in \mathbb{N}$ and let $D_\epsilon^{2^q}$ be a $\epsilon$-covering set of mixed states on $q$ qubits as per~\cref{prop:dens_mat_set}. Suppose we are given oracle access to an $n$-qubit phase oracle $O_f$, and want to decide which of the following holds:
\begin{enumerate}[label=(\roman*)]
    \item There exists an $n$-bit string $x^*$ such that $\mathcal{O}_f \ket{x^*}  = -\ket{x^*} $
    \item $\mathcal{O}_f \ket{x}  = \ket{x} $ for all $x$.
\end{enumerate}
Then even if we have a witness tuple $(i_1,\dots,i_q,\rho)$, where $i_{i_1},\dots,i_{i_q}$ is a set of the $q$ indices and a quantum witness $\rho \in D_\epsilon^k$ in support of case (i), we still need 
\begin{align*}
 \Omega \left( \sqrt{\frac{2^n}{\left(\frac{e p(n)}{q}\right)^q C \left( \frac{1}{\epsilon}\right)^{5 \cdot 2^q}} } \right).
\end{align*}
queries to verify the witness, with bounded probability of error. Here $C >0 $ is some universal constant.
\label{lem:lb_qproof_classical}
\end{lemma}
\begin{proof} 
Let $\Omega = \{0,1\}^n$, and let $\sigma$ be any $p$-uniform distribution over $\Omega$. Then we have that for all $\rho \in \mathcal{D}(\mathcal{H})$
\begin{align*}
    \max_{\sigma: p\,\text{-uniform}}\mathbb{E}_{\ket{x} \in \sigma} \left[\bra{x} \rho \ket{x}\right] &\leq \max_{\sigma: p\,\text{-uniform}}\max_{z \in \Omega} \mathbb{E}_{\ket{x} \in \sigma} \left[\bra{x} \ket{z} \bra{z} \ket{x}\right] \\
    &=\max_{\sigma: p\,\text{-uniform}}  \mathbb{E}_{\ket{x} \in \sigma} \left[\bra{x} \ket{0} \bra{0} \ket{x}\right]\\
    &=\max_{\sigma: p\,\text{-uniform}}\mathbb{E}_{\ket{x} \in \sigma} \left[ \abs{\bra{0}\ket{x}}^2\right],
    \end{align*}
by the fact that the expected fidelity is maximised for $\rho$'s that are basis states and then using the maximisation over $p$-uniform $\rho$ to choose the maximal $z$ to be $0$. Clearly, this is maximised by conditioning on any $\log_2 (1/p)$ of bits being in $0$, which happens with probability $(1/2)^{\log_2 (1/p)} = p$. Hence, for any $\sigma$ we have
\begin{align*}
    \mathbb{E}_{\ket{x} \in \sigma} \left[ \abs{\bra{0}\ket{x}}^2\right] & \leq \frac{2^{\log_2( \frac{1}{p})}}{2^n} = \frac{1/p}{2^n},
    \end{align*}
and thus
\begin{align*}
    \mathbb{E}_{\ket{x} \in \sigma} \left[\bra{x} \rho \ket{x}\right] &\leq  \frac{1/p}{2^n},
    \end{align*}
for all mixed states $\rho$ and all $p$-uniform distributions $\sigma$. Using the value of $p$ as in~\cref{lem:lb_qproof} , we obtain the lower bound of 
\begin{align*}
    T \geq \left( \sqrt{\frac{2^n}{\left(\frac{e p(n)}{q}\right)^q C \left( \frac{1}{\epsilon}\right)^{5 \cdot 2^q}} } \right).
\end{align*}
\end{proof}
Given~\cref{lem:lb_qproof_classical}, the following corollary follows directly from the same proof of~\cref{thm:q_oracle_sep}.
\classorsep*

\end{document}